\documentclass[a4paper]{article}

\usepackage{a4wide}
\usepackage{dsfont}
\usepackage{amsmath}
\usepackage{amsfonts}
\usepackage{mathrsfs}
\usepackage{amssymb}
\usepackage{amsthm}
\usepackage{color}
\usepackage{bm}
\newtheorem{thm}{Theorem}
\newtheorem{prop}{Proposition}
\newtheorem{remark}{Remark}

\newtheorem{lemma}{Lemma}

\newcommand{\bbr}{{\mathbb R}}
\newcommand{\bbs}{{\mathbb S}}

\usepackage{authblk}

\begin{document}

\title{Small solutions of the Einstein-Boltzmann-scalar field system with Bianchi symmetry}

\author[1]{Ho Lee\footnote{holee@khu.ac.kr}}
\author[2]{Jiho Lee\footnote{leejiho72@khu.ac.kr}}
\author[3]{Ernesto Nungesser\footnote{em.nungesser@upm.es}}

\affil[1]{Department of Mathematics and Research Institute for Basic Science, Kyung Hee University, Seoul, 02447, Republic of Korea}
\affil[2]{Department of Mathematics, Kyung Hee University, Seoul, 02447, Republic of Korea}
\affil[3]{M2ASAI, Universidad Polit\'{e}cnica de Madrid, ETSI Navales, Avda. de la Memoria, 4, 28040 Madrid, Spain}

\maketitle

\begin{abstract}
We show that small homogeneous solutions to the Einstein-Boltzmann-scalar field system exist globally towards the future and tend to the de Sitter solution in a suitable sense. More specifically, we assume that the spacetime is of Bianchi type I--VIII, that the matter is described by Israel particles and that there exists a scalar field with a potential which has a positive lower bound. This represents a generalization of the work \cite{LN4}, where a cosmological constant was considered, and a generalization of \cite{LL22}, where a spatially flat FLRW spacetime was considered. We obtain the global existence and asymptotic behavior of classical solutions to the Einstein-Boltzmann-scalar field system for small initial data. 
\end{abstract}

\section{Introduction}
In the very first paper \cite{Einstein} about cosmology after having established the theory of general relativity, Albert Einstein introduced the cosmological constant extending his theory to obtain a static solution describing the universe. Considering an isotropic spacetime and matter such that the resulting pressure and energy density is positive, there is no possibility for a static universe unless one introduces a cosmological constant. However, Einstein's solution has two problems. Hubble \cite{Hubble}, a few years later, discovered a redshift in the light rays coming from different galaxies, thus indicating a possible expansion of the universe making a static universe unrealistic. In fact, before the observation of the expansion of the universe, already in 1922, i.e., exactly 100 years ago, Alexander Friedmann \cite{Friedmann} found that there are other non-static, non-vacuum solutions. The second problem is that Einstein's solution, which corresponds in modern terminology to a closed FLRW universe ($k=1$) having a critical radius $a_c$ filled with dust having a critical density $\rho_c$ and a cosmological constant, is unstable. For a long time, the interest in the cosmological constant diminished considerably, and the most popular model became the closed FLRW universe ($k=1$) without cosmological constant. It was clear in any case that the value of the cosmological constant must be small, since we do not detect it in the solar system.

Later again the cosmological constant has become very popular for two different reasons. First, the theory of inflation \cite{Guth} has been able to explain different phenomena of the very early universe successfully, and observations of the microwave background have been consistent with that theory. On the other hand, different observations of supernova \cite{Riess} seem to indicate that the universe is in an accelerated expansion {\it now}, and a cosmological constant is the simplest way to model this. The cosmological constant can be understood as the effect of the vacuum on spacetime curvature, for in the vacuum case the Einstein tensor would have to be equal to minus the cosmological constant times the metric. It is natural to model the vacuum by a quantum field, and one can for simplicity consider just a (non-linear) scalar field. Certainly, the discovery of the Higgs field has increased the interest in cosmological models with a scalar field.

The interest in the cosmological constant, or its generalization the scalar field, comes also from the fact that it helps to explain isotropization of the universe. Consider any cosmological model which is future complete and generic within the class of initial data in consideration with a cosmological constant. All the mathematical results available until now indicate that the model tends to the de Sitter solution, i.e., it isotropizes. The first mathematical result in that direction was obtained by Wald \cite{Wald83} obtaining the asymptotic behavior for homogeneous spacetimes with a positive cosmological constant, where the matter model satisfies certain energy conditions. This was then generalized to the scalar field cases by Kitada and Maeda \cite{KM92, KM93} and Rendall \cite{Rendall04}. Considering collisionless matter and homogeneous spacetimes, Hayoung Lee was able to obtain precise decay rates in the cosmological constant case \cite{Lee04} and then in the scalar field case \cite{Lee05}. Future non-linear stability for the Einstein-Vlasov system coupled with a scalar field was shown in the monograph by Ringstr{\"o}m \cite{Ringstrom} based on previous work by him \cite{Ringstrom08, Ringstrom09}. Concerning the Einstein-Boltzmann case, the global existence was obtained by Noutchegueme and Takou \cite{NT} in the isotropic case with certain assumptions on the scattering kernel. Finally, more recently considering the scattering kernel for Israel particles, the asymptotic behavior in the isotropic case was obtained in Refs.~\cite{LL22, LN172}, and the isotropization in the cosmological constant case was proved in Ref.~\cite{LN4}.

In this paper, we study the Cauchy problem for the Einstein-Boltzmann-scalar field system with Bianchi symmetry. Suitable assumptions on the scalar field will be made. For the Boltzmann equation, we will consider the scattering kernel for Israel particles. This is a simple scattering kernel, first studied by Israel in Ref.~\cite{Israel63}, which has a specific, mathematically tractable form so that the collision operator can be easily estimated. In this paper, we show that small solutions to the Einstein-Boltzmann-scalar field system with Bianchi symmetry exist globally towards the future and tend to the de Sitter solution. We extend the work in Refs.~\cite{LN172, LN4}, where a cosmological constant was considered, to the scalar field case, and the work in Ref.~\cite{LL22}, where the FLRW case was treated, to the Bianchi case. Moreover, we generalize the work in Ref.~\cite{Lee05}, where the Vlasov case was studied, to the Boltzmann case.

The structure of this paper is as follows. In Section \ref{sec Bianchi}, we briefly introduce the spacetime of Bianchi types. The main equations of this paper, the Einstein-Boltzmann-scalar field system with Bianchi symmetry, will be introduced here. In Section \ref{sec main}, we first make the assumptions on the potential, the scattering kernel, etc., and then state the main result in Theorem \ref{thm main}. Note that we need the smallness of initial data, especially for the Boltzmann equation, which is not necessary in the Vlasov case, for instance see Ref.~\cite{Lee05}. In Section \ref{sec EBs}, we study the local existence and uniqueness for the Einstein-Boltzmann-scalar field system. In Section \ref{sec Boltzmann}, we first consider the local existence for the single Boltzmann equation. The argument is basically the same as in Refs.~\cite{LN172, LN4}, but in this paper we present a more detailed proof, for instance the higher order derivatives of post-collision momenta are estimated in Lemma \ref{lem dp'} in a mathematically rigorous way. In Section \ref{sec local}, we study the coupled Einstein-Boltzmann-scalar field system. The local existence is proved in Section \ref{sec existence}, and the uniqueness is proved in Section \ref{sec uniqueness}. It turns out that the uniqueness should be proved with great care in the case that the Boltzmann equation is coupled with Einstein's equations so that the distribution function needs to be estimated with a norm depending on the metric. This will be proved in Section \ref{sec uniqueness} in detail. In Section \ref{sec Global}, we assume that initial data is small and obtain the global existence and asymptotic behavior.

\subsection{Bianchi spacetimes}\label{sec Bianchi}

We consider a Lorentz manifold $ M = I \times G $, where $ I $ is a time interval and $ G $ is a 3-dimensional Lie group. Let $ { \bf E }_a $, $ a = 1 , 2 , 3 $, be a left-invariant frame on $ G $ with $ { \bf W }^a $ its dual frame. We take $ { \bf E }_0 = \partial_t $ to be future oriented and suppose that a metric $ ^4 g $ on $ M $ is given by
\[
^4 g = - d t^2 + g , \qquad g = g_{ a b } ( t ) { \bf W }^a { \bf W }^b,
\]
where $ g $ is a Riemannian metric on $ G $, which is left-invariant. Throughout the paper the Einstein summation convention will be used, where Latin indices run from 1 to 3 and Greek indices run from 0 to 3. In this paper we are interested in the following system of equations. We first consider the evolution equations for the metric $ g_{ a b } $ and the second fundamental form $ k_{ a b } $:
\begin{align}
\frac{ d g_{ a b } }{ d t } & = 2 k_{ a b } , \label{EBs1} \\
\frac{ d k_{ a b } }{ d t } & = 2 k^c_a k_{ b c } - k k_{ a b } - R_{ a b } + S_{ a b } + \frac12 ( \rho - S ) g_{ a b } + V ( \phi ) g_{ a b } ,  \label{EBs2}
\end{align}
where $ k $ is the trace of $ k_{ a b } $ and $ R_{ a b } $ is the Ricci curvature of $ g_{ a b } $. The quantities $ S_{ a b } $, $ \rho $ and $ S $ are defined by $ S_{ a b } = T_{ a b } $, $ \rho = T_{ 0 0 } $ and $ S = S_{ a b } g^{ a b } $, where $ T_{ \alpha \beta } $ is the energy-momentum tensor associated with the Boltzmann matter, and $ V ( \phi ) $ is the potential for the scalar field. We note that the structure constants are defined by 
\begin{align*}
[ { \bf E }_a , { \bf E }_b ] = C^c_{ a b } { \bf E }_c , \qquad [ { \bf E }_0 , { \bf E }_a ] = 0 ,
\end{align*}
where we choose $ { \bf E }_0 = \partial_t $. We define the connection coefficients by $ \nabla_{ { \bf E }_\alpha } { \bf E }_\beta = \Gamma^\gamma_{ \alpha \beta } { \bf E }_\gamma $ and use the Koszul formula to conclude that the only nontrivial coefficients are 
\begin{align*}
& \Gamma_{ c a b } = \frac12 \left( - g_{ a d } C^d_{ b c } + g_{ b d } C^d_{ c a } + g_{ c d } C^d_{ a b } \right) , \\
& \Gamma_{ 0 a b } = - k_{ a b } , \qquad \Gamma_{ c 0 b } = k_{ b c } , \qquad \Gamma_{ c a 0 } = k_{ c a } . 
\end{align*}
In this case, the Riemann curvature and the Ricci curvature of $ g_{ a b } $ can be written by 
\begin{align*}
R^a_{bcd} & = \Gamma^e_{db}\Gamma^a_{ce} - \Gamma^e_{cb}\Gamma^a_{de} - C^e_{cd}\Gamma^a_{eb} , \qquad R_{ b d } = R^a_{ b a d } . 
\end{align*}
Hence, we observe that the Ricci curvature $ R_{ a b } $ is a polynomial of $ g_{ a b } $ and $ g^{ a b } $. The equations for the scalar field $ \phi $ are given by
\begin{align}
\frac{ d \phi }{ d t } & = \psi , \label{EBs3} \\
\frac{ d \psi }{ d t } & = - k \psi - V ' ( \phi ) , \label{EBs4}
\end{align}
and the assumptions on $ V $ will be given in Section \ref{sec main}. The Boltzmann equation can be written in a simple form in the Bianchi cases by using the structure constants:
\begin{align}
\frac{ \partial f }{ \partial t } + \frac{ 1 }{ p^0 } C^d_{ b c } p_d p^b \frac{ \partial f }{ \partial p_c } = Q ( f , f ) , \label{EBs5}
\end{align}
where the left hand side is derived from the equation (25.20) of Ref.~\cite{Ringstrom} by applying the above representation of the connection coefficients. The collision operator $ Q ( f , f ) $ is defined by
\begin{align*}
Q ( f , f ) = ( \det g )^{ - \frac12 } \int_{ \bbr^3 } \int_{ \bbs^2 } \frac{ h \sqrt{ s } }{ p^0 q^0 } \sigma ( h , \omega ) ( f ( p' ) f ( q' ) - f ( p ) f ( q ) ) \, d \omega \, d q , 
\end{align*}
which will be studied in Section \ref{sec Boltzmann} in detail. Finally, the constraint equations are given by
\begin{align}
R - k_{ a b } k^{ a b } + k^2 & = 2 \rho + \psi^2 + 2 V ( \phi ) , \label{EBs6} \\
\nabla^a k_{ a b } & = - T_{ 0 b } , \label{EBs7}
\end{align}
where $ R $ is the Ricci scalar of $ g_{ a b } $. The system of equations \eqref{EBs1}--\eqref{EBs7} will be referred to as the Einstein-Boltzmann-scalar field (hereafter EBs) system in this paper. The assumptions on the potential $ V $ for the scalar field and the scattering kernel $ \sigma $ for the Boltzmann equation will be given in Section \ref{sec main}, and the energy-momentum tensor $ T_{ \alpha \beta } $ associated with the Boltzmann equation will be given in Section \ref{sec Boltzmann}.

\subsection{Main result} \label{sec main}
We state the main result of this paper. The following will be assumed throughout the paper. 

\begin{itemize}

\item We assume that the spacetime is of Bianchi types I--VIII. Hence, the Ricci scalar $ R $ of $ g_{ a b } $ will be non-positive.

\item We assume that the potential $ V = V ( \phi ) $ is a smooth function and that there exists a positive constant $ V_0 $ such that 
\[
V ( \phi ) \geq V_0 > 0 . 
\]
Moreover, we assume that the minimum $ V_0 $ is attained at $ \phi = 0 $ and is non-degenerate: 
\begin{align*}
V ( 0 ) = V_0 > 0 , \qquad V ' ( 0 ) = 0 , \qquad V '' ( 0 ) > 0 . 
\end{align*}

\item We assume that the scattering kernel $ \sigma ( h , \omega ) $ for the Boltzmann equation is the one for Israel particles. We further assume that it is independent of the angular variable $ \omega $, i.e., 
\[
\sigma ( h , \omega ) = \frac{ 1 }{ h s } , 
\]
where $ h $ and $ s $ are defined in Section \ref{sec Boltzmann}. 

\item The unknowns of the EBs system will be the metric $ g^{ a b } $, the second fundamental form $ k^{ a b } $, the scalar fields $ \phi $ and $ \psi $, and the distribution function $ f $. Initial value of these quantities will be given at some $ t_0 > 0 $ and will be denoted by 
\[
{ g_0 }^{ a b } = g^{ a b } ( t_0 ) , \qquad { k_0 }^{ a b } = k^{ a b } ( t_0 ) , \qquad \phi_0 = \phi ( t_0 ) , \qquad \psi_0 = \psi ( t_0 ) , \qquad f_0 ( p ) = f ( t_0 , p ) . 
\]

\item The Hubble variable $ H = H ( t ) $ will be defined by 
\[
H = \frac13 k , \qquad k = k_{ a b } g^{ a b } , \qquad H_0 = H ( t_0 ) , 
\]
where $ H_0 $ denotes the initial value of the Hubble variable. 

\item We introduce a scale factor $ Z = Z ( t ) $ defined by 
\begin{align}
Z = e^{ \gamma t } , \qquad \gamma = \sqrt{ \frac{ V_0 }{ 3 } } , \label{scale factor}
\end{align}
where $ V_0 $ denotes the positive lower bound of $ V $.

\end{itemize}

For an introduction to the Bianchi symmetry and the scalar field we refer to Refs.~\cite{Rendall04, Ringstrom}. For more details about the Israel particles we refer to Ref.~\cite{Israel63}. The following is the main result of this paper.

\begin{thm} \label{thm main}
Let $ { g_0 }^{ a b } $, $ { k_0 }^{ a b } $, $ \phi_0 $, $ \psi_0 $ and $ f_0 $ be a set of initial data of the EBs system \eqref{EBs1}--\eqref{EBs7}. Suppose that $ { g_0 }^{ a b } $ and $ { k_0 }^{ a b } $ are positive definite and $ H_0 > \gamma $. Then, there exist positive numbers $ \varepsilon $, $ C_0 $, $ b_1 $, $ b_2 $ and $ b_3 $ such that if initial data satisfy 
\[
H_0 - \gamma + E_0 < \varepsilon , \qquad \max_{ a , b } | { g_0 }^{ a b } | \leq C_0 , \qquad \| f_0 \|^2_{ g_0 , m + \frac12 , N } < \varepsilon , 
\]
with $ m > 7 / 2 $ and $ N \geq 3 $, then the EBs system admits a unique classical solution $ g^{ a b } $, $ k^{ a b } $, $ \phi $, $ \psi $ and $ f $ on $ [ t_0 , \infty ) $ such that $ f $ is non-negative and 
\begin{align*}
\sup_{ t \in [ t_0 , \infty ) } \max_{ a , b } | g^{ a b } ( t ) | & \leq C_0 + C \sqrt{ \varepsilon } , \\ 
\sup_{ t \in [ t_0 , \infty ) } \max_{ a , b } | k^{ a b } ( t ) | & \leq C_0 \gamma + C \sqrt{ \varepsilon } , \\ 
\sup_{ t \in [ t_0 , \infty ) } \| f ( t ) \|^2_{ g , m , N } & \leq C \varepsilon , 
\end{align*}
with the following asymptotics: 
\begin{align*}
| \phi ( t ) | + | \psi ( t ) | & \leq C \sqrt{ \varepsilon } e^{ - \frac12 b_1 t } , \\ 
\rho ( t ) + | R ( t ) | + \sigma_{ a b } \sigma^{ a b } ( t ) & \leq C \varepsilon e^{ - 2 \gamma t } , \\ 
0 \leq H ( t ) - \gamma & \leq C \varepsilon e^{ - b_2 t } , \\ 
\det g^{ - 1 } & \leq C e^{ - 6 \gamma t } , \\ 
| Z^2 ( t ) g^{ a b } ( t ) - { g_\infty }^{ a b } | & \leq C \sqrt{ \varepsilon } e^{ - b_3 t } , \\ 
| Z^{ - 2 } ( t ) g_{ a b } ( t ) - { g_\infty }_{ a b } | & \leq C \sqrt{ \varepsilon } e^{ - b_3 t } , 
\end{align*}
for some constant metrics $ { g_\infty }^{ a b } $ and $ { g_\infty }_{ a b } $. 
\end{thm}

Here, the quantities $ H $ and $ H_0 $ denote the Hubble variable and the initial value of it, and $ E_0 $ is the initial value of $ E $, which is a quantity equivalent to $ \phi^2 + \psi^2 $. The quantity $ \sigma_{ a b } $ denotes the shear tensor, defined by the trace free part of the second fundamental form $ k_{ a b } $. These quantities will be defined in Section \ref{sec Global}.

\section{The Einstein-Boltzmann-scalar field system}\label{sec EBs}
In this part we study the local existence of solutions to the EBs system. Let us first consider the equations \eqref{EBs1}--\eqref{EBs4}, which are simple ODEs for $ g_{ a b } $, $ k_{ a b } $, $ \phi $ and $ \psi $, where the right hand sides are basically given by polynomials of $ g_{ a b } $, $ g^{ a b } $, $ k_{ a b } $, $ \phi $ and $ \psi $ together with matter terms induced by the (particle) distribution function $ f $. Hence, one can easily obtain the local existence for $ g_{ a b } $, $ k_{ a b } $, $ \phi $ and $ \psi $, provided that the distribution function is properly estimated. Below, we first consider the local existence for the Boltzmann equation. The local existence for the coupled EBs system will be studied in Section \ref{sec local}.

\subsection{The Boltzmann equation}\label{sec Boltzmann}
In the Bianchi case, by using the left invariant frame $ { \bf E }_a $, the distribution function can be written as a function of time $ t $ and momentum $ p $. It is well known that the Boltzmann equation can be written in a simple form if we use covariant variables:
\[
f = f ( t , p ) , \qquad p = ( p_1 , p_2 , p_3 ) \in \bbr^3.
\]
Then, the Boltzmann equation is written as 
\begin{align}
\frac{ \partial f }{ \partial t } + \frac{ 1 }{ p^0 } C^d_{ b c } p_d p^b \frac{ \partial f }{ \partial p_c } = Q ( f , f ) , \label{Boltzmann}
\end{align}
where $ C^d_{ b c } $ are the structure constants. The collision operator $ Q ( f , f ) $ can be written as 
\begin{align*}
Q ( f , f ) & = Q_+ ( f , f ) - Q_- ( f , f ) , 
\end{align*}
where $ Q_+ ( f , f ) $ and $ Q_- ( f , f ) $ are called the gain and loss terms, respectively, and defined by 
\begin{align*}
Q_+ ( f , f ) & = ( \det g )^{ - \frac12 } \int_{ \bbr^3 } \int_{ \bbs^2 } \frac{ 1 }{ p^0 q^0 \sqrt{ s } } f ( p' ) f ( q' ) \, d \omega \, d q , \\
Q_- ( f , f ) & = ( \det g )^{ - \frac12 } \int_{ \bbr^3 } \int_{ \bbs^2 } \frac{ 1 }{ p^0 q^0 \sqrt{ s } } f ( p ) f ( q ) \, d \omega \, d q . 
\end{align*}
Here, $ d q = d q_1 \, d q_2 \, d q_3 $, and $ d \omega $ is the usual surface element on $ \bbs^2 $, where $ \omega \in \bbs^2 $ is the unit vector in the sense that $ \delta^{ i j } \omega_i \omega_j = 1 $. Throughout the paper, the quantity $ \det g $ will denote the determinant of the $ 3 \times 3 $ matrix $ g_{ a b } $, while $ \det g^{ - 1 } $ will denote the determinant of the $ 3 \times 3 $ matrix $ g^{ a b } $. The energy-momentum tensor $ T_{ \alpha \beta } $ is defined by
\begin{align}
T_{ \alpha \beta } = ( \det g )^{ - \frac12 } \int_{ \bbr^3 } f p_\alpha p_\beta \frac{ d p }{ p^0 } , \label{T_ab} 
\end{align}
where $ d p = d p_1 \, d p_2 \, d p_3 $. The mass shell condition implies that
\[
p^0 = \sqrt{ 1 + g^{ a b } p_a p_b } , \qquad q^0 = \sqrt{ 1 + g^{ a b } q_a q_b } ,
\]
and the quantities $ h $ and $ s $ are the relative momentum and total energy defined by
\[
h = \sqrt{ g^{ \alpha \beta } ( p_\alpha - q_\alpha ) ( p_\beta - q_\beta ) } , \qquad s = - g^{ \alpha \beta } ( p_\alpha + q_\alpha ) ( p_\beta + q_\beta ) .
\]
In order to have explicit representations of post-collision momenta we need to introduce an orthonormal frame $ { \bf e }_i $, which can be written with respect to the left invariant frame as
\begin{align*}
{ \bf e }_i = { e_i }^a { \bf E }_a , \qquad { \bf e }^i = { e^i }_a { \bf W }^a ,
\end{align*}
where $ { e^i }_a $ is the inverse of $ { e_i }^a $ as a $ 3 \times 3 $ matrix and satisfies $ \delta^{ i j } = { e^i }_a { e^j }_b g^{ a b } $. It will be convenient to write
\[
n_\alpha = p_\alpha + q_\alpha. 
\]
Then, the post-collision momentum is written as
\begin{align}
p'^0 & = p^0 + 2 \bigg( - q^0 \frac{ n^a \omega_i { e^i }_a }{ \sqrt{ s } } + q^a \omega_i { e^i }_a + \frac{ n^a \omega_i { e^i }_a n_b q^b }{ \sqrt{ s } ( n^0 + \sqrt{ s } ) } \bigg) \frac{ n^c \omega_j { e^j }_c }{ \sqrt{ s } } , \label{p'^0} \\
q'^0 & = q^0 - 2 \bigg( - q^0 \frac{ n^a \omega_i { e^i }_a }{ \sqrt{ s } } + q^a \omega_i { e^i }_a + \frac{ n^a \omega_i { e^i }_a n_b q^b }{ \sqrt{ s } ( n^0 + \sqrt{ s } ) } \bigg) \frac{ n^c \omega_j { e^j }_c }{ \sqrt{ s } } , \label{q'^0} 
\end{align}
and
\begin{align}
p'_d & = p_d + 2 \bigg( - q^0 \frac{ n^a \omega_i { e^i }_a }{ \sqrt{ s } } + q^a \omega_i { e^i }_a + \frac{ n^a \omega_i { e^i }_a n_b q^b }{ \sqrt{ s } ( n^0 + \sqrt{ s } ) } \bigg) \bigg( \omega_j { e^j }_d + \frac{ n^c \omega_j { e^j }_c n_d }{ \sqrt{ s } ( n^0 + \sqrt{ s } ) } \bigg) , \label{p'} \\
q'_d & = q_d - 2 \bigg( - q^0 \frac{ n^a \omega_i { e^i }_a }{ \sqrt{ s } } + q^a \omega_i { e^i }_a + \frac{ n^a \omega_i { e^i }_a n_b q^b }{ \sqrt{ s } ( n^0 + \sqrt{ s } ) } \bigg) \bigg( \omega_j { e^j }_d + \frac{ n^c \omega_j { e^j }_c n_d }{ \sqrt{ s } ( n^0 + \sqrt{ s } ) } \bigg) , \label{q'} 
\end{align}
where $ \omega = ( \omega_1 , \omega_2 , \omega_3 ) \in \bbs^2 $ in the sense that $ \delta^{ i j } \omega_i \omega_j = 1 $. Note that the pre- and post-collision momenta satisfy the energy-momentum conservation:
\begin{align}\label{energy-momentum}
p'_\alpha + q'_\alpha = p_\alpha + q_\alpha. 
\end{align}
We also note that the above representation reduces to the one in the Minkowski case when written with respect to the orthonormal basis. For instance, the quantity $ n^a \omega_i { e^i }_a $ in the above representation of $ p'^0 $ can be written as
\[
n^a \omega_i { e^i }_a = g^{ a b } n_b \omega_i { e^i }_a = g^{ a b } { \hat n }_j { e^j }_b \omega_i { e^i }_a = \delta^{ i j } { \hat n }_j \omega_i , 
\]
where $ n_b $ and $ { \hat n }_j $ denote the components of $ n $ with respect to $ { \bf W }^b $ and $ { \bf e }^j $, respectively, so that the last quantity on the right side above is $ { \hat n } \cdot \omega $, which is the Euclidean inner product of $ { \hat n } $ and $ \omega $. Since the quantities $ p^0 $, $ q^0 $, $ n^0 $ and $ s $ are independent of the choice of basis, one can see that the above representation of $ p'^0 $ is the same as the one in the Minkowski case, which can be found in (36) of Ref.~\cite{LN172}. For more details about the Boltzmann equation, we refer to Refs.~\cite{CIP, CK, Glassey}.

\subsubsection{Assumptions on the metric} 
In Proposition \ref{prop Boltzmann}, we will show that the Boltzmann equation \eqref{Boltzmann} admits a local solution in a given spacetime. To establish the existence result for the Boltzmann equation in a given spacetime we need to make suitable assumptions on the metric. To be precise, we make the following assumptions.

\medskip

\noindent{\bf Assumptions on the metric.} {\it We assume that there exist a time interval $ [ t_0 , T ] $ and a constant $ c_1 \geq 1 $ such that a metric $ g^{ a b } $ and a second fundamental form $ k^{ a b } $ exist on $ [ t_0 , T ] $ and satisfy 
\begin{align}
\frac{ 1 }{ c_1 } | p |^2 \leq Z^2 g^{ a b } p_a p_b \leq c_1 | p |^2, \qquad \frac{ 1 }{ c_1 } | p |^2 \leq Z^2 k^{ a b } p_a p_b \leq c_1 | p |^2 , \qquad \max_{ a , b } | Z^2 g^{ a b } | \leq c_1 , \label{assumption1}
\end{align}
where $ | p | = \sqrt{ \delta^{ a b } p_a p_b } $ for $ p \in \bbr^3 $. }

\subsubsection{Basic estimates}

We define the following norm for the distribution function:
\begin{align}\label{norm1}
\| f ( t ) \|^2_{ g , m , N } = \sum_{ | { \mathcal I } | \leq N } \int_{ \bbr^3 } | \partial^{ \mathcal I }_p f ( t , p ) |^2 \langle p \rangle^{ 2 m } e^{ p^0 } \, d p, 
\end{align}
where $ { \mathcal I } $ is a multi-index, and $ \langle p \rangle $ is defined by
\[
\langle p \rangle = \sqrt{ 1 + | p |^2 }, 
\]
which does not depend on the metric $ g $. We notice that the norm depends on the metric $ g $ due to the quantity $ e^{ p^0 } $, which is necessary to control the collision term. We will also need the following norm: 
\begin{align}\label{norm2}
\| f ( t ) \|^2_{ m , N } = \sum_{ | { \mathcal I } | \leq N } \int_{ \bbr^3 } | \partial^{ \mathcal I }_p f ( t , p ) |^2 \langle p \rangle^{ 2 m } \, d p, 
\end{align}
which does not depend on the metric.

Below, we will denote by $ { \mathfrak L }^{ ( n ) } $ a linear combination of the following form: 
\[
{ \mathfrak L }^{ ( n ) } = \sum_{ \substack{ 2 k + l + m = n \\ k , l , m \geq 0 } } C_{ a_1 b_1 a_2 \cdots d_m } g^{ a_1 b_1 } g^{ a_2 b_2 } \cdots g^{ a_k b_k } \frac{ p^{ c_1 } }{ p^0 } \frac{ p^{ c_2 } }{ p^0 } \cdots \frac{ p^{ c_l } }{ p^0 } \frac{ q^{ d_1 } }{ q^0 } \frac{ q^{ d_2 } }{ q^0 } \cdots \frac{ q^{ d_m } }{ q^0 } , 
\]
where $ C_{ a_1 b_1 a_2 \cdots d_m } $ are constants. For instance, $ { \mathfrak L }^{ ( 1 ) } $ denotes any linear combination of $ p^a / p^0 $ and $ q^b / q^0 $; $ { \mathfrak L }^{ ( 2 ) } $ is any linear combination of $ g^{ a b } $, $ p^a p^b / ( p^0 )^2 $, $ p^a q^b / ( p^0 q^0 ) $ and $ q^a q^b / ( q^0 )^2 $; and so on. We will also denote by $ { \mathfrak L }^{ ( n ) }_P $ a combination of the following form:
\[
{ \mathfrak L }^{ ( n ) }_P = \sum_{ \substack{ 2 k + l + m = n \\ k , l , m \geq 0 } } P_{ a_1 b_1 a_2 \cdots d_m } g^{ a_1 b_1 } g^{ a_2 b_2 } \cdots g^{ a_k b_k } \frac{ p^{ c_1 } }{ p^0 } \frac{ p^{ c_2 } }{ p^0 } \cdots \frac{ p^{ c_l } }{ p^0 } \frac{ q^{ d_1 } }{ q^0 } \frac{ q^{ d_2 } }{ q^0 } \cdots \frac{ q^{ d_m } }{ q^0 } ,
\]
where $ P_{ a_1 b_1 a_2 \cdots d_m } $ denote some polynomials of the following quantities:
\begin{align}
\frac{ 1 }{ p^0 } , \quad \frac{ 1 }{ q^0 } , \quad \frac{ 1 }{ \sqrt{ s } } , \quad \frac{ 1 }{ n^0 + \sqrt{ s } } . \label{Pcoeff} 
\end{align}
We notice that the above quantities are bounded by $ 1 $, $ 1 $, $ 1/2 $ and $ 1/4 $, respectively, so that $ P_{ a_1 b_1 a_2 \cdots d_m } $ are also bounded. Finally, we will denote by $ { \mathcal C } $ a finite sum of the following form:
\[
{ \mathcal C } = \sum_{ n = M }^N Z^n | { \mathfrak L }^{ ( n ) }_P | , 
\]
for any $ N \geq M \geq 0 $, where $ Z $ is defined by \eqref{scale factor}.

\begin{lemma}\label{lem L}
The following hold for any $ n \geq 1 $:
\begin{align}
& \partial_{ p_a } { \mathfrak L }^{ ( n ) } = \frac{ 1 }{ p^0 } { \mathfrak L }^{ ( n + 1 ) } , \label{dL} \\
& \partial_{ p_a } { \frak L }^{ ( n ) }_P = q^0 { \frak L }^{ ( n + 1 ) }_P . \label{dLP}
\end{align}
\end{lemma}
\begin{proof}
We first notice that $ { \mathfrak L }^{ ( m ) } { \mathfrak L }^{ ( n ) } = { \mathfrak L }^{ ( m + n ) } $ and $ { \mathfrak L }^{ ( m ) }_P { \mathfrak L }^{ ( n ) }_P = { \mathfrak L }^{ ( m + n ) }_P $. By a direct calculation, we have
\[
\partial_{ p_a } \left( \frac{ p^b }{ p^0 } \right) = \frac{ 1 }{ p^0 } \left( g^{ a b } - \frac{ p^a p^b }{ ( p^0 )^2 } \right) = \frac{ 1 }{ p^0 } { \mathfrak L}^{ ( 2 ) } , 
\]
and this proves \eqref{dL}. Recall that $ s = 2 + 2 p^0 q^0 - 2 p_a q^a $. By direct calculations, we have
\begin{align}
& \partial_{ p_a } \left( \frac{ 1 }{ p^0 } \right) = - \frac{ p^a }{ ( p^0 )^3 } = \frac{ 1 }{ ( p^0 )^2 } { \mathfrak L}^{ ( 1 ) } = \frac{ q^0 }{ ( p^0 )^2 q^0 } { \mathfrak L}^{ ( 1 ) } = q^0 { \mathfrak L }^{ ( 1 ) }_P , \label{d1/p} \\
& \partial_{ p_a } \left( \frac{ 1 }{ \sqrt{ s } } \right) = - \frac{ 1 }{ 2 s^{ \frac32 } } \left( 2 \frac{ p^a }{ p^0 } q^0 - 2 q^a \right) = \frac{ q^0 }{ s^{ \frac32 } } { \mathfrak L }^{ ( 1 ) } = q^0 { \mathfrak L }^{ ( 1 ) }_P , \label{d1/s} \\
& \partial_{ p_a } \left( \frac{ 1 }{ n^0 + \sqrt{ s } } \right) = - \frac{ 1 }{ ( n^0 + \sqrt{ s } )^2 } \left( \frac{ p^a }{ p^0 } + \frac{ 1 }{ 2 \sqrt{ s } } \left( 2 \frac{ p^a }{ p^0 } q^0 - 2 q^a \right) \right) = \frac{ q^0 }{ ( n^0 + \sqrt{ s } )^2 } { \frak L }^{ ( 1 ) }_P , \label{d1/ns} 
\end{align}
where the last one can also be written as $ q^0 { \mathfrak L }^{ ( 1 ) }_P $. The above three calculations show that for any polynomial $ P $ of the quantities in \eqref{Pcoeff}, we have
\[
\partial_{ p_a } P = q^0 { \mathfrak L }^{ ( 1 ) }_P . 
\]
We notice that the first result \eqref{dL} can be written as
\[
\partial_{ p_a } { \mathfrak L }^{ ( n ) } = \frac{ q^0 }{ p^0 q^0 } { \mathfrak L }^{ ( n + 1 ) } = q^0 { \mathfrak L }^{ ( n + 1 ) }_P .
\]
We now combine the above two calculations to obtain the desired result \eqref{dLP}. 
\end{proof}

\begin{remark}
Let $ { \mathcal I } $ be a multi-index with $ | { \mathcal I } | = n \geq 0 $ and $ P $ be a polynomial of the quantities in \eqref{Pcoeff}. Then, from the proof of Lemma \ref{lem L}, one can deduce that 
\begin{align}
\partial^{ \mathcal I }_p P = ( q^0 )^n { \mathfrak L }^{ ( n ) }_P , \label{dP} 
\end{align}
for some $ { \mathfrak L }^{ ( n ) }_P $. Moreover, the calculation \eqref{d1/p} shows that for any $ { \mathcal I } $ with $ | { \mathcal I } | = n \geq 0 $, 
\[
\partial^{ { \mathcal I } }_p \left( \frac{ 1 }{ p^0 } \right) = \frac{ 1 }{ ( p^0 )^{ n + 1 } } { \mathfrak L }^{ ( n ) } , 
\]
which implies that
\begin{align}
\partial^{ { \mathcal I } }_p \left( \frac{ 1 }{ p^0 q^0 } \right) = \frac{ 1 }{ p^0 q^0 } { \mathfrak L }^{ ( n ) }_P , \label{d1/pq}
\end{align}
for some $ { \mathfrak L }^{ ( n ) }_P$. 
\end{remark}

\begin{lemma}\label{lem p/p^0}
Let $ { \mathcal I } $ be a multi-index with $ | { \mathcal I } | = n \geq 0 $. Then, we have
\[
\partial_p^{ \mathcal I } \bigg( \frac{ p^a }{ p^0 } \bigg) = \frac{ 1 }{ ( p^0 )^n } { \mathfrak L }^{ ( n + 1 ) } ,
\]
for some $ { \mathfrak L }^{ ( n + 1 ) } $. 
\end{lemma}
\begin{proof}
By a direct calculation using \eqref{dL}, we have
\begin{align*}
\partial_{ p_b } \left( \frac{ 1 }{ ( p^0 )^n } { \mathfrak L }^{ ( n + 1 ) } \right) = \frac{ - n p^b }{ ( p^0 )^{ n + 2 } } { \mathfrak L }^{ ( n + 1 ) } + \frac{ 1 }{ ( p^0 )^{ n + 1 } } { \mathfrak L }^{ ( n + 2 ) } = \frac{ 1 }{ ( p^0 )^{ n + 1 } } { \mathfrak L}^{ ( n + 2 ) } , 
\end{align*}
for some different $ { \mathfrak L}^{ ( n + 2 ) } $.  We now obtain the desired result by an induction. 
\end{proof}

\begin{lemma} \label{lem 1/s}
Let $ { \mathcal I } $ be a multi-index with $ | { \mathcal I } | = n \geq 1 $. Then, we have
\[
\partial^{ \mathcal I }_p \left( \frac{ 1 }{ \sqrt{ s } } \right) = \frac{ ( q^0 )^n }{ s^{ \frac32 } } { \mathfrak L }^{ ( n ) }_P ,  
\]
for some $ { \mathfrak L }^{ ( n ) }_P $. 
\end{lemma}
\begin{proof}
Recall that $ s = 2 + 2 p^0 q^0 - 2 p_a q^a $. By a direct calculation, we have
\begin{align}
\partial_{ p_a } s = 2 \frac{ p^a }{ p^0 } q^0 - 2 q^a = q^0 { \mathfrak L }^{ ( 1 ) } , \label{ds} 
\end{align}
for some $ { \mathfrak L }^{ ( 1 ) } $. Hence, we obtain
\begin{align*}
\partial_{ p_a } \left( \frac{ 1 }{ \sqrt{ s } } \right) = \frac{ q^0 }{ s^{ \frac32 } } { \mathfrak L }^{ ( 1 ) } , 
\end{align*}
for some different $ { \mathfrak L }^{ ( 1 ) } $. This shows that the lemma holds for $ n = 1 $. Suppose that the lemma holds for $ | \mathcal I | = n \geq 1 $, and compute 
\begin{align*}
\partial_{ p_a } \left( \frac{ ( q^0 )^n }{ s^{ \frac32 } } { \mathfrak L }^{ ( n ) }_P \right) & = ( q^0 )^n \partial_{ p_a } \left( \frac{ 1 }{ s^{ \frac32 } } \right) { \mathfrak L }^{ ( n ) }_P + \frac{ ( q^0 )^n }{ s^{ \frac32 } } \partial_{ p_a } { \mathfrak L }^{ ( n ) }_P \\
& = ( q^0 )^n \left( - \frac{ 3 }{ 2 s^{ \frac52 } } q^0 { \mathfrak L }^{ ( 1 ) } \right) { \mathfrak L }^{ ( n ) }_P + \frac{ ( q^0 )^{ n + 1 } }{ s^{ \frac32 } } { \mathfrak L }^{ ( n + 1 ) }_P , 
\end{align*}
where we used \eqref{ds} and \eqref{dLP}. Since we can write
\[
\frac{ 1 }{ s^{ \frac52 } } { \mathfrak L }^{ ( 1 ) } = \frac{ 1 }{ s^{ \frac32 } } { \mathfrak L }^{ ( 1 ) }_P , 
\]
we obtain the desired result for $ | { \mathcal I } | = n + 1 $. This completes the proof of the lemma. 
\end{proof}

\begin{lemma} \label{lem 1/ns}
Let $ { \mathcal I } $ be a multi-index with $ | { \mathcal I } | = n \geq 1 $. Then, we have
\[
\partial^{ \mathcal I }_p \left( \frac{ 1 }{ n^0 + \sqrt{ s } } \right) = \frac{ ( q^0 )^n }{ ( n^0 + \sqrt{ s } )^2 } { \mathfrak L }^{ ( n ) }_P , 
\]
for some $ { \mathfrak L }^{ ( n ) }_P $. 
\end{lemma}
\begin{proof}
We have already showed that in \eqref{d1/ns} the lemma holds for $ n = 1 $:
\[
\partial_{ p_a } \left( \frac{ 1 }{ n^0 + \sqrt{ s } } \right) = \frac{ q^0 }{ ( n^0 + \sqrt{ s } )^2 } { \frak L }^{ ( 1 ) }_P . 
\]
Suppose that the lemma holds for $ n \geq 1 $. Then, we obtain
\begin{align*}
\partial_{ p_a } \left( \frac{ ( q^0 )^n }{ ( n^0 + \sqrt{ s } )^2 } { \mathfrak L }^{ ( n ) }_P \right) & = \partial_{ p_a } \left( \frac{ ( q^0 )^n }{ ( n^0 + \sqrt{ s } )^2 } \right) { \mathfrak L }^{ ( n ) }_P + \frac{ ( q^0 )^n }{ ( n^0 + \sqrt{ s } )^2 } \partial_{ p_a } { \mathfrak L }^{ ( n ) }_P \\
& = \frac{ - 2 ( q^0 )^n }{ ( n^0 + \sqrt{ s } )^3 } \left( \frac{ p^a }{ p^0 } + \frac{ 1 }{ 2 \sqrt{ s } } \left( 2 \frac{ p^a }{ p^0 } q^0 - 2 q^a \right) \right) { \mathfrak L }^{ ( n ) }_P + \frac{ ( q^0 )^n }{ ( n^0 + \sqrt{ s } )^2 } \partial_{ p_a } { \mathfrak L }^{ ( n ) }_P \\
& = \frac{ ( q^0 )^{ n + 1 } }{ ( n^0 + \sqrt{ s } )^2 } { \mathfrak L }^{ ( n + 1 ) }_P , 
\end{align*}
where we used $ { \mathfrak L }^{ ( 1 ) }_P { \mathfrak L }^{ ( n ) }_P = { \mathfrak L }^{ ( n + 1 ) }_P $ and $ \partial_{ p_a } { \frak L }^{ ( n ) }_P = q^0 { \frak L }^{ ( n + 1 ) }_P $ in \eqref{dLP}. This completes the proof of the lemma. 
\end{proof}

\begin{remark}
Note that the estimates in Lemma \ref{lem 1/s} and \ref{lem 1/ns} can also be written as follows: for any multi-index $ { \mathcal I } $ with $ | { \mathcal I } | = n \geq 0 $, we have
\begin{align}
& \partial^{ \mathcal I }_p \left( \frac{ 1 }{ \sqrt{ s } } \right) = \frac{ ( q^0 )^n }{ \sqrt{ s } } { \mathfrak L }^{ ( n ) }_P , \label{rem 1/s} \\
& \partial^{ \mathcal I }_p \left( \frac{ 1 }{ n^0 + \sqrt{ s } } \right) = \frac{ ( q^0 )^n }{ n^0 + \sqrt{ s } } { \mathfrak L }^{ ( n ) }_P , \label{rem 1/ns}
\end{align}
for some different $ { \mathfrak L }^{ ( n ) }_P $. 
\end{remark}

\begin{lemma}\label{lem dp'}
Suppose that a metric $ g^{ a b } $ is given and satisfies \eqref{assumption1}. Then, the post-collision momenta $ p' $ and $ q' $ satisfy for any multi-index $ { \mathcal I } $ with $ | { \mathcal I } | = n \geq 1 $, 
\begin{align*}
& | \partial^{ \mathcal I }_p p' | \leq C \delta^n_1 + ( q^0 )^{ n + 8 } \left( \max_{ i , a } | { e^i }_a | \right)^2 Z^{ - n - 1 } { \mathcal C } , \\ 
& | \partial^{ \mathcal I }_p q' | \leq ( q^0 )^{ n + 8 } \left( \max_{ i , a } | { e^i }_a | \right)^2 Z^{ - n - 1 } { \mathcal C } , 
\end{align*}
where $ \delta^n_1 $ denotes the Kronecker delta. 
\end{lemma}
\begin{proof}
Recall that the post-collision momenta are given by \eqref{p'} and \eqref{q'}. Note that $ { e^i }_a $ are the components of $ { \bf e }^i $ with respect to $ { \bf W }^a $, which depends only on $ t $, and $ \omega_i \in \bbs^2 $ is an independent variable. Hence, concerning $ p $-derivatives of $ p' $ and $ q' $, we only need to estimate the following quantities:
\[
\frac{ n^a }{ \sqrt{ s } } , \qquad \frac{ n_b }{ n^0 + \sqrt{ s } } . 
\]
We note that $ n^a $ and $ n_b $ in the above quantities are estimated as follows. Since we can write 
\begin{align*}
n^a = p^0 q^0 \left( \frac{ 1 }{ q^0 } \frac{ p^a }{ p^0 } + \frac{ 1 }{ p^0 } \frac{ q^a }{ q^0 } \right) = p^0 q^0 { \mathfrak L }^{ ( 1 ) }_P , 
\end{align*}
we obtain 
\begin{align}
| n^a | \leq p^0 q^0 Z^{ - 1 } { \mathcal C } , \label{n^a} 
\end{align}
for some $ { \mathcal C } $. We also have 
\begin{align}
| n_b | \leq | p | + | q | \leq \sqrt{ c_1 } Z ( p^0 + q^0 ) \leq C Z p^0 q^0 , \label{n_b} 
\end{align}
where we used the assumption \eqref{assumption1} such that $ | p | \leq \sqrt{ c_1 } Z p^0 $. We also need the following property: for any $ p , q \in \bbr^3 $, 
\begin{align}
s \geq \max \left( \frac{ p^0 }{ q^0 } , \frac{ q^0 }{ p^0 } \right) , \label{sp^0}
\end{align}
which can be found in Lemma 3 of Ref.~\cite{LN171}. We first consider the following quantity:
\[
\partial^{ \mathcal I }_p \left( \frac{ n^a }{ \sqrt{ s } } \right) . 
\]
If $ | \mathcal I | = 0 $, we use \eqref{n^a} and \eqref{sp^0} to obtain 
\[
\left| \frac{ n^a }{ \sqrt{ s } } \right| \leq p^0 q^0 Z^{ - 1 } { \mathcal C } \sqrt{ \frac{ q^0 }{ p^0 } } \leq \sqrt{ p^0 } ( q^0 )^{ \frac32 } Z^{ - 1 } { \mathcal C } . 
\]
If $ | \mathcal I | = n \geq 1 $, we write 
\begin{align*}
\partial^{ \mathcal I }_p \left( \frac{ n^a }{ \sqrt{ s } } \right) & = \sum_{ { \mathcal I }_1 \leq { \mathcal I } } { { \mathcal I } \choose { \mathcal I }_1 } ( \partial^{ { \mathcal I }_1 }_p n^a ) \partial^{ { \mathcal I } - { \mathcal I }_1 }_p \left( \frac{ 1 }{ \sqrt{ s } } \right) \\
& = n^a \partial^{ { \mathcal I } }_p \left( \frac{ 1 }{ \sqrt{ s } } \right) + \sum_{ \substack{ { \mathcal I }_1 \leq { \mathcal I } \\ | { \mathcal I }_1 | = 1 } } { { \mathcal I } \choose { \mathcal I }_1 } ( \partial^{ { \mathcal I }_1 }_p n^a ) \partial^{ { \mathcal I } - { \mathcal I }_1 }_p \left( \frac{ 1 }{ \sqrt{ s } } \right) ,
\end{align*}
where the first quantity is estimated as
\begin{align*}
\left| n^a \partial^{ { \mathcal I } }_p \left( \frac{ 1 }{ \sqrt{ s } } \right) \right| & \leq p^0 q^0 Z^{ - 1 } { \mathcal C } \frac{ ( q^0 )^n }{ s^{ \frac32 } } | { \mathfrak L }^{ ( n ) }_P | \leq \frac{ ( q^0 )^{ n + \frac52 } }{ \sqrt{ p^0 } } Z^{ - n - 1 } { \mathcal C } . 
\end{align*}
Note that $ \partial^{ { \mathcal I }_1 }_p n^a = { \mathfrak L }^{ ( 2 ) } $ for $ | { \mathcal I }_1 | = 1 $. We use \eqref{rem 1/s} and \eqref{sp^0} for the second quantity to obtain 
\begin{align*}
\left| \sum_{ \substack{ { \mathcal I }_1 \leq { \mathcal I } \\ | { \mathcal I }_1 | = 1 } } { { \mathcal I } \choose { \mathcal I }_1 } ( \partial^{ { \mathcal I }_1 }_p n^a ) \partial^{ { \mathcal I } - { \mathcal I }_1 }_p \left( \frac{ 1 }{ \sqrt{ s } } \right) \right| \leq C | { \mathfrak L }^{ ( 2 ) } | \frac{ ( q^0 )^{ n - 1 } }{ \sqrt{ s } } | { \mathfrak L }^{ ( n - 1 ) }_P | \leq \frac{ ( q^0 )^{ n - \frac12 } }{ \sqrt{ p^0 } } Z^{ - n - 1 } { \mathcal C } , 
\end{align*}
where the constant $ C $ is absorbed into $ { \mathcal C } $. Next, we consider the following quantity:
\[
\partial^{ \mathcal I }_p \left( \frac{ n_b }{ n^0 + \sqrt{ s } } \right) . 
\]
If $ | { \mathcal I } | = 0 $, we use \eqref{n_b} to obtain 
\[
\left| \frac{ n_b }{ n^0 + \sqrt{ s } } \right| \leq \frac{ C Z p^0 q^0 }{ n^0 + \sqrt{ s } } \leq C Z q^0 . 
\]
If $ | \mathcal I | = n \geq 1 $, we write 
\begin{align*}
\partial^{ \mathcal I }_p \left( \frac{ n_b }{ n^0 + \sqrt{ s } } \right) & = \sum_{ { \mathcal I }_1 \leq { \mathcal I } } { { \mathcal I } \choose { \mathcal I }_1 } ( \partial^{ { \mathcal I }_1 }_p n_b ) \partial^{ { \mathcal I } - { \mathcal I }_1 }_p \left( \frac{ 1 }{ n^0 + \sqrt{ s } } \right) \\
& = n_b \partial^{ { \mathcal I } }_p \left( \frac{ 1 }{ n^0 + \sqrt{ s } } \right) + \sum_{ \substack{ { \mathcal I }_1 \leq { \mathcal I } \\ | { \mathcal I }_1 | = 1 } } { { \mathcal I } \choose { \mathcal I }_1 } ( \partial^{ { \mathcal I }_1 }_p n_b ) \partial^{ { \mathcal I } - { \mathcal I }_1 }_p \left( \frac{ 1 }{ n^0 + \sqrt{ s } } \right) ,
\end{align*}
where the first quantity is estimated as
\[
\left| n_b \partial^{ { \mathcal I } }_p \left( \frac{ 1 }{ n^0 + \sqrt{ s } } \right) \right| \leq C Z p^0 q^0 \frac{ ( q^0 )^n }{ ( n^0 + \sqrt{ s } )^2 } | { \mathfrak L }^{ ( n ) }_P | \leq \frac{ ( q^0 )^{ n + 1 } }{ p^0 } Z^{ - n + 1 } { \mathcal C } . 
\]
For the second quantity we use \eqref{rem 1/ns} to obtain 
\[
\left| \sum_{ \substack{ { \mathcal I }_1 \leq { \mathcal I } \\ | { \mathcal I }_1 | = 1 } } { { \mathcal I } \choose { \mathcal I }_1 } ( \partial^{ { \mathcal I }_1 }_p n_b ) \partial^{ { \mathcal I } - { \mathcal I }_1 }_p \left( \frac{ 1 }{ n^0 + \sqrt{ s } } \right) \right| \leq C \frac{ ( q^0 )^{ n - 1 } }{ n^0 + \sqrt{ s } } | { \mathfrak L }^{ ( n - 1 ) }_P | \leq \frac{ ( q^0 )^{ n - 1 } }{ p^0 } Z^{ - n + 1 } { \mathcal C } . 
\]
To summarize, we have obtained
\begin{align}
\left| \partial^{ \mathcal I }_p \left( \frac{ n^a }{ \sqrt{ s } } \right) \right| \leq 
\left\{
\begin{aligned}\label{n/s}
& \sqrt{ p^0 } ( q^0 )^{ \frac32 } Z^{ - 1 } { \mathcal C } , & & | { \mathcal I } | = 0 , \\
& \frac{ ( q^0 )^{ n + \frac52 } }{ \sqrt{ p^0 } } Z^{ - n - 1 } { \mathcal C } , & & | { \mathcal I } | = n \geq 1 , 
\end{aligned}
\right. 
\end{align}
and
\begin{align}\label{n/n^0s}
\left| \partial^{ \mathcal I }_p \left( \frac{ n_b }{ n^0 + \sqrt{ s } } \right) \right| \leq
\left\{
\begin{aligned}
& C Z q^0 , & & | { \mathcal I } | = 0 , \\
& \frac{ ( q^0 )^{ n + 1 } }{ p^0 } Z^{ - n + 1 } { \mathcal C } , & & | { \mathcal I } | = n \geq 1 . 
\end{aligned}
\right. 
\end{align}
We also need to estimate the following quantity, which can be easily estimated by using the above estimates \eqref{n/s} and \eqref{n/n^0s}: 
\[
\partial^{ \mathcal I }_p \left( \frac{ n^a n_b }{ \sqrt{ s } ( n^0 + \sqrt{ s } ) } \right) . 
\]
If $ | { \mathcal I } | = 0 $, we obtain 
\[
\left| \frac{ n^a n_b }{ \sqrt{ s } ( n^0 + \sqrt{ s } ) } \right| \leq \sqrt{ p^0 } ( q^0 )^{ \frac52 } { \mathcal C } . 
\]
If $ | { \mathcal I } | = n \geq 1 $, we have 
\begin{align*}
& \partial^{ \mathcal I }_p \left( \frac{ n^a n_b }{ \sqrt{ s } ( n^0 + \sqrt{ s } ) } \right) \\ 
& = \sum_{ { \mathcal I }_1 \leq { \mathcal I } } { { \mathcal I } \choose { \mathcal I }_1 } \partial^{ { \mathcal I }_1 }_p \left( \frac{ n^a }{ \sqrt{ s } } \right) \partial^{ { \mathcal I } - { \mathcal I }_1 }_p \left( \frac{ n_b }{ n^0 + \sqrt{ s } } \right) \\
& = \left( \frac{ n^a }{ \sqrt{ s } } \right) \partial^{ { \mathcal I } }_p \left( \frac{ n_b }{ n^0 + \sqrt{ s } } \right) + \sum_{ 0 < { \mathcal I }_1 < { \mathcal I } } { { \mathcal I } \choose { \mathcal I }_1 } \partial^{ { \mathcal I }_1 }_p \left( \frac{ n^a }{ \sqrt{ s } } \right) \partial^{ { \mathcal I } - { \mathcal I }_1 }_p \left( \frac{ n_b }{ n^0 + \sqrt{ s } } \right) + \partial^{ { \mathcal I } }_p \left( \frac{ n^a }{ \sqrt{ s } } \right) \left( \frac{ n_b }{ n^0 + \sqrt{ s } } \right) , 
\end{align*}
where the first and third quantities are estimated as
\begin{align*}
& \left| \left( \frac{ n^a }{ \sqrt{ s } } \right) \partial^{ { \mathcal I } }_p \left( \frac{ n_b }{ n^0 + \sqrt{ s } } \right) \right| \leq \frac{ ( q^0 )^{ n + \frac52 } }{ \sqrt{ p^0 } } Z^{ - n } { \mathcal C } , \\
& \left| \partial^{ { \mathcal I } }_p \left( \frac{ n^a }{ \sqrt{ s } } \right) \left( \frac{ n_b }{ n^0 + \sqrt{ s } } \right) \right| \leq \frac{ ( q^0 )^{ n + \frac72 } }{ \sqrt{ p^0 } } Z^{ - n } { \mathcal C } , 
\end{align*}
and the second quantity, which appears only for $ | { \mathcal I } | = n \geq 2 $, is estimated as follows:
\begin{align*}
& \left| \sum_{ 0 < { \mathcal I }_1 < { \mathcal I } } { { \mathcal I } \choose { \mathcal I }_1 } \partial^{ { \mathcal I }_1 }_p \left( \frac{ n^a }{ \sqrt{ s } } \right) \partial^{ { \mathcal I } - { \mathcal I }_1 }_p \left( \frac{ n_b }{ n^0 + \sqrt{ s } } \right) \right| \\
& \leq \sum_{ \substack{ l + m = n \\ l , m \neq 0 } } \frac{ ( q^0 )^{ l + \frac52 } }{ \sqrt{ p^0 } } Z^{ - l - 1 } { \mathcal C } \frac{ ( q^0 )^{ m + 1 } }{ p^0 } Z^{ - m + 1 } { \mathcal C } \\
& \leq \frac{ ( q^0 )^{ n + \frac72 } }{ ( p^0 )^{ \frac32 } } Z^{ - n } { \mathcal C } , 
\end{align*}
for some $ { \mathcal C } $. To summarize, we have 
\begin{align}\label{nn/sn^0s}
\left| \partial^{ \mathcal I }_p \left( \frac{ n^a n_b }{ \sqrt{ s } ( n^0 + \sqrt{ s } ) } \right) \right| \leq 
\left\{
\begin{aligned}
& \sqrt{ p^0 } ( q^0 )^{ \frac52 } { \mathcal C } , & & | { \mathcal I } | = 0 , \\
& \frac{ ( q^0 )^{ n + \frac72 } }{ \sqrt{ p^0 } } Z^{ - n } { \mathcal C } , & & | { \mathcal I } | = n \geq 1 . 
\end{aligned}
\right. 
\end{align}
Now, we can estimate the following quantity. We need the estimates \eqref{n/s} and \eqref{nn/sn^0s}: 
\[
\partial^{ \mathcal I }_p \left( - q^0 \frac{ n^a \omega_i { e^i }_a }{ \sqrt{ s } } + q^a \omega_i { e^i }_a + \frac{ n^a \omega_i { e^i }_a n_b q^b }{ \sqrt{ s } ( n^0 + \sqrt{ s } ) } \right) . 
\]
If $ | { \mathcal I } | = 0 $, we obtain
\begin{align}
& \left| - q^0 \frac{ n^a \omega_i { e^i }_a }{ \sqrt{ s } } + q^a \omega_i { e^i }_a + \frac{ n^a \omega_i { e^i }_a n_b q^b }{ \sqrt{ s } ( n^0 + \sqrt{ s } ) } \right| \nonumber \\ 
& \leq C \left( q^0 \sqrt{ p^0 } ( q^0 )^{ \frac32 } Z^{ - 1 } { \mathcal C } \left( \max_{ i , a } | { e^i }_a | \right) + q^0 | { \mathfrak L }^{ ( 1 ) } | \left( \max_{ i , a } | { e^i }_a | \right) + \sqrt{ p^0 } ( q^0 )^{ \frac52 } { \mathcal C } \left( \max_{ i , a } | { e^i }_a | \right) q^0 | { \mathfrak L }^{ ( 1 ) } | \right) \nonumber \\
& \leq \left( \max_{ i , a } | { e^i }_a | \right) \left( \sqrt{ p^0 } ( q^0 )^{ \frac52 } Z^{ - 1 } { \mathcal C } + q^0 | { \mathfrak L }^{ ( 1 ) } | + \sqrt{ p^0 } ( q^0 )^{ \frac72 } { \mathcal C } | { \mathfrak L }^{ ( 1 ) } | \right) \nonumber \\
& \leq \sqrt{ p^0 } ( q^0 )^{ \frac72 } \left( \max_{ i , a } | { e^i }_a | \right) Z^{ - 1 } { \mathcal C } . \label{p'est1} 
\end{align}
If $ | { \mathcal I } | = n \geq 1 $, we write 
\begin{align}
& \left| \partial^{ \mathcal I }_p \left( - q^0 \frac{ n^a \omega_i { e^i }_a }{ \sqrt{ s } } + q^a \omega_i { e^i }_a + \frac{ n^a \omega_i { e^i }_a n_b q^b }{ \sqrt{ s } ( n^0 + \sqrt{ s } ) } \right) \right| \nonumber \\ 
& = \left| - q^0 \partial^{ \mathcal I }_p \left( \frac{ n^a }{ \sqrt{ s } } \right) \omega_i { e^i }_a + \partial^{ \mathcal I }_p \left( \frac{ n^a n_b }{ \sqrt{ s } ( n^0 + \sqrt{ s } ) } \right) \omega_i { e^i }_a q^b \right| \nonumber \\ 
& \leq C \left( q^0 \frac{ ( q^0 )^{ n + \frac52 } }{ \sqrt{ p^0 } } Z^{ - n - 1 } { \mathcal C } \left( \max_{ i , a } | { e^i }_a | \right) + \frac{ ( q^0 )^{ n + \frac72 } }{ \sqrt{ p^0 } } Z^{ - n } { \mathcal C } \left( \max_{ i , a } | { e^i }_a | \right) q^0 | { \mathfrak L }^{ ( 1 ) } | \right) \nonumber \allowdisplaybreaks \\ 
& \leq \left( \max_{ i , a } | { e^i }_a | \right) \left( \frac{ ( q^0 )^{ n + \frac72 } }{ \sqrt{ p^0 } } Z^{ - n - 1 } { \mathcal C } + \frac{ ( q^0 )^{ n + \frac92 } }{ \sqrt{ p^0 } } Z^{ - n } { \mathcal C } | { \mathfrak L }^{ ( 1 ) } | \right) \nonumber \\
& \leq \frac{ ( q^0 )^{ n + \frac92 } }{ \sqrt{ p^0 } } \left( \max_{ i , a } | { e^i }_a | \right) Z^{ - n - 1 } { \mathcal C } . \label{p'est2} 
\end{align}
Similarly, we estimate the following quantity. Here, we only need the estimate \eqref{nn/sn^0s}:
\[
\partial^{ \mathcal I }_p \left( \omega_j { e^j }_d + \frac{ n^c \omega_j { e^j }_c n_d }{ \sqrt{ s } ( n^0 + \sqrt{ s } ) } \right) . 
\]
If $ | \mathcal I | = 0 $, then we have
\begin{align}
& \left| \omega_j { e^j }_d + \frac{ n^c \omega_j { e^j }_c n_d }{ \sqrt{ s } ( n^0 + \sqrt{ s } ) } \right| \nonumber \\
& \leq C \left( \left( \max_{ i , a } | { e^i }_a | \right) + \sqrt{ p^0 } ( q^0 )^{ \frac52 } { \mathcal C } \left( \max_{ i , a } | { e^i }_a | \right) \right) \nonumber \\
& \leq \sqrt{ p^0 } ( q^0 )^{ \frac52 } \left( \max_{ i , a } | { e^i }_a | \right) { \mathcal C } . \label{p'est3} 
\end{align}
If $ | { \mathcal I } | = n \geq 1 $, then we easily obtain 
\begin{align}
& \left| \partial^{ \mathcal I }_p \left( \omega_j { e^j }_d + \frac{ n^c \omega_j { e^j }_c n_d }{ \sqrt{ s } ( n^0 + \sqrt{ s } ) } \right) \right| \leq \frac{ ( q^0 )^{ n + \frac72 } }{ \sqrt{ p^0 } } \left( \max_{ i , a } | { e^i }_a | \right) Z^{ - n } { \mathcal C } . \label{p'est4} 
\end{align}
Now, we combine the above results to prove the lemma. Recall that the post-collision momenta are given by 
\begin{align*}
p'_d & = p_d + 2 \bigg( - q^0 \frac{ n^a \omega_i { e^i }_a }{ \sqrt{ s } } + q^a \omega_i { e^i }_a + \frac{ n^a \omega_i { e^i }_a n_b q^b }{ \sqrt{ s } ( n^0 + \sqrt{ s } ) } \bigg) \bigg( \omega_j { e^j }_d + \frac{ n^c \omega_j { e^j }_c n_d }{ \sqrt{ s } ( n^0 + \sqrt{ s } ) } \bigg) , \\
q'_d & = q_d - 2 \bigg( - q^0 \frac{ n^a \omega_i { e^i }_a }{ \sqrt{ s } } + q^a \omega_i { e^i }_a + \frac{ n^a \omega_i { e^i }_a n_b q^b }{ \sqrt{ s } ( n^0 + \sqrt{ s } ) } \bigg) \bigg( \omega_j { e^j }_d + \frac{ n^c \omega_j { e^j }_c n_d }{ \sqrt{ s } ( n^0 + \sqrt{ s } ) } \bigg) , 
\end{align*}
and let $ { \mathcal I } $ be a multi-index of order $ | \mathcal I | = n \geq 1 $. We first notice that
\begin{align}
| \partial^{ \mathcal I }_p p_d | \leq 
\left\{
\begin{aligned}
& C , & & | { \mathcal I } | = 1 , \\
& 0 , & & | { \mathcal I } | = n \geq 2 ,  
\end{aligned}
\right. \label{p'est5}
\end{align}
and $ \partial^{ \mathcal I }_p q_d = 0 $ for any $ | { \mathcal I } | = n \geq 1 $. Next, we consider the following quantity: 
\begin{align*}
& \partial^{ \mathcal I }_p \left( \bigg( - q^0 \frac{ n^a \omega_i { e^i }_a }{ \sqrt{ s } } + q^a \omega_i { e^i }_a + \frac{ n^a \omega_i { e^i }_a n_b q^b }{ \sqrt{ s } ( n^0 + \sqrt{ s } ) } \bigg) \bigg( \omega_j { e^j }_d + \frac{ n^c \omega_j { e^j }_c n_d }{ \sqrt{ s } ( n^0 + \sqrt{ s } ) } \bigg) \right) \\
& = \partial^{ \mathcal I }_p \bigg( - q^0 \frac{ n^a \omega_i { e^i }_a }{ \sqrt{ s } } + q^a \omega_i { e^i }_a + \frac{ n^a \omega_i { e^i }_a n_b q^b }{ \sqrt{ s } ( n^0 + \sqrt{ s } ) } \bigg) \bigg( \omega_j { e^j }_d + \frac{ n^c \omega_j { e^j }_c n_d }{ \sqrt{ s } ( n^0 + \sqrt{ s } ) } \bigg) \\
& \quad + \sum_{ 0 < { \mathcal I }_1 < { \mathcal I } } { { \mathcal I } \choose { \mathcal I }_1 } \partial^{ { \mathcal I }_1 }_p \left( - q^0 \frac{ n^a \omega_i { e^i }_a }{ \sqrt{ s } } + q^a \omega_i { e^i }_a + \frac{ n^a \omega_i { e^i }_a n_b q^b }{ \sqrt{ s } ( n^0 + \sqrt{ s } ) } \right) \partial^{ { \mathcal I } - { \mathcal I }_1 }_p \left( \omega_j { e^j }_d + \frac{ n^c \omega_j { e^j }_c n_d }{ \sqrt{ s } ( n^0 + \sqrt{ s } ) } \right) \\
& \quad + \bigg( - q^0 \frac{ n^a \omega_i { e^i }_a }{ \sqrt{ s } } + q^a \omega_i { e^i }_a + \frac{ n^a \omega_i { e^i }_a n_b q^b }{ \sqrt{ s } ( n^0 + \sqrt{ s } ) } \bigg) \partial^{ \mathcal I }_p \bigg( \omega_j { e^j }_d + \frac{ n^c \omega_j { e^j }_c n_d }{ \sqrt{ s } ( n^0 + \sqrt{ s } ) } \bigg) , 
\end{align*}
where the first and third terms are estimated as follows:
\begin{align}
& \left| \partial^{ \mathcal I }_p \bigg( - q^0 \frac{ n^a \omega_i { e^i }_a }{ \sqrt{ s } } + q^a \omega_i { e^i }_a + \frac{ n^a \omega_i { e^i }_a n_b q^b }{ \sqrt{ s } ( n^0 + \sqrt{ s } ) } \bigg) \bigg( \omega_j { e^j }_d + \frac{ n^c \omega_j { e^j }_c n_d }{ \sqrt{ s } ( n^0 + \sqrt{ s } ) } \bigg) \right| \nonumber \\
& \leq ( q^0 )^{ n + 7 } \left( \max_{ i , a } | { e^i }_a | \right)^2 Z^{ - n - 1 } { \mathcal C } , \label{p'est6}
\end{align}
where we used \eqref{p'est2} and \eqref{p'est3}, and 
\begin{align}
& \left| \bigg( - q^0 \frac{ n^a \omega_i { e^i }_a }{ \sqrt{ s } } + q^a \omega_i { e^i }_a + \frac{ n^a \omega_i { e^i }_a n_b q^b }{ \sqrt{ s } ( n^0 + \sqrt{ s } ) } \bigg) \partial^{ \mathcal I }_p \bigg( \omega_j { e^j }_d + \frac{ n^c \omega_j { e^j }_c n_d }{ \sqrt{ s } ( n^0 + \sqrt{ s } ) } \bigg) \right| \nonumber \\ 
& \leq ( q^0 )^{ n + 7 } \left( \max_{ i , a } | { e^i }_a | \right)^2 Z^{ - n - 1 } { \mathcal C } , \label{p'est7}
\end{align}
where we used \eqref{p'est1} and \eqref{p'est4}. The second term, which appears only for $ | { \mathcal I } | = n \geq 2 $, is estimated by using \eqref{p'est2} and \eqref{p'est4} as follows: 
\begin{align}
& \left| \sum_{ 0 < { \mathcal I }_1 < { \mathcal I } } { { \mathcal I } \choose { \mathcal I }_1 } \partial^{ { \mathcal I }_1 }_p \left( - q^0 \frac{ n^a \omega_i { e^i }_a }{ \sqrt{ s } } + q^a \omega_i { e^i }_a + \frac{ n^a \omega_i { e^i }_a n_b q^b }{ \sqrt{ s } ( n^0 + \sqrt{ s } ) } \right) \partial^{ { \mathcal I } - { \mathcal I }_1 }_p \left( \omega_j { e^j }_d + \frac{ n^c \omega_j { e^j }_c n_d }{ \sqrt{ s } ( n^0 + \sqrt{ s } ) } \right) \right| \nonumber \\
& \leq \frac{ ( q^0 )^{ n + 8 } }{ p^0 } \left( \max_{ i , a } | { e^i }_a | \right)^2 Z^{ - n - 1 } { \mathcal C } . \label{p'est8} 
\end{align}
We combine \eqref{p'est5}--\eqref{p'est8} to conclude that 
\[
| \partial^{ \mathcal I }_p p'_d | \leq 
\left\{
\begin{aligned}
& C + ( q^0 )^{ 8 } \left( \max_{ i , a } | { e^i }_a | \right)^2 Z^{ - 2 } { \mathcal C } , & & | { \mathcal I } | = 1 , \\
& ( q^0 )^{ n + 8 } \left( \max_{ i , a } | { e^i }_a | \right)^2 Z^{ - n - 1 } { \mathcal C } , & & | { \mathcal I } | = n \geq 2 , 
\end{aligned}
\right. 
\]
which implies the desired result for $ \partial^{ \mathcal I }_p p' $. The estimates for $ \partial^{ \mathcal I }_p q' $ are almost the same, and we skip the proof of it. This completes the proof of Lemma \ref{lem dp'}. 
\end{proof}

\begin{remark}
We observe that the $ 3 \times 3 $ matrix $ ( { e^i }_a )_{ i , a = 1 , 2 , 3 } $ defines a linear map: $ p_a = { \hat p }_i { e^i }_a $. Since the assumption \eqref{assumption1} implies 
\[
\frac{ | p |^2 }{ | { \hat p } |^2 } \leq \frac{ c_1 Z^2 g^{ a b } p_a p_b }{ | { \hat p } |^2 } = c_1 Z^2 , 
\]
the norm of the matrix $ ( { e^i }_a )_{ i , a = 1 , 2 , 3 } $, as an operator, is estimated by 
\[
\| ( { e^i }_a )_{ i , a = 1 , 2 , 3 } \| \leq \sqrt{ c_1 } Z . 
\]
Hence, we obtain
\begin{align}
\max_{ i , a } | { e^i }_a | \leq C Z , \label{eZ} 
\end{align}
where the constant $ C $ depends on $ c_1 $. Now, applying the estimate \eqref{eZ}, we may write Lemma \ref{lem dp'} as follows: for any multi-index $ { \mathcal I } $ with $ | { \mathcal I } | = n \geq 1 $, we have 
\begin{align}
| \partial^{ \mathcal I }_p p' | , | \partial^{ \mathcal I }_p q' | \leq ( q^0 )^{ n + 8 } Z^{ - n + 1 } { \mathcal C } , \label{dp'} 
\end{align}
where we notice that $ C \delta^n_1 = ( q^0 )^{ 9 } { \mathcal C } $ holds for $ n = 1 $. 
\end{remark}

\begin{remark}\label{rem ortho}
We remark that the orthonormal basis can be chosen as follows:
\begin{align*}
& ( { e^i }_a ) 
= \begin{pmatrix}
{ \bf e }^1 \\ 
{ \bf e }^2 \\ 
{ \bf e }^3 
\end{pmatrix}
= \begin{pmatrix}
\frac{ 1 }{ \sqrt{ g^{ 1 1 } } }  & 0 & 0 \\ 
\frac{ - g^{ 1 2 } }{ \sqrt{ g^{ 1 1 } ( g^{ 1 1 } g^{ 2 2 } - ( g^{ 1 2 } )^2 ) } } & \frac{ g^{ 1 1 } }{ \sqrt{ g^{ 1 1 } ( g^{ 1 1 } g^{ 2 2 } - ( g^{ 1 2 } )^2 ) } } & 0 \\
\frac{ g^{ 1 2 } g^{ 2 3 } - g^{ 1 3 } g^{ 2 2 } }{ \sqrt{ ( g^{ 1 1 } g^{ 2 2 } - ( g^{ 1 2 } )^2 ) ( \det g^{ - 1 } ) } } & \frac{ - g^{ 1 1 } g^{ 2 3 } + g^{ 1 2 } g^{ 1 3 } }{ \sqrt{ ( g^{ 1 1 } g^{ 2 2 } - ( g^{ 1 2 } )^2 ) ( \det g^{ - 1 } ) } } & \frac{ g^{ 1 1 } g^{ 2 2 } - ( g^{ 1 2 } )^2 }{ \sqrt{ ( g^{ 1 1 } g^{ 2 2 } - ( g^{ 1 2 } )^2 ) ( \det g^{ - 1 } ) } }
\end{pmatrix} , \\ 
& ( { e_i }^a ) 
= \begin{pmatrix}
{ \bf e }_1 & { \bf e }_2 & { \bf e }_3 
\end{pmatrix}
= \begin{pmatrix}
\frac{ g^{ 1 1 } }{ \sqrt{ g^{ 1 1 } } }  & 0 & 0 \\ 
\frac{ g^{ 1 2 } }{ \sqrt{ g^{ 1 1 } } } & \frac{ g^{ 1 1 } g^{ 2 2 } - ( g^{ 1 2 } )^2 }{ \sqrt{ g^{ 1 1 } ( g^{ 1 1 } g^{ 2 2 } - ( g^{ 1 2 } )^2 ) } } & 0 \\
\frac{ g^{ 1 3 } }{ \sqrt{ g^{ 1 1 } } } & \frac{ g^{ 1 1 } g^{ 2 3 } - g^{ 1 2 } g^{ 1 3 } }{ \sqrt{ g^{ 1 1 } ( g^{ 1 1 } g^{ 2 2 } - ( g^{ 1 2 } )^2 ) } } & \frac{ \det g^{ - 1 } }{ \sqrt{ ( g^{ 1 1 } g^{ 2 2 } - ( g^{ 1 2 } )^2 ) ( \det g^{ - 1 } ) } } 
\end{pmatrix} ,
\end{align*}
which satisfy $ { e^i }_a { e^j }_b g^{ a b } = \delta^{ i j } $ and $ { e_i }^a { e_j }^b g_{ a b } = \delta_{ i j } $. 
\end{remark}

We note that the post-collision momenta \eqref{p'} and \eqref{q'} depend on the metric $ g^{ a b } $ through the quantities $ q^0 $, $ n^a $, $ { e^i }_a $, $ \sqrt{ s } $, etc. In the following lemma, we estimate the derivatives of the post-collision momenta with respect to the metric $ g^{ a b } $. This will be necessary for the proof of uniqueness of solutions to the coupled EBs system in Proposition \ref{prop local}.

\begin{lemma}\label{lem dp'g}
Suppose that a metric $ g^{ a b } $ is given and satisfies 
\[
\frac{ 1 }{ C } | p |^2 \leq g^{ a b } p_a p_b \leq C | p |^2 , \qquad \max_{ a , b } | g^{ a b } | \leq C , \qquad \det g^{ - 1 } \geq \frac{ 1 }{ C } , 
\]
on an interval $ [ t_0 , T ] $. Then, the post-collision momenta $ p' $ and $ q' $ satisfy 
\[
\left| \frac{ \partial p_d' }{ \partial g^{ u v } } \right| + \left| \frac{ \partial q_d' }{ \partial g^{ u v } } \right| \leq C \langle p \rangle \langle q \rangle^4 , 
\]
on $ [ t_0 , T ] $. 
\end{lemma}
\begin{proof}
Recall that $ p' $ is given by 
\begin{align*}
p'_d & = p_d + 2 \bigg( - q^0 \frac{ n^a \omega_i { e^i }_a }{ \sqrt{ s } } + q^a \omega_i { e^i }_a + \frac{ n^a \omega_i { e^i }_a n_b q^b }{ \sqrt{ s } ( n^0 + \sqrt{ s } ) } \bigg) \bigg( \omega_j { e^j }_d + \frac{ n^c \omega_j { e^j }_c n_d }{ \sqrt{ s } ( n^0 + \sqrt{ s } ) } \bigg) . 
\end{align*}
Note that the quantities $ p_d $, $ \omega_i $ and $ n_b $ are independent of $ g^{ a b } $ so that we only need to compute the following quantities: 
\[
\frac{ \partial q^0 }{ \partial g^{ u v } } , \qquad \frac{ \partial n^a }{ \partial g^{ u v } } , \qquad \frac{ \partial { e^i }_a }{ \partial g^{ u v } } , \qquad \frac{ \partial }{ \partial g^{ u v } } \left( \frac{ 1 }{ \sqrt{ s } } \right) , \qquad \frac{ \partial q^a }{ \partial g^{ u v } } , \qquad \frac{ \partial }{ \partial g^{ u v } } \left( \frac{ 1 }{ n^0 + \sqrt{ s } } \right) , 
\]
where $ u , v = 1 , 2 , 3 $. We will need to compute the derivatives of $ g^{ a b } $ with respect to $ g^{ u v } $, which can be explicitly written as follows: 
\[
\frac{ \partial g^{ a b } }{ \partial g^{ u v } } = \frac{ 1 }{ 2^{ \delta_{ u v } } } \left( \delta^a_u \delta^b_v + \delta^b_u \delta^e_v \right) , 
\]
but it will be enough to know that they are bounded quantities. Note that the first condition of the lemma implies that $ p^0 $ and $ \langle p \rangle $ are equivalent. Then, we obtain the following estimates: 
\begin{align}
\left| \frac{ \partial q^0 }{ \partial g^{ u v } } \right| = \left| \frac{ 1 }{ 2 q^0 } \frac{ \partial g^{ a b } }{ \partial g^{ u v } } q_a q_b \right| & \leq C \langle q \rangle , \label{dq^0g} \\ 
\left| \frac{ \partial n^a }{ \partial g^{ u v } } \right| = \left| \frac{ \partial g^{ a b } }{ \partial g^{ u v } } n_b \right| & \leq C \left( \langle p \rangle + \langle q \rangle \right) , \nonumber \\ 
\left| \frac{ \partial q^a }{ \partial g^{ u v } } \right| = \left| \frac{ \partial g^{ a b } }{ \partial g^{ u v } } q_b \right| & \leq C \langle q \rangle . \nonumber 
\end{align}
To compute the derivatives of $ { e^i }_a $, we use the expression in Remark \ref{rem ortho}. Since we have 
\[
\frac{ 1 }{ g^{ 1 1 } } \leq \frac{ g^{ 2 2 } g^{ 3 3 } }{ \det g^{ - 1 } } , \qquad \frac{ 1 }{ g^{ 1 1 } g^{ 2 2 } - ( g^{ 1 2 } )^2 } \leq \frac{ g^{ 3 3 } }{ \det g^{ - 1 } } , 
\]
(see (7) of Ref.~\cite{LN171}), and $ g^{ a b } $, $ \partial g^{ a b } / \partial g^{ u v } $ and $ 1 / \det g^{ - 1 } $ are all bounded quantities, we can conclude that 
\[
\left| \frac{ \partial { e^i }_a }{ \partial g^{ u v } } \right| \leq C . 
\]
To compute the derivatives of $ 1 / \sqrt{ s } $, we recall that $ s = 2 + 2 p^0 q^0 - 2 g^{ a b } p_a q_b $. Hence, we have
\[
\left| \frac{ \partial s }{ \partial g^{ u v } } \right| = \left| 2 \frac{ \partial p^0 }{ \partial g^{ u v } } q^0 + 2 p^0 \frac{ \partial q^0 }{ \partial g^{ u v } } - 2 \frac{ \partial g^{ a b } }{ \partial g^{ u v } } p_a q_a \right| \leq C \langle p \rangle \langle q \rangle , 
\]
and obtain 
\begin{align}
\left| \frac{ \partial }{ \partial g^{ u v } } \left( \frac{ 1 }{ \sqrt{ s } } \right) \right| \leq \frac{ C }{ s^{ \frac32 } } \langle p \rangle \langle q \rangle . \label{dsg} 
\end{align}
The derivatives of $ 1 / ( n^0 + \sqrt{ s } ) $ are estimated by using $ n^0 + \sqrt{ s } \geq C^{ - 1 } \langle p \rangle $ for some $ C > 0 $: 
\begin{align*}
\left| \frac{ \partial }{ \partial g^{ u v } } \left( \frac{ 1 }{ n^0 + \sqrt{ s } } \right) \right| & \leq \frac{ C }{ ( n^0 + \sqrt{ s } )^2 } \left( \langle p \rangle + \langle q \rangle + \frac{ 1 }{ \sqrt{ s } } \langle p \rangle \langle q \rangle \right) \\ 
& \leq C \left( \frac{ 1 }{ \langle p \rangle } + \frac{ \langle q \rangle }{ \langle p \rangle^2 } + \frac{ \langle q \rangle }{ \sqrt{ s } \langle p \rangle } \right) \\ 
& \leq C \frac{ \langle q \rangle }{ \langle p \rangle } , 
\end{align*}
where we used $ \sqrt{ s } \geq 2 $. Now, we write $ p'_d = p_d + A B_d $, where 
\begin{align*}
A & = 2 \bigg( - q^0 \frac{ n^a \omega_i { e^i }_a }{ \sqrt{ s } } + q^a \omega_i { e^i }_a + \frac{ n^a \omega_i { e^i }_a n_b q^b }{ \sqrt{ s } ( n^0 + \sqrt{ s } ) } \bigg) , \qquad B_d = \bigg( \omega_j { e^j }_d + \frac{ n^c \omega_j { e^j }_c n_d }{ \sqrt{ s } ( n^0 + \sqrt{ s } ) } \bigg) . 
\end{align*}
The quantity $ A $ is estimated by using \eqref{sp^0} and the inequality $ n^0 + \sqrt{ s } \geq C^{ - 1 } ( \langle p \rangle + \langle q \rangle ) $: 
\begin{align*}
| A | & \leq C \left( \langle q \rangle \frac{ \langle p \rangle + \langle q \rangle }{ \sqrt{ s } } + \langle q \rangle + \frac{ ( \langle p \rangle + \langle q \rangle )^2 \langle q \rangle }{ \sqrt{ s } ( n^0 + \sqrt{ s } ) } \right) \\ 
& \leq C \left( \langle p \rangle^{ \frac12 } \langle q \rangle^{ \frac32 } + \langle q \rangle + \frac{ ( \langle p \rangle + \langle q \rangle ) \langle q \rangle }{ \sqrt{ s } } \right) \\ 
& \leq C \langle p \rangle^{ \frac12 } \langle q \rangle^{ \frac32 } . 
\end{align*}
We combine the above results to estimate the derivatives of $ A $: 
\begin{align*}
\left| \frac{ \partial A }{ \partial g^{ u v } } \right| & \leq C \bigg( \langle q \rangle \frac{ \langle p \rangle + \langle q \rangle }{ \sqrt{ s } } + \langle q \rangle \frac{ \langle p \rangle + \langle q \rangle }{ s^{ \frac32 } } \langle p \rangle \langle q \rangle + \langle q \rangle \\
& \qquad \qquad + \frac{ ( \langle p \rangle + \langle q \rangle )^2 \langle q \rangle }{ \sqrt{ s } ( n^0 + \sqrt{ s } ) } + \frac{ ( \langle p \rangle + \langle q \rangle )^2 \langle q \rangle }{ s^{ \frac32 } ( n^0 + \sqrt{ s } ) } \langle p \rangle \langle q \rangle + \frac{ ( \langle p \rangle + \langle q \rangle )^2 \langle q \rangle^2 }{ \sqrt{ s } \langle p \rangle } \bigg) \\ 
& \leq C \left( \langle p \rangle^{ \frac12 } \langle q \rangle^{ \frac32 } + \langle p \rangle^{ \frac12 } \langle q \rangle^{ \frac72 } + \langle q \rangle + \frac{ 1 }{ \sqrt{ s } } \left( \langle p \rangle \langle q \rangle^2 + \langle q \rangle^3 + \frac{ \langle q \rangle^4 }{ \langle p \rangle } \right) \right) \\ 
& \leq C \langle p \rangle^{ \frac12 } \langle q \rangle^{ \frac72 } . 
\end{align*}
In a similar way, the quantity $ B_d $ is estimated as 
\begin{align*}
| B_d | & \leq C \left( 1 + \frac{ ( \langle p \rangle + \langle q \rangle )^2 }{ \sqrt{ s } ( n^0 + \sqrt{ s } ) } \right) \leq C \left( 1 + \frac{ \langle p \rangle + \langle q \rangle }{ \sqrt{ s } } \right) \leq C \langle p \rangle^{ \frac12 } \langle q \rangle^{ \frac12 } , 
\end{align*}
and its derivatives are 
\begin{align*} 
\left| \frac{ \partial B_d }{ \partial g^{ u v } } \right| & \leq C \left( 1 + \frac{ ( \langle p \rangle + \langle q \rangle )^2 }{ \sqrt{ s } ( n^0 + \sqrt{ s } ) } + \frac{ ( \langle p \rangle + \langle q \rangle )^2 \langle p \rangle \langle q \rangle }{ s^{ \frac32 } ( n^0 + \sqrt{ s } ) } + \frac{ ( \langle p \rangle + \langle q \rangle )^2 \langle q \rangle }{ \sqrt{ s } \langle p \rangle } \right) \\ 
& \leq C \langle p \rangle^{ \frac12 } \langle q \rangle^{ \frac52 } . 
\end{align*}
Now, we can conclude that the derivatives of $ p' $ with respect to $ g^{ u v } $ are estimated as follows: 
\[
\left| \frac{ \partial p'_d }{ \partial g^{ u v } } \right| \leq C \langle p \rangle \langle q \rangle^4 . 
\]
The estimates of the derivatives of $ q' $ are exactly the same, and this completes the proof. 
\end{proof}

\begin{lemma}\label{lem pp'q'}
Suppose that a metric $ g^{ a b } $ is given and satisfies \eqref{assumption1}. Then, the post-collision momenta $ p' $ and $ q' $ satisfy 
\[
\langle p \rangle \leq C \langle p' \rangle \langle q' \rangle . 
\]
\end{lemma}
\begin{proof}
We note that $ g^{ a b } p_a p_b = \delta^{ i j } { \hat p }_i { \hat p }_j = | { \hat p } |^2 $, where $ { \hat p }_i $ denotes the components of a momentum with respect to the orthonormal frame $ { \bf e }^i $, i.e., $ p_a = { \hat p }_i { e^i }_a $. Hence, the assumption \eqref{assumption1} can be written as
\[
\frac{ 1 }{ c_1 } | p |^2 \leq Z^2 | { \hat p } |^2 \leq c_1 | p |^2 , 
\]
so that we obtain 
\begin{align}
\frac{ 1 + | p |^2 }{ ( 1 + | p' |^2 ) ( 1 + | q' |^2 ) } & \leq \frac{ 1 + c_1 Z^2 | { \hat p } |^2 }{ ( 1 + c_1^{ - 1 } Z^2 | { \hat p }' |^2 ) ( 1 + c_1^{ - 1 } Z^2 | { \hat q }' |^2 ) } \\
& \leq \frac{ 1 + c_1 Z^2 | { \hat p } |^2 }{ 1 + c_1^{ - 1 } Z^2 | { \hat p }' |^2 + c_1^{ - 1 } Z^2 | { \hat q }' |^2 } \\
& \leq \frac{ c_1^2 ( 1 + Z^2 | { \hat p } |^2 ) }{ 1 + Z^2 | { \hat p }' |^2 + Z^2 | { \hat q }' |^2 } , 
\end{align}
where we used the assumption $ c_1 \geq 1 $. Since $ \langle p \rangle = \sqrt{ 1 + | p |^2 } $, we only need to show that the following quantity is bounded:
\[
\frac{ 1 + Z^2 | { \hat p } |^2 }{ 1 + Z^2 | { \hat p }' |^2 + Z^2 | { \hat q }' |^2 } . 
\]
Now, we follow the proof of Lemma 4 of Ref.~\cite{LN172} to obtain the desired result. 
\end{proof}

\subsubsection{Local existence for the Boltzmann equation}
We assume that a metric $ g^{ a b } $ is given and satisfies the assumption \eqref{assumption1}. The existence can be obtained by using the following standard iteration:
\begin{align}\label{iteration}
\frac{ \partial f^{ n + 1 } }{ \partial t } + \frac{ 1 }{ p^0 } C_{ b c }^d p_d p^b \frac{ \partial f^{ n + 1 } }{ \partial p_c } = ( \det g )^{ - \frac12 } \iint \frac{ 1 }{ p^0 q^0 \sqrt{ s } } ( f^n ( p' ) f^n ( q' ) - f^{ n + 1 } ( p ) f^n ( q ) ) \, d \omega \, d q . 
\end{align}
Let $ { \mathcal I } $ be a multi-index with $ | { \mathcal I } | \leq N $. We take the derivative $ \partial_p^{ \mathcal I } $ of both sides of the above equation and multiply it by $ \partial_p^{ \mathcal I } f^{ n + 1 } \langle p \rangle^{ 2 m } e^{ p^0 } $. From the first term on the left hand side we obtain
\begin{align}
\frac12 \frac{ \partial }{ \partial t } \left( ( \partial^{ \mathcal I }_p f^{ n + 1 } )^2 \langle p \rangle^{ 2 m } e^{ p^0 } \right) + \frac12 ( \partial_p^{ \mathcal I } f^{ n + 1 } )^2 \langle p \rangle^{ 2 m } \frac{ k^{ a b } p_a p_b }{ p^0 } e^{ p^0 } , \label{localB0} 
\end{align}
where the second term above is non-negative, since we assumed that the second fundamental form $ k^{ a b } $ is positive definite in the sense of \eqref{assumption1}. Integrating the above quantity over $ \bbr^3_p $ we obtain 
\begin{align}
\frac12 \frac{ d }{ d t } \int_{ \bbr^3 } ( \partial^{ \mathcal I }_p f^{ n + 1 } )^2 \langle p \rangle^{ 2 m } e^{ p^0 } \, d p , \label{localB1}
\end{align}
where the second term of \eqref{localB0} has been ignored. From the second term on the left hand side of \eqref{iteration} we have
\begin{align}\label{iteration2}
\frac{ 1 }{ 2 p^0 } C^d_{ b c } p_d p^b \frac{ \partial ( \partial_p^{ \mathcal I } f^{ n + 1 } )^2 }{ \partial p_c } \langle p \rangle^{ 2 m } e^{ p^0 } + \sum_{ \substack{ { \mathcal J } + { \mathcal K } = { \mathcal I } \\ | { \mathcal J } | \geq 1 } } { { \mathcal I } \choose { \mathcal J } } \partial_p^{ \mathcal J } \left( \frac{ 1 }{ p^0 } C^d_{ b c } p_d p^b \right) \frac{ \partial ( \partial_p^{ \mathcal K } f^{ n + 1 } ) }{ \partial p_c } \partial_p^{ \mathcal I } f^{ n + 1 } \langle p \rangle^{ 2 m } e^{ p^0 } . 
\end{align}
Integrating the first term of \eqref{iteration2} by parts we obtain
\[
- \frac12 \int_{ \bbr^3 } \frac{ \partial }{ \partial p_c } \left( \frac{ 1 }{ p^0 } C^d_{ b c } p_d p^b \langle p \rangle^{ 2 m } e^{ p^0 } \right) ( \partial_p^{ \mathcal I } f^{ n + 1 } )^2 \, d p .
\]
Notice that the terms involving the derivatives of $ ( p^0 )^{ - 1 } $, $ p^b $ and $ e^{ p^0 } $ vanish due to the anti-symmetry of the structure constants, since 
\[
\frac{ \partial p^0 }{ \partial p_c } = \frac{ p^c }{ p^0 } , \qquad \frac{ \partial p_d }{ \partial p_c } = \delta^c_d , \qquad \frac{ \partial p^b }{ \partial p_c } = g^{ b c } , \qquad \frac{ \partial \langle p \rangle }{ \partial p_c } = \frac{ p_c }{ \langle p \rangle } . 
\]
Hence, we only need to consider the derivatives of $ p_d $ and $ \langle p \rangle $:
\begin{align}
& \left| \int_{ \bbr^3 } \frac{ 1 }{ p^0 } C^d_{ b c } p^b \left( \delta^c_d \langle p \rangle^{ 2 m } + 2 m p_d \langle p \rangle^{ 2m - 1 } \frac{ p_c }{ \langle p \rangle } \right) e^{ p^0 } ( \partial_p^{ \mathcal I } f^{ n + 1 } )^2 \, d p \right| , \nonumber \\
& \leq \int_{ \bbr^3 } | { \mathfrak L }^{ ( 1 ) } | \langle p \rangle^{ 2 m } e^{ p^0 } ( \partial_p^{ \mathcal I } f^{ n + 1 } )^2 \, d p \nonumber \\ 
& \leq Z^{ - 1 } \left( \sup_{ p , q } { \mathcal C } \right) \| f^{ n + 1 } ( t ) \|^2_{ g , m , N } , \label{localB2}
\end{align}
for some $ { \mathcal C } $. The second term of \eqref{iteration2} appears only for $ | { \mathcal I } | \geq 1 $. In this case, we need to estimate the following quantity for $ | { \mathcal J } | \geq 1 $:
\[
\partial_p^{ \mathcal J } \left( \frac{ 1 }{ p^0 } p_d p^b \right) = \sum_{ \substack{ { \mathcal J }_1 + { \mathcal J }_2 = { \mathcal J } \\ | { \mathcal J }_1 | \leq 1 } } { { \mathcal J } \choose { \mathcal J }_1 } \left( \partial^{ { \mathcal J }_1 }_p p_d \right) \partial_p^{ { \mathcal J }_2 } \left( \frac{ p^b }{ p^0 } \right) , 
\]
where we only need to consider the following two cases: $ | \partial^{ { \mathcal J }_1 }_p p_d | \leq | p | $ for $ { \mathcal J }_1 = 0 $ and $ | \partial^{ { \mathcal J }_1 }_p p_d | \leq C $ for $ | { \mathcal J }_1 | = 1 $. Hence, we apply Lemma \ref{lem p/p^0} to obtain 
\begin{align*}
\left| \partial_p^{ \mathcal J } \left( \frac{ 1 }{ p^0 } p_d p^b \right) \right| & \leq \frac{ | p | }{ ( p^0 )^{ | { \mathcal J } | } } | { \mathfrak L }^{ ( | { \mathcal J } | + 1 ) } | + \frac{ 1 }{ ( p^0 )^{ | { \mathcal J } | - 1 } } | { \mathfrak L }^{ ( | { \mathcal J } | ) } | \\
& \leq \frac{ c_1 }{ ( p^0 )^{ | { \mathcal J } | - 1 } } \left( Z | { \mathfrak L }^{ ( | \mathcal J | + 1 ) } | + | { \mathfrak L }^{ ( | { \mathcal J } | ) } | \right) \\
& \leq \frac{ Z^{ - | { \mathcal J } | } }{ ( p^0 )^{ | { \mathcal J } | - 1 } } { \mathcal C } , 
\end{align*}
where we used the assumption \eqref{assumption1} such that $ | p | \leq c_1 Z p^0 $, and the constant $ c_1 $ was absorbed into $ { \mathcal C } $. Now, we integrate the second term of \eqref{iteration2} to obtain the following estimate: 
\begin{align}
& \left| \int_{ \bbr^3 } \sum_{ \substack{ { \mathcal J } + { \mathcal K } = { \mathcal I } \\ | { \mathcal J } | \geq 1 } } { { \mathcal I } \choose { \mathcal J } } \partial_p^{ \mathcal J } \left( \frac{ 1 }{ p^0 } C^d_{ b c } p_d p^b \right) \frac{ \partial ( \partial_p^{ \mathcal K } f^{ n + 1 } ) }{ \partial p_c } \partial_p^{ \mathcal I } f^{ n + 1 } \langle p \rangle^{ 2 m } e^{ p^0 } \, d p \right| \nonumber \\
& \leq Z^{ - 1 } \left( \sup_{ p , q } { \mathcal C } \right) \| f^{ n + 1 } ( t ) \|^2_{ g , m , N } , \label{localB3}
\end{align}
where we used $ p^0 \geq 1 $ and $ Z \geq 1 $. Finally, we consider the collision term of \eqref{iteration}. We first consider the gain term: taking the derivative $ \partial^{ \mathcal I }_p $ of the gain term, we obtain 
\[
( \det g )^{ - \frac12 } \sum_{ { \mathcal J } + { \mathcal K } + { \mathcal L } = { \mathcal I } } { \mathcal I \choose { \mathcal J } + { \mathcal K } } { { \mathcal J } + { \mathcal K } \choose { \mathcal J } } \iint \partial^{ { \mathcal J } }_p \left( \frac{ 1 }{ p^0 q^0 \sqrt{ s } } \right) \partial^{ { \mathcal K } }_p \left( f^n ( p' ) \right) \partial^{ { \mathcal L } }_p \left( f^n ( q' ) \right) \, d \omega \, d q . 
\]
We notice that the following quantity can be estimated by using \eqref{dP} and \eqref{d1/pq}: 
\begin{align}
\left| \partial^{ { \mathcal J } }_p \left( \frac{ 1 }{ p^0 q^0 \sqrt{ s } } \right) \right| & \leq C \sum_{ { \mathcal J }_1 + { \mathcal J }_2 = { \mathcal J } } \left| \partial^{ { \mathcal J }_1 }_p \left( \frac{ 1 }{ p^0 q^0 } \right) \right| \left| \partial^{ { \mathcal J }_2 } \left( \frac{ 1 }{ \sqrt{ s } } \right) \right| \nonumber \\
& \leq \frac{ ( q^0 )^{ | { \mathcal J } | } }{ p^0 q^0 } | { \mathfrak L }^{ ( | { \mathcal J } | ) }_P | \nonumber \\
& \leq \frac{ ( q^0 )^{ | { \mathcal J } | } }{ p^0 q^0 } Z^{ - | { \mathcal J } | } { \mathcal C } . \label{est gain 1}
\end{align}
The quantities $ \partial^{ { \mathcal K } }_p \left( f^n ( p' ) \right) $ and $ \partial^{ { \mathcal L } }_p \left( f^n ( q' ) \right) $ can be estimated as in the proof of Lemma 10 of Ref.~\cite{LN172}. For $ | { \mathcal K } | = k  \geq 1 $, we apply the multivariate Faa di Bruno formula \cite{CS96} to write $ \partial^{ { \mathcal K } }_p \left( f^n ( p' ) \right) $ as a linear combination of the following quantities: 
\[
( \partial^{ { \mathcal K }_{ 1 } }_p f^n ) ( p' ) \prod_{ j = 1 }^s \left( \partial^{ { \mathcal K }_{ 2 j } }_p p' \right)^{ { \mathcal K }_{ 3 j } } , 
\]
where $ { \mathcal K }_{ 1 } $, $ { \mathcal K }_{ 2 j } $ and $ { \mathcal K }_{ 3 j } $ are some multi-indices satisfying 
\[
1 \leq | { \mathcal K }_1 | , | { \mathcal K }_{ 2 j } | , | { \mathcal K }_{ 3 j } | \leq k = | { \mathcal K } | , \qquad k = \sum_{ j = 1 }^s | { \mathcal K }_{ 2 j } | | { \mathcal K }_{ 3 j } | . 
\]
Since $ | { \mathcal K }_{ 2 j } | \geq 1 $, we can apply \eqref{dp'} to obtain 
\begin{align*}
\prod_{ j = 1 }^s \left| \left( \partial^{ { \mathcal K }_{ 2 j } }_p p' \right)^{ { \mathcal K }_{ 3 j } } \right| & \leq \prod_{ j = 1 }^s \left( ( q^0 )^{ | { \mathcal K }_{ 2 j } | + 8 } Z^{ - | { \mathcal K }_{ 2 j } | + 1 } { \mathcal C } \right)^{ | { \mathcal K }_{ 3 j } | } \leq ( q^0 )^{ 9 k } { \mathcal C } , 
\end{align*}
so that we have for any multi-index $ { \mathcal K } $ with $ | { \mathcal K } | = k \geq 1 $, 
\[
\left| \partial^{ { \mathcal K } }_p \left( f^n ( p' ) \right) \right| \leq ( q^0 )^{ 9 k } { \mathcal C } \sum_{ 1 \leq | { \mathcal K }_1 | \leq k } \left| ( \partial^{ { \mathcal K }_{ 1 } }_p f^n ) ( p' ) \right| . 
\]
We note that the estimate for the case $ { \mathcal K } = 0 $ is trivial. Hence, we have for any multi-index $ { \mathcal K } $ with $ | { \mathcal K } | = k \geq 0 $,
\begin{align}\label{est gain 2}
\left| \partial^{ { \mathcal K } }_p \left( f^n ( p' ) \right) \right| \leq ( q^0 )^{ 9 k } { \mathcal C } \sum_{ 0 \leq | { \mathcal K }_1 | \leq k } \left| ( \partial^{ { \mathcal K }_{ 1 } }_p f^n ) ( p' ) \right| . 
\end{align}
In a similar way, we obtain for any multi-index $ { \mathcal L } $ with $ | { \mathcal L } | = l \geq 0 $, 
\begin{align}\label{est gain 3}
\left| \partial^{ { \mathcal L } }_p \left( f^n ( q' ) \right) \right| \leq ( q^0 )^{ 9 l } { \mathcal C } \sum_{ 0 \leq | { \mathcal L }_1 | \leq l } \left| ( \partial^{ { \mathcal L }_{ 1 } }_p f^n ) ( p' ) \right| . 
\end{align}
We now combine the estimates \eqref{est gain 1}--\eqref{est gain 3} to obtain the following estimate for the gain term:
\begin{align*}
& \left| \int_{ \bbr^3 } \partial^{ { \mathcal I } }_p Q_+ ( f^n , f^n ) \partial^{ { \mathcal I } }_p f^{ n + 1 } \langle p \rangle^{ 2 m } e^{ p^0 } \, d p \right| \\
& \leq ( \det g )^{ - \frac12 } \sum \iiint \frac{ ( q^0 )^{ 9 N } }{ p^0 q^0 } { \mathcal C } \left| ( \partial^{ { \mathcal K }_{ 1 } }_p f^n ) ( p' ) \right| \left| ( \partial^{ { \mathcal L }_{ 1 } }_p f^n ) ( p' ) \right| \left| \partial^{ { \mathcal I } }_p f^{ n + 1 } ( p ) \right| \langle p \rangle^{ 2 m } e^{ p^0 } \, d \omega \, d q \, d p \\
& \leq ( \det g )^{ - \frac12 } \left( \sup_{ p , q } { \mathcal C } \right) \sum \iiint ( q^0 )^{ 9 N } \left| ( \partial^{ { \mathcal K }_{ 1 } }_p f^n ) ( p' ) \right| \left| ( \partial^{ { \mathcal L }_{ 1 } }_p f^n ) ( p' ) \right| \left| \partial^{ { \mathcal I } }_p f^{ n + 1 } ( p ) \right| \langle p \rangle^{ 2 m } e^{ p^0 } \, \frac{ d \omega \, d q \, d p }{ p^0 q^0 } , 
\end{align*}
where the summation is over all the possible $ { \mathcal K }_1 $ and $ { \mathcal L }_1 $ such that $ | { \mathcal K }_1 | + | { \mathcal L }_1 | \leq | { \mathcal I } | $. The above integral is estimated as follows: 
\begin{align*}
& \iiint ( q^0 )^{ 9 N } \left| ( \partial^{ { \mathcal K }_{ 1 } }_p f^n ) ( p' ) \right| \left| ( \partial^{ { \mathcal L }_{ 1 } }_p f^n ) ( p' ) \right| \left| \partial^{ { \mathcal I } }_p f^{ n + 1 } ( p ) \right| \langle p \rangle^{ 2 m } e^{ p^0 } \, \frac{ d \omega \, d q \, d p }{ p^0 q^0 } \\
& \leq C \iiint ( q^0 )^{ 9 N } \left| ( \partial^{ { \mathcal K }_{ 1 } }_p f^n ) ( p' ) \right| \langle p' \rangle^{ m } e^{ \frac12 p'^0 } \left| ( \partial^{ { \mathcal L }_{ 1 } }_p f^n ) ( p' ) \right| \langle q' \rangle^{ m } e^{ \frac12 q'^0 } \left| \partial^{ { \mathcal I } }_p f^{ n + 1 } ( p ) \right| \langle p \rangle^{ m } e^{ \frac12 p^0 } e^{ - \frac12 q^0 } \, \frac{ d \omega \, d q \, d p }{ p^0 q^0 } \\
& \leq C \left( \iiint \left| ( \partial^{ { \mathcal K }_{ 1 } }_p f^n ) ( p' ) \right|^2 \langle p' \rangle^{ 2 m } e^{ p'^0 } \left| ( \partial^{ { \mathcal L }_{ 1 } }_p f^n ) ( p' ) \right|^2 \langle q' \rangle^{ 2 m } e^{ q'^0 } \, \frac{ d \omega \, d q \, d p }{ p^0 q^0 } \right)^{ \frac12 } \\ 
& \qquad \times \left( \iiint ( q^0 )^{ 18 N } \left| \partial^{ { \mathcal I } }_p f^{ n + 1 } ( p ) \right|^2 \langle p \rangle^{ 2 m } e^{ p^0 } e^{ - q^0 } \, \frac{ d \omega \, d q \, d p }{ p^0 q^0 } \right)^{ \frac12 } \\
& \leq C \| f^n ( t ) \|^2_{ g , m , N } \| f^{ n + 1 } ( t ) \|_{ g , m , N } \left( \int ( q^0 )^{ 18 N } e^{ - q^0 } \, d q \right)^{ \frac12 } , 
\end{align*}
where we used Lemma \ref{lem pp'q'} and the following change of variables \cite{LL22}:
\begin{align}
\frac{ d p' \, d q' }{ p'^0 q'^0 } = \frac{ d p \, d q }{ p^0 q^0 } . \label{change}
\end{align}
By a direct calculation we have 
\[
\left( \int ( q^0 )^{ 18 N } e^{ - q^0 } \, d q \right)^{ \frac12 } \leq C ( \det g )^{ \frac14 } , 
\]
so that we obtain for the gain term the following estimate: 
\begin{align}
\left| \int_{ \bbr^3 } \partial^{ { \mathcal I } }_p Q_+ ( f^n , f^n ) \partial^{ { \mathcal I } }_p f^{ n + 1 } \langle p \rangle^{ 2 m } e^{ p^0 } \, d p \right| \leq ( \det g )^{ - \frac14 } \left( \sup_{ p , q } { \mathcal C } \right) \| f^n ( t ) \|^2_{ g , m , N } \| f^{ n + 1 } ( t ) \|_{ g , m , N } , \label{localB4} 
\end{align}
where the constant $ C $ has been absorbed into $ { \mathcal C } $. In a similar way we obtain the following estimate for the loss term: 
\begin{align}
\left| \int_{ \bbr^3 } \partial^{ { \mathcal I } }_p Q_- ( f^{ n + 1 } , f^n ) \partial^{ { \mathcal I } }_p f^{ n + 1 } \langle p \rangle^{ 2 m } e^{ p^0 } \, d p \right| \leq ( \det g )^{ - \frac14 } \left( \sup_{ p , q } { \mathcal C } \right) \| f^n ( t ) \|_{ g , m , N } \| f^{ n + 1 } ( t ) \|^2_{ g , m , N } . \label{localB5}
\end{align}
We now combine the estimates \eqref{localB1}, \eqref{localB2}, \eqref{localB3}, \eqref{localB4} and \eqref{localB5} to obtain 
\begin{align}
& \frac{ d }{ d t } \| f^{ n + 1 } ( t ) \|^2_{ g , m , N } \nonumber \\ 
& \leq Z^{ - 1 } \left( \sup_{ p , q } { \mathcal C }_1 \right) \| f^{ n + 1 } ( t ) \|^2_{ g , m , N } + ( \det g )^{ - \frac14 } \left( \sup_{ p , q } { \mathcal C }_2 \right) \| f^n ( t ) \|^2_{ g , m , N } \| f^{ n + 1 } ( t ) \|_{ g , m , N } \nonumber \\
& \quad + ( \det g )^{ - \frac14 } \left( \sup_{ p , q } { \mathcal C }_3 \right) \| f^n ( t ) \|_{ g , m , N } \| f^{ n + 1 } ( t ) \|^2_{ g , m , N } , \label{iterationB}
\end{align}
for some $ { \mathcal C }_1 $, $ { \mathcal C }_2 $ and $ { \mathcal C }_3 $. We now obtain the local existence:

\begin{prop}\label{prop Boltzmann}
Suppose that a metric $ g^{ a b } $ is given and satisfies \eqref{assumption1}. Let $ f_0 $ be an initial data of the Boltzmann equation \eqref{Boltzmann} such that 
\[
\| f_0 \|^2_{ g , m + \frac12 , N } \leq M , 
\]
for some $ M > 0 $, with $ m \geq 0 $ and $ N \geq 3 $. Then, there exists $ t_0 < T_1 \leq T $ such that the Boltzmann equation admits a unique non-negative classical solution $ f $ on $ [ t_0 , T_1 ] $ satisfying 
\[
\sup_{ t_0 \leq \tau \leq T_1 } \| f ( \tau ) \|^2_{ g , m , N } \leq 2 M . 
\]
\end{prop}
\begin{proof}
The proposition is proved by using the iteration described above. Let us first consider the quantities $ { \mathcal C_1 } $, $ { \mathcal C_2 } $ and $ { \mathcal C_3 } $ in \eqref{iterationB}. Note that the assumptions on $ g^{ a b } $ implies that 
\[
\frac{ | p^a | }{ p^0 } \leq \left( \max_{ a , b } | g^{ a b } | \right) \frac{ | p | }{ p^0 } \leq c_1^{ 3 / 2 } Z^{ - 1 } ,  
\]
where we used $ | p | \leq \sqrt{ c_1 } Z p^0 $, so that for any $ 2 k + l + m = n $, we have 
\[
\left| g^{ a_1 b_1 } g^{ a_2 b_2 } \cdots g^{ a_k b_k } \frac{ p^{ c_1 } }{ p^0 } \frac{ p^{ c_2 } }{ p^0 } \cdots \frac{ p^{ c_l } }{ p^0 } \frac{ q^{ d_1 } }{ q^0 } \frac{ q^{ d_2 } }{ q^0 } \cdots \frac{ q^{ d_m } }{ q^0 } \right| \leq \left( c_1 Z^{ - 2 } \right)^k \left( c_1^{ 3 / 2 } Z^{ - 1 } \right)^{ l + m } \leq C Z^{ - n } , 
\]
which shows that $ | { \mathfrak L }^{ ( n ) }_P | \leq C Z^{ - n } $. Hence, we can conclude that the quantities $ { \mathcal C_1 } $, $ { \mathcal C_2 } $ and $ { \mathcal C_3 } $ are bounded uniformly with respect to $ p , q \in \bbr^3 $ and $ t \in [ t_0 , T ] $ such that 
\[
\sup_{ p , q , t } { \mathcal C }_j \leq C , \qquad j = 1 , 2 , 3 , 
\]
where $ C $ depends on $ c_1 $. Note that $ ( \det g )^{ - 1 } $ is the determinant of the matrix $ ( g^{ a b } )_{ a , b = 1 , 2 , 3 } $. Hence, we obtain by the assumptions on $ g^{ a b } $, 
\[
( \det g )^{ - \frac14 } \leq C Z^{ - \frac32 } , 
\]
which holds on $ [ t_0 , T ] $. Applying the above estimates to \eqref{iterationB}, and integrating it with respect to $ t $, we obtain 
\begin{align*}
\| f^{ n + 1 } ( t ) \|^2_{ g , m , N } & \leq M + C \sup_{ t_0 \leq \tau \leq t } \| f^{ n + 1 } ( \tau ) \|^2_{ g , m , N } \int_{ t_0 }^t Z^{ - 1 } ( \tau ) \, d \tau \\
& \quad + C \sup_{ t_0 \leq \tau \leq t } \| f^n ( \tau ) \|^2_{ g , m , N } \sup_{ t_0 \leq \tau \leq t } \| f^{ n + 1 } ( \tau ) \|_{ g , m , N } \int_{ t_0 }^t Z^{ - \frac32 } ( \tau ) \, d \tau \\ 
& \quad + C \sup_{ t_0 \leq \tau \leq t } \| f^n ( \tau ) \|_{ g , m , N } \sup_{ t_0 \leq \tau \leq t } \| f^{ n + 1 } ( \tau ) \|^2_{ g , m , N } \int_{ t_0 }^t Z^{ - \frac32 } ( \tau ) \, d \tau . 
\end{align*}
Since $ Z $ is continuous, we can find a time interval $ [ t_0 , T_1 ] \subset [ t_0 , T ] $ such that 
\[
\sup_{ t_0 \leq \tau \leq T_1 } \| f^{ n } ( \tau ) \|_{ g , m , N }^2 \leq 2 M , \qquad n = 1 , 2 , \cdots , 
\]
where $ T_1 $ is independent of $ n $. Now, by the standard arguments (see Refs.~\cite{G03, LL22}), we obtain the local existence of classical solutions. To ensure the differentiability with respect to $ t $, we need to estimate the second term of \eqref{localB0}, which had been ignored, as follows: 
\begin{align}
\frac12 \int_{ \bbr^3 } ( \partial_p^{ \mathcal I } f )^2 \langle p \rangle^{ 2 m } \frac{ k^{ a b } p_a p_b }{ p^0 } e^{ p^0 } \, d p & \leq C Z^{ - 2 } \int_{ \bbr^3 } ( \partial_p^{ \mathcal I } f )^2 \langle p \rangle^{ 2 m } \frac{ | p |^2 }{ p^0 } e^{ p^0 } \, d p \nonumber \\ 
& \leq C Z^{ - 1 } \int_{ \bbr^3 } ( \partial_p^{ \mathcal I } f )^2 \langle p \rangle^{ 2 m } | p | e^{ p^0 } \, d p \nonumber \\ 
& \leq C Z^{ - 1 } \| f ( t ) \|^2_{ g , m + \frac12 , N } , \label{diff}
\end{align}
where we used the assumption \eqref{assumption1}. Then, we obtain the following estimate:
\begin{align*}
\left| \frac{ d }{ d t } \| f ( t ) \|^2_{ g , m , N } \right| & \leq C Z^{ - 1 } \| f ( t ) \|^2_{ g , m + \frac12 , N } + Z^{ - 1 } \left( \sup_{ p , q } { \mathcal C }_1 \right) \| f ( t ) \|^2_{ g , m , N } \nonumber \\ 
& \quad + ( \det g )^{ - \frac14 } \left( \sup_{ p , q } { \mathcal C }_2 + \sup_{ p , q } { \mathcal C }_3 \right) \| f ( t ) \|^3_{ g , m , N } , 
\end{align*}
which is the estimate \eqref{iterationB}, where $ f^n $ and $ f^{ n + 1 } $ are replaced by $ f $, with the term \eqref{diff} taken into account. This ensures the differentiability with respect to $ t $, provided that the initial data $ \| f_0 \|^2_{ g , m + \frac12 , N } $ is finite. We also need to prove the uniqueness, but it will be proved for the coupled EBs system in Section \ref{sec uniqueness}. This completes the proof of Proposition \ref{prop Boltzmann}. 
\end{proof}

\subsection{Local existence for the EBs system}\label{sec local}
In this part we study the local existence for the coupled EBs system. The evolution equations for $ g_{ a b } $ and $ k_{ a b } $ are given by \eqref{EBs1} and \eqref{EBs2}, but for technical reasons we will consider the equations for $ g^{ a b } $ and $ k^{ a b } $:
\begin{align}
\frac{ d g^{ a b } }{ d t } & = - 2 k^{ a b } , \label{EBs^1} \\
\frac{ d k^{ a b } }{ d t } & = - 2 k^a_c k^{ b c } - k k^{ a b } - R^{ a b } + S^{ a b } + \frac12 ( \rho - S ) g^{ a b } + V ( \phi ) g^{ a b } . \label{EBs^2} 
\end{align}
We note that the quantities $ k_a^c $, $ k $ and $ R^{ a b } $ on the right hand side of the second equation are some polynomials of $ g_{ a b } $, $ k_{ a b } $, $ g^{ a b } $ and $ k^{ a b } $, or equivalently some polynomials of $ g^{ a b } $, $ k^{ a b } $ and $ \det g $, which is the determinant of the matrix $ g_{ a b } $. The equations for the scalar field \eqref{EBs3} and \eqref{EBs4} are 
\begin{align*}
\frac{ d \phi }{ d t } & = \psi , \\
\frac{ d \psi }{ d t } & = - k \psi - V ' ( \phi ) . 
\end{align*}
We recall that the potential $ V ( \phi ) $ is a smooth function, and the matter terms are given by $ S^{ a b } = T^{ a b } $, $ \rho = T^{ 0 0 } $ and $ S = T_{ a b } g^{ a b } $, so that we have by using \eqref{T_ab}:
\begin{align*}
S^{ a b } & = ( \det g )^{ - \frac12 } \int_{ \bbr^3 } f p^a p^b \frac{ d p }{ p^0 } , \\
\rho & = ( \det g )^{ - \frac12 } \int_{ \bbr^3 } f p^0 \, d p , \\
S & = ( \det g )^{ - \frac12 } \int_{ \bbr^3 } f p_a p_b g^{ a b } \frac{ d p }{ p^0 } .
\end{align*}
We observe that the matter terms are also smooth functions of $ g_{ a b } $ and $ g^{ a b } $ provided that the distribution function $ f $ is given appropriately. 

Now, we define the iteration. Let $ { g_0 }^{ a b } $, $ { k_0 }^{ a b } $, $ \phi_0 $, $ \psi_0 $ and $ f_0 $ be a set of initial data of the EBs system. Define
\begin{gather*}
{ g_{ 0 } }^{ a b } ( t ) = { g_0 }^{ a b } , \qquad { k_{ 0 } }^{ a b } ( t ) = { k_0 }^{ a b } , \qquad \phi_0 ( t ) = \phi_0 , \qquad \psi_0 ( t ) = \psi_0 , \qquad f_0 ( t , p ) = f_0 ( p ) , 
\end{gather*}
by abuse of the notation, and suppose that $ { g_j }^{ a b } $, $ { k_j }^{ a b } $, $ \phi_j $, $ \psi_j $ and $ f_j $ are given for some $ j \geq 0 $. We first consider the scalar field equations to obtain $ \phi_{ j + 1 } $ and $ \psi_{ j + 1 } $: 
\begin{align}
\frac{ d \phi_{ j + 1 } }{ d t } & = \psi_j , \label{phi_j} \\
\frac{ d \psi_{ j + 1 } }{ d t } & = - k_j \psi_j - V ' ( \phi_j ) , \label{psi_j}
\end{align}
where $ k_j = { k_j }^{ a b } { g_j }_{ a b } $ and $ { g_j }_{ a b } $ is defined by the inverse of $ { g_j }^{ a b } $. For the Boltzmann equation, we solve the following equation to obtain $ f_{ j + 1 } $:
\begin{align}
\frac{ \partial f_{ j + 1 } }{ \partial t } + \frac{ 1 }{ { p_j }^0 } C^d_{ b c } p_d { p_j }^b \frac{ \partial f_{ j + 1 } }{ \partial p_c } = Q ( f_{ j + 1 } , f_{ j + 1 } ) , \label{f_j}
\end{align}
where $ { p_j }^0 $ and $ { p_j }^b $ are defined by 
\[
{ p_j }^0 = \sqrt{ 1 + { g_j }^{ a b } p_a p_b } , \qquad { p_j }^b = { g_j }^{ a b } p_a , 
\]
and the collision term is given by 
\begin{align*}
Q ( f_{ j + 1 } , f_{ j + 1 } ) & = ( \det g_j )^{ - \frac12 } \iint \frac{ 1 }{ { p_j }^0 { q_j }^0 \sqrt{ s_j } } ( f_{ j + 1 } ( { p_j }' ) f_{ j + 1 } ( { q_j }' ) - f_{ j + 1 } ( p ) f_{ j + 1 } ( q ) ) \, d \omega \, d q , 
\end{align*}
where $ \det g_j $ denotes the determinant of the $ 3 \times 3 $ matrix $ { g_j }_{ a b } $, and $ { q_j }^0 $, $ s_j $, $ { p_j }' $ and $ { q_j }' $ are defined by 
\[
{ q_j }^0 = \sqrt{ 1 + { g_j }^{ a b } q_a q_b } , \qquad s_j = ( { p_{ j } }^0 + { q_j }^0 )^2 - { g_j }^{ a b } ( p_a + q_a ) ( p_b + q_b ) , 
\]
and
\begin{align*}
{ p_j }'_d & = p_d + 2 \bigg( - { q_j }^0 \frac{ { n_j }^a \omega_i { { e_j }^i }_a }{ \sqrt{ s_j } } + { q_j }^a \omega_i { { e_j }^i }_a + \frac{ { n_j }^a \omega_i { { e_j }^i }_a n_b { q_j }^b }{ \sqrt{ s_{ j } } ( { n_{ j } }^0 + \sqrt{ s_j } ) } \bigg) \bigg( \omega_k { { e_j }^k }_d + \frac{ { n_j }^c \omega_k { { e_j }^k }_c n_d }{ \sqrt{ s_j } ( { n_j }^0 + \sqrt{ s_j } ) } \bigg) , \\
{ q_j }'_d & = q_d - 2 \bigg( - { q_j }^0 \frac{ { n_j }^a \omega_i { { e_j }^i }_a }{ \sqrt{ s_j } } + { q_j }^a \omega_i { { e_j }^i }_a + \frac{ { n_j }^a \omega_i { { e_j }^i }_a n_b { q_j }^b }{ \sqrt{ s_{ j } } ( { n_{ j } }^0 + \sqrt{ s_j } ) } \bigg) \bigg( \omega_k { { e_j }^k }_d + \frac{ { n_j }^c \omega_k { { e_j }^k }_c n_d }{ \sqrt{ s_j } ( { n_j }^0 + \sqrt{ s_j } ) } \bigg) . 
\end{align*}
In other words, we consider the Boltzmann equation with the $ j $-th metric $ { g_j }^{ a b } $ to obtain $ f_{ j + 1 } $ with initial data $ f_0 $. Now, we define $ { S_{ j + 1 } }^{ a b } $, $ \rho_{ j + 1 } $ and $ S_{ j + 1 } $ by 
\begin{align*}
{ S_{ j + 1 } }^{ a b } & = ( \det g_j )^{ - \frac12 } \int_{ \bbr^3 } f_{ j + 1 } { p_j }^a { p_j }^b \frac{ d p }{ { p_j }^0 } , \\
\rho_{ j + 1 } & = ( \det g_j )^{ - \frac12 } \int_{ \bbr^3 } f_{ j + 1 } { p_j }^0 \, d p , \\
S_{ j + 1 } & = ( \det g_j )^{ - \frac12 } \int_{ \bbr^3 } f_{ j + 1 } p_a p_b { g_j }^{ a b } \frac{ d p }{ { p_j }^0 } . 
\end{align*}
For the Einstein equations, we consider the following equations to obtain $ { g_{ j + 1 } }^{ a b } $ and $ { k_{ j + 1 } }^{ a b } $: 
\begin{align}
\frac{ d { g_{ j + 1 } }^{ a b } }{ d t } & = - 2 { k_{ j } }^{ a b } , \label{g_j} \\
\frac{ d { k_{ j + 1 } }^{ a b } }{ d t } & = - 2 { k_j }^a_c { k_j }^{ b c } - k_j { k_j }^{ a b } - { R_j }^{ a b } + { S_{ j + 1 } }^{ a b } + \frac12 ( \rho_{ j + 1 } - S_{ j + 1 } ) { g_j }^{ a b } + V ( \phi_{ j + 1 } ) { g_j }^{ a b } , \label{k_j} 
\end{align}
where $ { k_j }^a_c = { k_j }^{ a d } { g_j }_{ c d } $ and $ { R_j }^{ a b } $ is given by a polynomial of $ { g_j }_{ a b } $ and $ { g_j }^{ a b } $. This completes the iteration, and the local existence will be obtained by using the iteration described above.

In Proposition \ref{prop local} below, we prove the existence and uniqueness of local solutions to the EBs system. To prove the uniqueness, we will need the following lemma, which shows that determinants of positive definite matrices are, in some sense, concave.

\begin{lemma}\label{lem HJ}
Let $ A $ and $ B $ be positive definite $ n \times n $ matrices. For any $ 0 < \alpha < 1 $, we have 
\[
\det ( \alpha A + ( 1 - \alpha ) B ) \geq ( \det A )^\alpha ( \det B )^{ 1 - \alpha } , 
\]
where the equality holds, if and only if $ A = B $. 
\end{lemma}
\begin{proof}
We refer to Corollary 7.6.8.\ of Ref.~\cite{HJ} for the proof. 
\end{proof}

\begin{prop}\label{prop local}
Let $ { g_0 }^{ a b } $, $ { k_0 }^{ a b } $, $ \phi_0 $, $ \psi_0 $ and $ f_0 $ be a set of initial data of the EBs system \eqref{EBs1}--\eqref{EBs7}. Suppose that $ { g_0 }^{ a b } $ and $ { k_0 }^{ a b } $ are positive definite and that there exist real numbers $ M_g $, $ M_\phi $ and $ M_f $ satisfying 
\begin{align*}
\max_{ a , b } | { g_0 }^{ a b } | + \max_{ a , b } | { k_0 }^{ a b } | \leq M_g , \qquad | \phi_0 | + | \psi_0 | \leq M_\phi , \qquad \| f_0 \|^2_{ g_0 , m + \frac12 , N } \leq M_f , 
\end{align*}
with $ m > 7 / 2 $ and $ N \geq 3 $. Then, there exists $ T > t_0 $ such that the EBs system admits a unique classical solution $ g^{ a b } $, $ k^{ a b } $, $ \phi $, $ \psi $ and $ f $ on $ [ t_0 , T ] $ satisfying 
\begin{align}
\sup_{ t_0 \leq t \leq T } \max_{ a , b } | { g }^{ a b } ( t ) | + \sup_{ t_0 \leq t \leq T } \max_{ a , b } | { k }^{ a b } ( t ) | \leq 2 M_g , \label{local 1} \\ 
\sup_{ t_0 \leq t \leq T } | \phi ( t ) | + \sup_{ t_0 \leq t \leq T } | \psi ( t ) | \leq 2 M_\phi , \label{local 2} \\ 
\sup_{ t_0 \leq t \leq T } \| f ( t ) \|^2_{ g , m , N } \leq 2 M_f . \label{local 3} 
\end{align}
Moreover, there exists $ c_1 \geq 1 $ such that $ g^{ a b } $ and $ k^{ a b } $ satisfy 
\begin{align}
\frac{ 1 }{ c_1 } | p |^2 \leq Z^2 ( t ) { g }^{ a b } ( t ) p_a p_b \leq c_1 | p |^2 , \qquad & \frac{ 1 }{ c_1 } | p |^2 \leq Z^2 ( t ) { k }^{ a b } ( t ) p_a p_b \leq c_1 | p |^2 , \label{local 4} \\ 
\max_{ a , b } | Z^2 ( t ) { g }^{ a b } ( t ) | \leq c_1 , \qquad & Z^6 ( t ) \det g^{ - 1 } ( t ) \geq \frac{ 1 }{ c_1 } , \label{local 5} 
\end{align}
for any $ p \in \bbr^3 $ on $ [ t_0 , T ] $.
\end{prop}

The proof of Proposition \ref{prop local} will be given  separately in Sections \ref{sec existence} and \ref{sec uniqueness}, since it requires lengthy calculations. The proof of the existence will be given in Section \ref{sec existence}, and the proof of uniqueness will be given in Section \ref{sec uniqueness}.

\subsubsection{Proof of Proposition \ref{prop local} (existence)}\label{sec existence}
Let us consider the iteration described above. Below, we will write 
\begin{align*}
X_0 & = \max_{ a , b } | { g_0 }^{ a b } | + \max_{ a , b } | { k_0 }^{ a b } | , \\ 
Y_0 & = | \phi_0 | + | \psi_0 | , 
\end{align*}
and 
\begin{align*}
X_j ( t ) & = \sup_{ t_0 \leq \tau \leq t } \max_{ a , b } | { g_j }^{ a b } ( \tau ) | + \sup_{ t_0 \leq \tau \leq t } \max_{ a , b } | { k_j }^{ a b } ( \tau ) | , \\ 
Y_j ( t ) & = \sup_{ t_0 \leq \tau \leq t } | \phi_j ( \tau ) | + \sup_{ t_0 \leq \tau \leq t } | \psi_j ( \tau ) | , 
\end{align*}
for $ j \geq 0 $. We notice that $ X_j ( t_0 ) = X_0 $ and $ Y_j ( t_0 ) = Y_0 $ for all $ j \geq 0 $. Let $ M_g $, $ M_\phi $ and $ M_f $ be real numbers such that 
\[
X_0 \leq M_g , \qquad Y_0 \leq M_\phi , \qquad \| f_0 \|^2_{ g_0 , m + \frac12 , N } \leq M_f . 
\]
Since $ { g_0 }^{ a b } $ and $ { k_0 }^{ a b } $ are positive definite, we can find a real number $ c_0 \geq 1 $ satisfying 
\begin{align}
\frac{ 1 }{ c_0 } | p |^2 \leq Z^2 ( t_0 ) { g_0 }^{ a b } p_a p_b \leq c_0 | p |^2 , \qquad & \frac{ 1 }{ c_0 } | p |^2 \leq Z^2 ( t_0 ) { k_0 }^{ a b } p_a p_b \leq c_0 | p |^2 , \label{c est1} \\
\max_{ a , b } | Z^2 ( t_0 ) { g_0 }^{ a b } | \leq c_0 , \qquad & Z^6 ( t_0 ) \det g_0^{ - 1 } \geq \frac{ 1 }{ c_0 } , \label{c est2}
\end{align}
for any $ p \in \bbr^3 $. Now, suppose that there exists an interval $ [ t_0 , T ] $ on which we have 
\[
X_j ( t ) \leq 2 M_g , \qquad Y_j ( t ) \leq 2 M_\phi , \qquad \sup_{ t_0 \leq \tau \leq t } \| f_j ( \tau ) \|^2_{ g_{ j - 1 } , m , N } \leq 2 M_f , 
\]
where $ \| f_0 ( \tau ) \|^2_{ g_{ - 1 } , m , N } $ is to be understood as $ \| f_0 ( \tau ) \|^2_{ g_0 , m , N } $, and 
\begin{align*}
\frac{ 1 }{ 2 c_0 } | p |^2 \leq Z^2 ( t ) { g_j }^{ a b } ( t ) p_a p_b \leq 2 c_0 | p |^2 , \qquad & \frac{ 1 }{ 2 c_0 } | p |^2 \leq Z^2 ( t ) { k_j }^{ a b } ( t ) p_a p_b \leq 2 c_0 | p |^2 , \\ 
\max_{ a , b } | Z^2 ( t ) { g_j }^{ a b } ( t ) | \leq 2 c_0 , \qquad & Z^6 ( t ) \det g_j^{ - 1 } ( t ) \geq \frac{ 1 }{ 2 c_0 } , 
\end{align*}
for any $ p \in \bbr^3 $. Here, $ \det g_0^{ - 1 } $ and $ \det g_j^{ - 1 } ( t ) $ denote the determinants of the matrices $ { g_0 }^{ a b } $ and $ { g_j }^{ a b } ( t ) $, respectively. Below, $ C_M $ will denote a positive constant which may depend on $ M_g $, $ M_\phi $, $ M_f $ or $ c_0 $, but does not depend on $ j $. 

We first consider the equations \eqref{phi_j} and \eqref{psi_j} for the scalar field, which are easily estimated as follows: 
\begin{align*}
| \phi_{ j + 1 } ( t ) | & \leq | \phi_0  | + C ( t - t_0 ) Y_j ( t ) , \\
| \psi_{ j + 1 } ( t ) | & \leq | \psi_0 | + C ( t - t_0 ) \left( c_0 X_j^3 ( t ) Y_j ( t ) + \sup_{ | x | \leq Y_j ( t ) } | V' ( x ) | \right) , 
\end{align*}
where we used $ k_j = { k_j }^{ a b } { g_j }_{ a b } $ and $ { g_j }_{ a b } $ is a quadratic polynomial of $ { g_j }^{ a b } $ multiplied by $ \det g_j $, which is the determinant of the matrix $ { g_j }_{ a b } $. Since $ V' $ is a smooth function, we obtain 
\begin{align*}
Y_{ j + 1 } ( t ) \leq M_\phi + C_M ( t - t_0 ) . 
\end{align*}
Hence, we find an interval $ [ t_0 , T_1 ] \subset [ t_0 , T ] $ on which we have 
\begin{align}
Y_{ j + 1 } ( t ) \leq 2 M_\phi . \label{phi_j est} 
\end{align}
For the Boltzmann equation, we apply Proposition \ref{prop Boltzmann} to obtain $ f_{ j + 1 } $. We note that \eqref{f_j} is the equation of $ f_{ j + 1 } $ with the metric $ { g_j }^{ a b } $, which satisfies the conditions of Proposition \ref{prop Boltzmann}, so that we obtain a time interval $ [ t_0 , T_2 ] \subset [ t_0 , T_1 ] $ on which $ f_{ j + 1 } $ satisfies 
\begin{align}
\sup_{ t_0 \leq \tau \leq T_2 } \| f_{ j + 1 } ( \tau ) \|^2_{ g_j , m , N } \leq 2 M_f . \label{f_j est}
\end{align}
For the estimates of $ { g_{ j + 1 } }^{ a b } $ and $ { k_{ j + 1 } }^{ a b } $, we need to consider the quantities $ { S_{ j + 1 } }^{ a b } $, $ \rho_{ j + 1 } $ and $ S_{ j + 1 } $ in \eqref{g_j} and \eqref{k_j}: 
\begin{align*}
{ S_{ j + 1 } }^{ a b } & = ( \det g_{ j } )^{ - \frac12 } \int_{ \bbr^3 } f_{ j + 1 } { p_{ j } }^a { p_{ j } }^b \frac{ d p }{ { p_{ j } }^0 } \\
& = ( \det g_{ j }^{ - 1 } )^{ \frac12 } { g_{ j } }^{ a c } { g_{ j } }^{ b d } \int_{ \bbr^3 } f_{ j + 1 } p_c p_d \frac{ d p }{ { p_{ j } }^0 } , 
\end{align*}
which can be estimated, since $ \det g_{ j }^{ - 1 } $ is a cubic polynomial of $ { g_{ j } }^{ a b } $, as follows: 
\begin{align}
| { S_{ j + 1 } }^{ a b } | & \leq C X_{ j }^{ \frac{ 7 }{ 2 } } \left( \int_{ \bbr^3 } | f_{ j + 1 } |^2 \langle p \rangle^{ 2 m } e^{ { p_{ j } }^0 } \, dp \right)^{ \frac12 } \left( \int_{ \bbr^3 } \langle p \rangle^{ 4 - 2 m } e^{ - { p_{ j } }^0 } \frac{ d p }{ ( { p_{ j } }^0 )^2 } \right)^{ \frac12 } \nonumber \\
& \leq C X_{ j }^{ \frac{ 7 }{ 2 } } \| f_{ j + 1 } \|_{ g_j , m , N } , \label{matter est} 
\end{align}
since $ m > 7 / 2 $. The quantities $ \rho_{ j + 1 } $ and $ S_{ j + 1 } $ are similarly estimated, and we obtain by \eqref{f_j est}: 
\begin{align}
| { S_{ j + 1 } }^{ a b } | , | \rho_{ j + 1 } | , | S_{ j + 1 } | \leq C_M . \label{S_j est}
\end{align}
Since we have from \eqref{g_j} and \eqref{k_j} 
\begin{align}
{ g_{ j + 1 } }^{ a b } ( t ) & = { g_0 }^{ a b } - 2 \int_{ t_0 }^t { k_{ j } }^{ a b } ( \tau ) \, d \tau , \label{g_j int} \\
{ k_{ j + 1 } }^{ a b } ( t ) & = { k_0 }^{ a b } + \int_{ t_0 }^t - 2 { k_j }^a_c ( \tau ) { k_j }^{ b c } ( \tau ) - k_j ( \tau ) { k_j }^{ a b } ( \tau ) - { R_j }^{ a b } ( \tau ) \, d \tau \nonumber \\
& \quad + \int_{ t_0 }^t { S_{ j + 1 } }^{ a b } ( \tau ) + \frac12 ( \rho_{ j + 1 } - S_{ j + 1 } ) { g_j }^{ a b } ( \tau ) + V ( \phi_{ j + 1 } ) { g_j }^{ a b } ( \tau ) \, d \tau , \label{k_j int} 
\end{align}
we obtain by applying \eqref{phi_j est} and \eqref{S_j est} a time interval $ [ t_0 , T_3 ] \subset [ t_0 , T_2 ] $ on which 
\begin{align}
X_{ j + 1 } ( t ) \leq 2 M_g . \label{g_j est} 
\end{align}
Now, concerning the positive definiteness of $ { g_{ j + 1 } }^{ a b } $, we notice from \eqref{g_j int} that 
\begin{align*}
{ g_0 }^{ a b } p_a p_b - C_M ( t - t_0 ) | p |^2 \leq { g_{ j + 1 } }^{ a b } ( t ) p_a p_b \leq { g_0 }^{ a b } p_a p_b + C_M ( t - t_0 ) | p |^2 . 
\end{align*}
Applying the initial condition for $ { g_0 }^{ a b } p_a p_b $ in \eqref{c est1}, we have 
\[
\left( \frac{ Z^2 ( t ) }{ c_0 Z^2 ( t_0 ) } - C_M Z^2 ( t ) ( t - t_0 ) \right) | p |^2 \leq Z^2 ( t ) { g_{ j + 1 } }^{ a b } ( t ) p_a p_b \leq \left( \frac{ c_0 Z^2 ( t ) }{ Z^2 ( t_0 ) } + C_M Z^2 ( t ) ( t - t_0 ) \right) | p |^2 . 
\]
Since $ Z $ is continuous, we can find $ [ t_0 , T_4 ] \subset [ t_0 , T_3 ] $ on which 
\[
\frac{ 1 }{ 2 c_0 } | p |^2 \leq Z^2 ( t ) { g_{ j + 1 } }^{ a b } ( t ) p_a p_b \leq 2 c_0 | p |^2 . 
\]
In a similar way we can find $ [ t_0 , T_5 ] \subset [ t_0 , T_4 ] $ on which 
\[
\frac{ 1 }{ 2 c_0 } | p |^2 \leq Z^2 ( t ) { k_{ j + 1 } }^{ a b } ( t ) p_a p_b \leq 2 c_0 | p |^2 , \qquad \max_{ a , b } | Z^2 ( t ) { g_{ j + 1 } }^{ a b } ( t ) | \leq 2 c_0 . 
\]
For the quantity $ \det g_{ j + 1 }^{ - 1 } $ we notice that 
\[
\frac{ d \det g_{ j + 1 }^{ - 1 } }{ d t } = \left( { g_{ j + 1 } }_{ a b } \frac{ d { g_{ j + 1 } }^{ a b } }{ d t } \right) \det g_{ j + 1 }^{ - 1 } = - 2 { g_{ j + 1 } }_{ a b } { k_j }^{ a b } \det g_{ j + 1 }^{ - 1 } , 
\]
which implies that 
\[
\det g_{ j + 1 }^{ - 1 } ( t ) = \det g_0^{ - 1 } \exp \left( - 2 \int_{ t_0 }^t { g_{ j + 1 } }_{ a b } ( \tau ) { k_j }^{ a b } ( \tau ) \, d \tau \right) . 
\]
Applying the initial condition for $ \det g_0^{ - 1 } $ in \eqref{c est2} we finally obtain an interval $ [ t_0 , T_6 ] \subset [ t_0 , T_5 ] $ on which 
\[
Z^6 ( t ) \det g_{ j + 1 }^{ - 1 } ( t ) \geq \frac{ Z^6 ( t ) }{ c_0 Z^6 ( t_0 ) } \exp \left( - C_M ( t - t_0 ) \right) \geq \frac{ 1 }{ 2 c_0 } . 
\]
Note that the interval $ [ t_0 , T_6 ] $ does not depend on $ j $. We conclude that $ X_j $, $ Y_j $ and $ \| f_j \|^2_{ g_{ j - 1 } , m , N } $ are bounded uniformly on $ j $, and we obtain a solution $ g^{ a b } $, $ k^{ a b } $, $ \phi $, $ \psi $ and $ f $ to the EBs system by taking $ j \to \infty $ on $ [ t_0 , T_6 ] $, which satisfies the properties \eqref{local 1}--\eqref{local 5}.

\subsubsection{Proof of Proposition \ref{prop local} (uniqueness)}\label{sec uniqueness} 
Uniqueness can be proved by the standard argument, but we need to make a slight modification to the Boltzmann equation. Suppose that $ { g_{ \mathfrak 1 } }^{ a b } $, $ { k_{ \mathfrak 1 } }^{ a b } $, $ \phi_{ \mathfrak 1 } $, $ \psi_{ \mathfrak 1 } $, $ f_{ \mathfrak 1 } $ and $ { g_{ \mathfrak 2 } }^{ a b } $, $ { k_{ \mathfrak 2 } }^{ a b } $, $ \phi_{ \mathfrak 2 } $, $ \psi_{ \mathfrak 2 } $, $ f_{ \mathfrak 2 } $ are two solutions of the EBs system, on the time interval $ [ t_0 , T_6 ] $, corresponding to an initial data $ { g_0 }^{ a b } $, $ { k_0 }^{ a b } $, $ \phi_0 $, $ \psi_0 $ and $ f_0 $. Now, we define 
\[
{ \tilde f }_{ \mathfrak 1 } ( t , p ) = e^{ \frac12 { p_{ \mathfrak 1 } }^0 } f_{ \mathfrak 1 } ( t , p ) , \qquad { \tilde f }_{ \mathfrak 2 } ( t , p ) = e^{ \frac12 { p_{ \mathfrak 2 } }^0 } f_{ \mathfrak 2 } ( t , p ) , 
\]
where $ { p_{ \mathfrak 1 } }^0 = \sqrt{ 1 + { g_{ \mathfrak 1 } }^{ a b } p_a p_b } $ and $ { p_{ \mathfrak 2 } }^0 = \sqrt{ 1 + { g_{ \mathfrak 2 } }^{ a b } p_a p_b } $. Then, $ { \tilde f }_{ \mathfrak 1 } $ satisfies 
\begin{align}
\frac{ \partial { \tilde f }_{ \mathfrak 1 } }{ \partial t } + \frac{ { k_{ \mathfrak 1 } }^{ a b } p_a p_b }{ 2 { p_{ \mathfrak 1 } }^0 } { \tilde f }_{ \mathfrak 1 } + \frac{ 1 }{ { p_{ \mathfrak 1 } }^0 } C^d_{ b c } p_d { p_{ \mathfrak 1 } }^b \frac{ \partial { \tilde f }_{ \mathfrak 1 } }{ \partial p_c } = { \tilde Q }_{ \mathfrak 1 + } ( { \tilde f }_{ \mathfrak 1 } , { \tilde f }_{ \mathfrak 1 } ) - { \tilde Q }_{ \mathfrak 1 - } ( { \tilde f }_{ \mathfrak 1 } , { \tilde f }_{ \mathfrak 1 } ) , \label{Boltzmanntilde}
\end{align}
where 
\begin{align*}
{ \tilde Q }_{ \mathfrak 1 + } ( { \tilde f }_{ \mathfrak 1 } , { \tilde f }_{ \mathfrak 1 } ) & = ( \det g_{ \mathfrak 1 }^{ - 1 } )^{ \frac12 } \iint \frac{ 1 }{ { p_{ \mathfrak 1 } }^0 { q_{ \mathfrak 1 } }^0 \sqrt{ s_{ \mathfrak 1 } } } { \tilde f }_{ \mathfrak 1 } ( p_{ \mathfrak 1 }' ) { \tilde f }_{ \mathfrak 1 } ( q_{ \mathfrak 1 }' ) e^{ - \frac12 { q_{ \mathfrak 1 } }^0 } \, d \omega \, d q , \\ 
{ \tilde Q }_{ \mathfrak 1 - } ( { \tilde f }_{ \mathfrak 1 } , { \tilde f }_{ \mathfrak 1 } ) & = ( \det g_{ \mathfrak 1 }^{ - 1 } )^{ \frac12 } \iint \frac{ 1 }{ { p_{ \mathfrak 1 } }^0 { q_{ \mathfrak 1 } }^0 \sqrt{ s_{ \mathfrak 1 } } } { \tilde f }_{ \mathfrak 1 } ( p ) { \tilde f }_{ \mathfrak 1 } ( q ) e^{ - \frac12 { q_{ \mathfrak 1 } }^0 } \, d \omega \, d q . 
\end{align*}
Here, $ { p_{ \mathfrak 1 } }^b = { g_{ \mathfrak 1 } }^{ a b } p_a $, $ s_{ \mathfrak 1 } = 2 + 2 { p_{ \mathfrak 1 } }^0 { q_{ \mathfrak 1 } }^0 - 2 { g_{ \mathfrak 1 } }^{ a b } p_a q_b $, and $ p_{ \mathfrak 1 }' $ and $ q_{ \mathfrak 1 }' $ are defined by \eqref{p'} and \eqref{q'} using the metric $ { g_{ \mathfrak 1 } }^{ a b } $. We obtain a similar equation for $ { \tilde f }_{ \mathfrak 2 } $. The uniqueness will be proved by using the modified equation \eqref{Boltzmanntilde} above. Let us first consider the equation for $ { \tilde f }_{ \mathfrak 1 } - { \tilde f }_{ \mathfrak 2 } $: 
\begin{align}
& \frac{ \partial ( { \tilde f }_{ \mathfrak 1 } - { \tilde f }_{ \mathfrak 2 } ) }{ \partial t } + \frac{ 1 }{ 2 } \left( \frac{ 1 }{ { p_{ \mathfrak 1 } }^0 } - \frac{ 1 }{ { p_{ \mathfrak 2 } }^0 } \right) { k_{ \mathfrak 1 } }^{ a b } p_a p_b { \tilde f }_{ \mathfrak 1 } + \frac{ 1 }{ 2 { p_{ \mathfrak 2 } }^0 } \left( { k_{ \mathfrak 1 } }^{ a b } - { k_{ \mathfrak 2 } }^{ a b } \right) p_a p_b { \tilde f }_{ \mathfrak 1 } \nonumber \\ 
& + \frac{ 1 }{ 2 { p_{ \mathfrak 2 } }^0 } { k_{ \mathfrak 2 } }^{ a b } p_a p_b \left( { \tilde f }_{ \mathfrak 1 } - { \tilde f }_{ \mathfrak 2 } \right) + \left( \frac{ 1 }{ { p_{ \mathfrak 1 } }^0 } - \frac{ 1 }{ { p_{ \mathfrak 2 } }^0 } \right) C^d_{ b c } p_d { p_{ \mathfrak 1 } }^b \frac{ \partial { \tilde f }_{ \mathfrak 1 } }{ \partial p_c } + \frac{ 1 }{ { p_{ \mathfrak 2 } }^0 } C^d_{ b c } p_d \left( { p_{ \mathfrak 1 } }^b - { p_{ \mathfrak 2 } }^b \right) \frac{ \partial { \tilde f }_{ \mathfrak 1 } }{ \partial p_c } \nonumber \\ 
& + \frac{ 1 }{ { p_{ \mathfrak 2 } }^0 } C^d_{ b c } p_d { p_{ \mathfrak 2 } }^b \frac{ \partial ( { \tilde f }_{ \mathfrak 1 } - { \tilde f }_{ \mathfrak 2 } ) }{ \partial p_c } = { \tilde Q }_{ \mathfrak 1 + } ( { \tilde f }_{ \mathfrak 1 } , { \tilde f }_{ \mathfrak 1 } ) - { \tilde Q }_{ \mathfrak 2 + } ( { \tilde f }_{ \mathfrak 2 } , { \tilde f }_{ \mathfrak 2 } ) - { \tilde Q }_{ \mathfrak 1 - } ( { \tilde f }_{ \mathfrak 1 } , { \tilde f }_{ \mathfrak 1 } ) + { \tilde Q }_{ \mathfrak 2 - } ( { \tilde f }_{ \mathfrak 2 } , { \tilde f }_{ \mathfrak 2 } ) . \label{Boltzmann12}
\end{align}
Multiplying the above equation by $ { \tilde f }_{ \mathfrak 1 } - { \tilde f }_{ \mathfrak 2 } $ and integrating it over $ \bbr^3_p $, we obtain from the first term on the left hand side: 
\begin{align}
\frac12 \frac{ d }{ d t } \| { \tilde f }_{ \mathfrak 1 } - { \tilde f }_{ \mathfrak 2 } \|^2_{ 0 , 0 } , \label{unique1}
\end{align}
where the norm $ \| \cdot \|_{ m , N } $ is defined in \eqref{norm2}. For the second term we notice that $ { p_{ \mathfrak 1 } }^0 $, $ { p_{ \mathfrak 2 } }^0 $ and $ \langle p \rangle $ are all equivalent, since $ { g_{ \mathfrak 1 } }^{ a b } $ and $ { g_{ \mathfrak 2 } }^{ a b } $ both satisfy the property \eqref{local 4} on $ [ t_0 , T_6 ] $: 
\[
\left| \frac{ 1 }{ { p_{ \mathfrak 1 } }^0 } - \frac{ 1 }{ { p_{ \mathfrak 2 } }^0 } \right| \leq \frac{ C }{ \langle p \rangle^2 } \left| { p_{ \mathfrak 1 } }^0 - { p_{ \mathfrak 2 } }^0 \right| . 
\]
To estimate the quantity $ \left| { p_{ \mathfrak 1 } }^0 - { p_{ \mathfrak 2 } }^0 \right| $, we consider $ p^0 $ as a function of $ g^{ a b } $ such that $ { p_{ \mathfrak 1 } }^0 = p^0 ( { g_{ \mathfrak 1 } }^{ a b } ) $ and $ { p_{ \mathfrak 2 } }^0 = p^0 ( { g_{ \mathfrak 2 } }^{ a b } ) $. Let us write for $ { \mathfrak s } \in [ 1 , 2 ] $, 
\[
{ g_{ \mathfrak s } }^{ a b } = ( 2 - { \mathfrak s } ) { g_{ \mathfrak 1 } }^{ a b } + ( { \mathfrak s } - 1 ) { g_{ \mathfrak 2 } }^{ a b } . 
\]
Then, the metric $ { g_{ \mathfrak s } }^{ a b } $ satisfies 
\begin{align}
\frac{ 1 }{ c_1 } | p |^2 \leq Z^2 { g_{ \mathfrak s } }^{ a b } p_a p_b \leq c_1 | p |^2 , \qquad \max_{ a , b } | Z^2 { g_{ \mathfrak s } }^{ a b } | \leq c_1 , \qquad Z^6 \det { g_{ \mathfrak s } }^{ - 1 } \geq \frac{ 1 }{ c_1 } , \label{g_s est}
\end{align}
since $ { g_{ \mathfrak 1 } }^{ a b } $ and $ { g_{ \mathfrak 2 } }^{ a b } $ both satisfy \eqref{local 4} and \eqref{local 5}. Here, we used Lemma \ref{lem HJ} for the third one. We now obtain
\begin{align}
\left| { p_{ \mathfrak 1 } }^0 - { p_{ \mathfrak 2 } }^0 \right| & = \left| \int_1^2 \frac{ \partial ( p^0 ( { g_{ \mathfrak s } }^{ a b } ) ) }{ \partial { \mathfrak s } } \, d { \mathfrak s } \right| \nonumber \\ 
& \leq C \int_1^2 \left| \frac{ \partial { p }^0 }{ \partial g^{ c d } } ( { g_{ \mathfrak s } }^{ a b } ) \right| | { g_{ \mathfrak 2 } }^{ c d } - { g_{ \mathfrak 1 } }^{ c d } | \, d { \mathfrak s } \nonumber \\ 
& \leq C \langle p \rangle \left( \max_{ c , d } | { g_{ \mathfrak 2 } }^{ c d } - { g_{ \mathfrak 1 } }^{ c d } | \right) , \label{g1-g2}
\end{align}
where we used the same argument as in \eqref{dq^0g} in the last inequality. Hence, we obtain from the second term:
\begin{align}
& \left| \int_{ \bbr^3 } \frac{ 1 }{ 2 } \left( \frac{ 1 }{ { p_{ \mathfrak 1 } }^0 } - \frac{ 1 }{ { p_{ \mathfrak 2 } }^0 } \right) { k_{ \mathfrak 1 } }^{ a b } p_a p_b { \tilde f }_{ \mathfrak 1 } ( { \tilde f }_{ \mathfrak 1 } - { \tilde f }_{ \mathfrak 2 } ) \, d p \right| \nonumber \\ 
& \leq C \max_{ a , b } | { g_{ \mathfrak 1 } }^{ a b } - { g_{ \mathfrak 2 } }^{ a b } | \int \langle p \rangle | { \tilde f }_{ \mathfrak 1 } | | { \tilde f }_{ \mathfrak 1 } - { \tilde f }_{ \mathfrak 2 } | \, d p \nonumber \\ 
& \leq C \max_{ a , b } | { g_{ \mathfrak 1 } }^{ a b } - { g_{ \mathfrak 2 } }^{ a b } | \left( \int | { \tilde f }_{ \mathfrak 1 } |^2 \langle p \rangle^2 \, d p \right)^{ \frac12 } \left( \int | { \tilde f }_{ \mathfrak 1 } - { \tilde f }_{ \mathfrak 2 } |^2 \, d p \right)^{ \frac12 } \nonumber \\ 
& \leq C \left( \max_{ a , b } | { g_{ \mathfrak 1 } }^{ a b } - { g_{ \mathfrak 2 } }^{ a b } | \right) \| { \tilde f }_{ \mathfrak 1 } - { \tilde f }_{ \mathfrak 2 } \|_{ 0 , 0 } , \label{unique2}
\end{align}
where we used $ \| { \tilde f }_{ \mathfrak 1 } \|_{ m , 0 } = \| f_{ \mathfrak 1 } \|_{ g_{ \mathfrak 1 } , m , 0 } $, which is bounded on $ [ t_0 , T_6 ] $. The third term on the left hand side of \eqref{Boltzmann12} is similarly estimated: 
\begin{align}
& \left| \int_{ \bbr^3 } \frac{ 1 }{ 2 { p_{ \mathfrak 2 } }^0 } \left( { k_{ \mathfrak 1 } }^{ a b } - { k_{ \mathfrak 2 } }^{ a b } \right) p_a p_b { \tilde f }_{ \mathfrak 1 } ( { \tilde f }_{ \mathfrak 1 } - { \tilde f }_{ \mathfrak 2 } ) \, d p \right| \nonumber \\ 
& \leq C \left( \max_{ a , b } | { k_{ \mathfrak 1 } }^{ a b } - { k_{ \mathfrak 2 } }^{ a b } | \right) \| { \tilde f }_{ \mathfrak 1 } - { \tilde f }_{ \mathfrak 2 } \|_{ 0 , 0 } . \label{unique3}
\end{align} 
For the fourth term we have the non-negativity: 
\begin{align}
\int_{ \bbr^3 } \frac{ 1 }{ 2 { p_{ \mathfrak 2 } }^0 } { k_{ \mathfrak 2 } }^{ a b } p_a p_b \left( { \tilde f }_{ \mathfrak 1 } - { \tilde f }_{ \mathfrak 2 } \right)^2 \, d p \geq 0 , \label{unique4} 
\end{align}
since $ { k_{ \mathfrak 2 } }^{ a b } $ is positive definite, so that it will be ignored. The fifth term is estimated as follows: 
\begin{align}
& \left| \int_{ \bbr^3 } \left( \frac{ 1 }{ { p_{ \mathfrak 1 } }^0 } - \frac{ 1 }{ { p_{ \mathfrak 2 } }^0 } \right) C^d_{ b c } p_d { p_{ \mathfrak 1 } }^b \frac{ \partial { \tilde f }_{ \mathfrak 1 } }{ \partial p_c } ( { \tilde f }_{ \mathfrak 1 } - { \tilde f }_{ \mathfrak 2 } ) \, d p \right| \nonumber \\ 
& \leq C \max_{ a , b } | { g_{ \mathfrak 1 } }^{ a b } - { g_{ \mathfrak 2 } }^{ a b } | \int \langle p \rangle \left| \frac{ \partial { \tilde f }_{ \mathfrak 1 } }{ \partial p_c } \right| | { \tilde f }_{ \mathfrak 1 } - { \tilde f }_{ \mathfrak 2 } | \, d p \nonumber \\ 
& \leq C \max_{ a , b } | { g_{ \mathfrak 1 } }^{ a b } - { g_{ \mathfrak 2 } }^{ a b } | \left( \int \left| \frac{ \partial { \tilde f }_{ \mathfrak 1 } }{ \partial p_c } \right|^2 \langle p \rangle^2 \, d p \right)^{ \frac12 } \left( \int | { \tilde f }_{ \mathfrak 1 } - { \tilde f }_{ \mathfrak 2 } |^2 \, d p \right)^{ \frac12 } \nonumber \\ 
& \leq C \left( \max_{ a , b } | { g_{ \mathfrak 1 } }^{ a b } - { g_{ \mathfrak 2 } }^{ a b } | \right) \| { \tilde f }_{ \mathfrak 1 } - { \tilde f }_{ \mathfrak 2 } \|_{ 0 , 0 } , \label{unique5}
\end{align}
where we used $ \| { \tilde f }_{ \mathfrak 1 } \|_{ m , N } \leq C \| f_{ \mathfrak 1 } \|_{ g_{ \mathfrak 1 } , m , N } $, which is bounded on $ [ t_0 , T ] $. For the sixth term, we note that
\[
| { p_{ \mathfrak 1 } }^b - { p_{ \mathfrak 2 } }^b | \leq C \left( \max_{ a , b } | { g_{ \mathfrak 1 } }^{ a b } - { g_{ \mathfrak 2 } }^{ a b } | \right) \langle p \rangle . 
\]
Hence, we obtain 
\begin{align}
& \left| \int_{ \bbr^3 } \frac{ 1 }{ { p_{ \mathfrak 2 } }^0 } C^d_{ b c } p_d \left( { p_{ \mathfrak 1 } }^b - { p_{ \mathfrak 2 } }^b \right) \frac{ \partial { \tilde f }_{ \mathfrak 1 } }{ \partial p_c } ( { \tilde f }_{ \mathfrak 1 } - { \tilde f }_{ \mathfrak 2 } ) \, d p \right| \nonumber \\ 
& \leq C \left( \max_{ a , b } | { g_{ \mathfrak 1 } }^{ a b } - { g_{ \mathfrak 2 } }^{ a b } | \right) \| { \tilde f }_{ \mathfrak 1 } - { \tilde f }_{ \mathfrak 2 } \|_{ 0 , 0 } . \label{unique6}
\end{align}
For the seventh term, we integrate it by parts to obtain 
\begin{align}
\left| \int_{ \bbr^3 } \frac{ 1 }{ { p_{ \mathfrak 2 } }^0 } C^d_{ b c } p_d { p_{ \mathfrak 2 } }^b \frac{ \partial ( { \tilde f }_{ \mathfrak 1 } - { \tilde f }_{ \mathfrak 2 } ) }{ \partial p_c } ( { \tilde f }_{ \mathfrak 1 } - { \tilde f }_{ \mathfrak 2 } ) \, d p \right| \leq C \| { \tilde f }_{ \mathfrak 1 } - { \tilde f }_{ \mathfrak 2 } \|^2_{ 0 , 0 } . \label{unique7} 
\end{align}
Now, it remains to estimate the collision terms on the right hand side of \eqref{Boltzmann12}. We first consider the loss terms. We can write 
\begin{align*}
& { \tilde Q }_{ \mathfrak 1 - } ( { \tilde f }_{ \mathfrak 1 } , { \tilde f }_{ \mathfrak 1 } ) - { \tilde Q }_{ \mathfrak 2 - } ( { \tilde f }_{ \mathfrak 2 } , { \tilde f }_{ \mathfrak 2 } ) \\ 
& = \left( ( \det g_{ \mathfrak 1 }^{ - 1 } )^{ \frac12 } - ( \det g_{ \mathfrak 2 }^{ - 1 } )^{ \frac12 } \right) \iint \frac{ 1 }{ { p_{ \mathfrak 1 } }^0 { q_{ \mathfrak 1 } }^0 \sqrt{ s_{ \mathfrak 1 } } } { \tilde f }_{ \mathfrak 1 } ( p ) { \tilde f }_{ \mathfrak 1 } ( q ) e^{ - \frac12 { q_{ \mathfrak 1 } }^0 } \, d \omega \, d q \\ 
& \quad + ( \det g_{ \mathfrak 2 }^{ - 1 } )^{ \frac12 } \iint \left( \frac{ 1 }{ { p_{ \mathfrak 1 } }^0 { q_{ \mathfrak 1 } }^0 \sqrt{ s_{ \mathfrak 1 } } } - \frac{ 1 }{ { p_{ \mathfrak 2 } }^0 { q_{ \mathfrak 2 } }^0 \sqrt{ s_{ \mathfrak 2 } } } \right) { \tilde f }_{ \mathfrak 1 } ( p ) { \tilde f }_{ \mathfrak 1 } ( q ) e^{ - \frac12 { q_{ \mathfrak 1 } }^0 } \, d \omega \, d q \\ 
& \quad + ( \det g_{ \mathfrak 2 }^{ - 1 } )^{ \frac12 } \iint \frac{ 1 }{ { p_{ \mathfrak 2 } }^0 { q_{ \mathfrak 2 } }^0 \sqrt{ s_{ \mathfrak 2 } } } \left( { \tilde f }_{ \mathfrak 1 } ( p ) { \tilde f }_{ \mathfrak 1 } ( q ) - { \tilde f }_{ \mathfrak 2 } ( p ) { \tilde f }_{ \mathfrak 2 } ( q ) \right) e^{ - \frac12 { q_{ \mathfrak 1 } }^0 } \, d \omega \, d q \\ 
& \quad + ( \det g_{ \mathfrak 2 }^{ - 1 } )^{ \frac12 } \iint \frac{ 1 }{ { p_{ \mathfrak 2 } }^0 { q_{ \mathfrak 2 } }^0 \sqrt{ s_{ \mathfrak 2 } } } { \tilde f }_{ \mathfrak 2 } ( p ) { \tilde f }_{ \mathfrak 2 } ( q ) \left( e^{ - \frac12 { q_{ \mathfrak 1 } }^0 } - e^{ - \frac12 { q_{ \mathfrak 2 } }^0 } \right) \, d \omega \, d q \\ 
& = : L_1 + L_2 + L_3 + L_4 . 
\end{align*}
For $ L_1 $, we apply the same argument as in \eqref{g1-g2} and use the properties \eqref{g_s est} to obtain 
\begin{align*}
\left| ( \det g_{ \mathfrak 1 }^{ - 1 } )^{ \frac12 } - ( \det g_{ \mathfrak 2 }^{ - 1 } )^{ \frac12 } \right| & = \left| \int_1^2 \frac{ \partial \left( ( \det g^{ - 1 } )^{ \frac12 } ( { g_{ \mathfrak s } }^{ a b } ) \right) }{ \partial { \mathfrak s } } \, d { \mathfrak s } \right| \\ 
& \leq C \int_1^2 \left| \frac{ \partial ( \det g^{ - 1 } )^{ \frac12 } }{ \partial g^{ c d } } ( { g_{ \mathfrak s } }^{ a b } ) \right| | { g_{ \mathfrak 2 } }^{ c d } - { g_{ \mathfrak 1 } }^{ c d } | \, d { \mathfrak s } \\ 
& \leq C \max_{ c , d } | { g_{ \mathfrak 2 } }^{ c d } - { g_{ \mathfrak 1 } }^{ c d } | . 
\end{align*}
Hence, we obtain 
\begin{align}
& \left| \int_{ \bbr^3 } L_1 ( { \tilde f }_{ \mathfrak 1 } - { \tilde f }_{ \mathfrak 2 } ) \, d p \right| \nonumber \\ 
& \leq C \left( \max_{ a , b } | { g_{ \mathfrak 1 } }^{ a b } - { g_{ \mathfrak 2 } }^{ a b } | \right) \iiint \frac{ 1 }{ { p_{ \mathfrak 1 } }^0 { q_{ \mathfrak 1 } }^0 \sqrt{ s_{ \mathfrak 1 } } } { \tilde f }_{ \mathfrak 1 } ( p ) { \tilde f }_{ \mathfrak 1 } ( q ) e^{ - \frac12 { q_{ \mathfrak 1 } }^0 } | ( { \tilde f }_{ \mathfrak 1 } - { \tilde f }_{ \mathfrak 2 } ) ( p ) | \, d \omega \, d q \, d p \nonumber \\ 
& \leq C \left( \max_{ a , b } | { g_{ \mathfrak 1 } }^{ a b } - { g_{ \mathfrak 2 } }^{ a b } | \right) \left( \iint | { \tilde f }_{ \mathfrak 1 } ( p ) |^2 | { \tilde f }_{ \mathfrak 1 } ( q ) |^2 \, d q \, d p \right)^{ \frac12 } \left( \iint e^{ - { q_{ \mathfrak 1 } }^0 } | ( { \tilde f }_{ \mathfrak 1 } - { \tilde f }_{ \mathfrak 2 } ) ( p ) |^2 \, d q \, d p \right)^{ \frac12 } \nonumber \\ 
& \leq C \left( \max_{ a , b } | { g_{ \mathfrak 1 } }^{ a b } - { g_{ \mathfrak 2 } }^{ a b } | \right) \| { \tilde f }_{ \mathfrak 1 } - { \tilde f }_{ \mathfrak 2 } \|_{ 0 , 0 } . \label{unique8} 
\end{align}
For $ L_2 $, we can apply \eqref{dq^0g} and \eqref{dsg} to obtain 
\[
\left| \frac{ \partial }{ \partial g^{ c d } } \left( \frac{ 1 }{ { p }^0 { q }^0 \sqrt{ s } } \right) \right| \leq C , 
\]
which is true for all $ { g_{ \mathfrak s } }^{ a b } $, $ { \mathfrak s } \in [ 1 , 2 ] $. Hence, we obtain by the same arguments as in \eqref{g1-g2} and \eqref{unique8}: 
\begin{align}
\left| \int_{ \bbr^3 } L_2 ( { \tilde f }_{ \mathfrak 1 } - { \tilde f }_{ \mathfrak 2 } ) \, d p \right| \leq C \left( \max_{ a , b } | { g_{ \mathfrak 1 } }^{ a b } - { g_{ \mathfrak 2 } }^{ a b } | \right) \| { \tilde f }_{ \mathfrak 1 } - { \tilde f }_{ \mathfrak 2 } \|_{ 0 , 0 } . \label{unique9} 
\end{align}
For $ L_3 $, we write 
\[
L_3 = \frac{ ( \det g_{ \mathfrak 2 }^{ - 1 } )^{ \frac12 } }{ 2 } \iint \frac{ 1 }{ { p_{ \mathfrak 2 } }^0 { q_{ \mathfrak 2 } }^0 \sqrt{ s_{ \mathfrak 2 } } } \left( ( { \tilde f }_{ \mathfrak 1 } + { \tilde f }_{ \mathfrak 2 } ) ( p ) ( { \tilde f }_{ \mathfrak 1 } - { \tilde f }_{ \mathfrak 2 } ) ( q ) + ( { \tilde f }_{ \mathfrak 1 } - { \tilde f }_{ \mathfrak 2 } ) ( p ) ( { \tilde f }_{ \mathfrak 1 } + { \tilde f }_{ \mathfrak 2 } ) ( q ) \right) e^{ - \frac12 { q_{ \mathfrak 1 } }^0 } \, d \omega \, d q , 
\]
and obtain 
\begin{align}
& \left| \int_{ \bbr^3 } L_3 ( { \tilde f }_{ \mathfrak 1 } - { \tilde f }_{ \mathfrak 2 } ) \, d p \right| \nonumber \\ 
& \leq C \iint | ( { \tilde f }_{ \mathfrak 1 } + { \tilde f }_{ \mathfrak 2 } ) ( p ) | | ( { \tilde f }_{ \mathfrak 1 } - { \tilde f }_{ \mathfrak 2 } ) ( q ) | | ( { \tilde f }_{ \mathfrak 1 } - { \tilde f }_{ \mathfrak 2 } ) ( p ) | e^{ - \frac12 { q_{ \mathfrak 1 } }^0 } \, d q \, d p \nonumber \\ 
& \quad + C \iint | ( { \tilde f }_{ \mathfrak 1 } - { \tilde f }_{ \mathfrak 2 } ) ( p ) |^2 | ( { \tilde f }_{ \mathfrak 1 } + { \tilde f }_{ \mathfrak 2 } ) ( q ) | e^{ - \frac12 { q_{ \mathfrak 1 } }^0 } \, d q \, d p \nonumber \\ 
& \leq C \left( \iint | ( { \tilde f }_{ \mathfrak 1 } + { \tilde f }_{ \mathfrak 2 } ) ( p ) |^2 | ( { \tilde f }_{ \mathfrak 1 } - { \tilde f }_{ \mathfrak 2 } ) ( q ) |^2 \, d q \, d p \right)^{ \frac12 } \left( \iint | ( { \tilde f }_{ \mathfrak 1 } - { \tilde f }_{ \mathfrak 2 } ) ( p ) |^2 e^{ - { q_{ \mathfrak 1 } }^0 } \, d q \, d p \right)^{ \frac12 } \nonumber \\ 
& \quad + C \| { \tilde f }_{ \mathfrak 1 } - { \tilde f }_{ \mathfrak 2 } \|^2_{ 0 , 0 } \left( \int | ( { \tilde f }_{ \mathfrak 1 } + { \tilde f }_{ \mathfrak 2 } ) ( q ) |^2 \, d q \right)^{ \frac12 } \left( \int e^{ - { q_{ \mathfrak 1 } }^0 } \, d q \right)^{ \frac12 } \nonumber \\ 
& \leq C \| { \tilde f }_{ \mathfrak 1 } - { \tilde f }_{ \mathfrak 2 } \|^2_{ 0 , 0 } , \label{unique10} 
\end{align}
where we used the boundedness of $ \| { \tilde f }_{ \mathfrak 1 } \|_{ 0 , 0 } $ and $ \| { \tilde f }_{ \mathfrak 2 } \|_{ 0 , 0 } $. For $ L_4 $, we apply again the argument of \eqref{g1-g2} with the estimate \eqref{dq^0g} to obtain 
\begin{align*}
| e^{ - \frac12 { q_{ \mathfrak 1 } }^0 } - e^{ - \frac12 { q_{ \mathfrak 2 } }^0 } | & = \left| \int_1^2 \frac{ \partial ( e^{ - \frac12 q^0 } ( { g_{ \mathfrak s } }^{ a b } ) ) }{ \partial { \mathfrak s } } \, d { \mathfrak s } \right| \\ 
& \leq C \int_1^2 \left| \left( \frac{ \partial q^0 }{ \partial g^{ c d } } e^{ - \frac12 q^0 } \right) ( { g_{ \mathfrak s } }^{ a b } ) \right| | { g_{ \mathfrak 2 } }^{ c d } - { g_{ \mathfrak 1 } }^{ c d } | \, d { \mathfrak s } \\ 
& \leq C \langle q \rangle e^{ - \langle q \rangle / C } \max_{ c , d } | { g_{ \mathfrak 2 } }^{ c d } - { g_{ \mathfrak 1 } }^{ c d } | , 
\end{align*}
for some $ C > 0 $, since $ \langle q \rangle $ and $ q^0 ( { g_{ \mathfrak s } }^{ a b } ) $ are equivalent. Now, we have 
\begin{align}
& \left| \int_{ \bbr^3 } L_4 ( { \tilde f }_{ \mathfrak 1 } - { \tilde f }_{ \mathfrak 2 } ) \, d p \right| \nonumber \\ 
& \leq C \left( \max_{ a , b } | { g_{ \mathfrak 2 } }^{ a b } - { g_{ \mathfrak 1 } }^{ a b } | \right) \iint | { \tilde f }_{ \mathfrak 2 } ( p ) | | { \tilde f }_{ \mathfrak 2 } ( q ) | \langle q \rangle e^{ - \langle q \rangle / C } | ( { \tilde f }_{ \mathfrak 1 } - { \tilde f }_{ \mathfrak 2 } ) ( p ) | \, d q \, d p \nonumber \\ 
& \leq C \left( \max_{ a , b } | { g_{ \mathfrak 2 } }^{ a b } - { g_{ \mathfrak 1 } }^{ a b } | \right) \left( \iint | { \tilde f }_{ \mathfrak 2 } ( p ) |^2 | { \tilde f }_{ \mathfrak 2 } ( q ) |^2 \, d q \, d p \right)^{ \frac12 } \left( \iint \langle q \rangle^2 e^{ - 2 \langle q \rangle / C } | ( { \tilde f }_{ \mathfrak 1 } - { \tilde f }_{ \mathfrak 2 } ) ( p ) |^2 \, d q \, d p \right)^{ \frac12 } \nonumber \\ 
& \leq C \left( \max_{ a , b } | { g_{ \mathfrak 2 } }^{ a b } - { g_{ \mathfrak 1 } }^{ a b } | \right) \| { \tilde f }_{ \mathfrak 1 } - { \tilde f }_{ \mathfrak 2 } \|_{ 0 , 0 } . \label{unique11} 
\end{align}
Finally, we consider the gain terms. We write 
\begin{align*}
& { \tilde Q }_{ \mathfrak 1 + } ( { \tilde f }_{ \mathfrak 1 } , { \tilde f }_{ \mathfrak 1 } ) - { \tilde Q }_{ \mathfrak 2 + } ( { \tilde f }_{ \mathfrak 2 } , { \tilde f }_{ \mathfrak 2 } ) \\ 
& = \left( ( \det g_{ \mathfrak 1 }^{ - 1 } )^{ \frac12 } - ( \det g_{ \mathfrak 2 }^{ - 1 } )^{ \frac12 } \right) \iint \frac{ 1 }{ { p_{ \mathfrak 1 } }^0 { q_{ \mathfrak 1 } }^0 \sqrt{ s_{ \mathfrak 1 } } } { \tilde f }_{ \mathfrak 1 } ( p_{ \mathfrak 1 }' ) { \tilde f }_{ \mathfrak 1 } ( q_{ \mathfrak 1 }' ) e^{ - \frac12 { q_{ \mathfrak 1 } }^0 } \, d \omega \, d q \\ 
& \quad + ( \det g_{ \mathfrak 2 }^{ - 1 } )^{ \frac12 } \iint \left( \frac{ 1 }{ { p_{ \mathfrak 1 } }^0 { q_{ \mathfrak 1 } }^0 \sqrt{ s_{ \mathfrak 1 } } } - \frac{ 1 }{ { p_{ \mathfrak 2 } }^0 { q_{ \mathfrak 2 } }^0 \sqrt{ s_{ \mathfrak 2 } } } \right) { \tilde f }_{ \mathfrak 1 } ( p_{ \mathfrak 1 }' ) { \tilde f }_{ \mathfrak 1 } ( q_{ \mathfrak 1 }' ) e^{ - \frac12 { q_{ \mathfrak 1 } }^0 } \, d \omega \, d q \\ 
& \quad + ( \det g_{ \mathfrak 2 }^{ - 1 } )^{ \frac12 } \iint \frac{ 1 }{ { p_{ \mathfrak 2 } }^0 { q_{ \mathfrak 2 } }^0 \sqrt{ s_{ \mathfrak 2 } } } \left( { \tilde f }_{ \mathfrak 1 } ( p_{ \mathfrak 1 }' ) { \tilde f }_{ \mathfrak 1 } ( q_{ \mathfrak 1 }' ) - { \tilde f }_{ \mathfrak 2 } ( p_{ \mathfrak 1 }' ) { \tilde f }_{ \mathfrak 2 } ( q_{ \mathfrak 1 }' ) \right) e^{ - \frac12 { q_{ \mathfrak 1 } }^0 } \, d \omega \, d q \\ 
& \quad + ( \det g_{ \mathfrak 2 }^{ - 1 } )^{ \frac12 } \iint \frac{ 1 }{ { p_{ \mathfrak 2 } }^0 { q_{ \mathfrak 2 } }^0 \sqrt{ s_{ \mathfrak 2 } } } \left( { \tilde f }_{ \mathfrak 2 } ( p_{ \mathfrak 1 }' ) { \tilde f }_{ \mathfrak 2 } ( q_{ \mathfrak 1 }' ) - { \tilde f }_{ \mathfrak 2 } ( p_{ \mathfrak 2 }' ) { \tilde f }_{ \mathfrak 2 } ( q_{ \mathfrak 2 }' ) \right) e^{ - \frac12 { q_{ \mathfrak 1 } }^0 } \, d \omega \, d q \\ 
& \quad + ( \det g_{ \mathfrak 2 }^{ - 1 } )^{ \frac12 } \iint \frac{ 1 }{ { p_{ \mathfrak 2 } }^0 { q_{ \mathfrak 2 } }^0 \sqrt{ s_{ \mathfrak 2 } } } { \tilde f }_{ \mathfrak 2 } ( p_{ \mathfrak 2 }' ) { \tilde f }_{ \mathfrak 2 } ( q_{ \mathfrak 2 }' ) \left( e^{ - \frac12 { q_{ \mathfrak 1 } }^0 } - e^{ - \frac12 { q_{ \mathfrak 2 } }^0 } \right) \, d \omega \, d q \\ 
& = : G_1 + G_2 + G_3 + G_4 + G_5 . 
\end{align*}
The estimates of $ G_1 $, $ G_2 $, $ G_3 $ and $ G_5 $ are almost the same as the estimates of $ L_1 $, $ L_2 $, $ L_3 $ and $ L_4 $, respectively, but we need to apply the change of variables \eqref{change}. We have 
\[
\frac{ d p_{ \mathfrak 1 }' \, d q_{ \mathfrak 1 }' }{ { p_{ \mathfrak 1 }' }^0 { q_{ \mathfrak 1 }' }^0 } = \frac{ d p \, d q }{ { p_{ \mathfrak 1 } }^0 { q_{ \mathfrak 1 } }^0 } , \qquad \frac{ d p_{ \mathfrak 2 }' \, d q_{ \mathfrak 2 }' }{ { p_{ \mathfrak 2 }' }^0 { q_{ \mathfrak 2 }' }^0 } = \frac{ d p \, d q }{ { p_{ \mathfrak 2 } }^0 { q_{ \mathfrak 2 } }^0 } , 
\]
with respect to $ { g_{ \mathfrak 1 } }^{ a b } $ and $ { g_{ \mathfrak 2 } }^{ a b } $, respectively. The estimate of $ G_1 $ is similar to \eqref{unique8}: 
\begin{align}
& \left| \int_{ \bbr^3 } G_1 ( { \tilde f }_{ \mathfrak 1 } - { \tilde f }_{ \mathfrak 2 } ) \, d p \right| \nonumber \\ 
& \leq C \left( \max_{ a , b } | { g_{ \mathfrak 1 } }^{ a b } - { g_{ \mathfrak 2 } }^{ a b } | \right) \iiint \frac{ 1 }{ { p_{ \mathfrak 1 } }^0 { q_{ \mathfrak 1 } }^0 \sqrt{ s_{ \mathfrak 1 } } } { \tilde f }_{ \mathfrak 1 } ( p_{ \mathfrak 1 }' ) { \tilde f }_{ \mathfrak 1 } ( q_{ \mathfrak 1 }' ) e^{ - \frac12 { q_{ \mathfrak 1 } }^0 } | ( { \tilde f }_{ \mathfrak 1 } - { \tilde f }_{ \mathfrak 2 } ) ( p ) | \, d \omega \, d q \, d p \nonumber \\ 
& \leq C \left( \max_{ a , b } | { g_{ \mathfrak 1 } }^{ a b } - { g_{ \mathfrak 2 } }^{ a b } | \right) \left( \iiint | { \tilde f }_{ \mathfrak 1 } ( p_{ \mathfrak 1 }' ) |^2 | { \tilde f }_{ \mathfrak 1 } ( q_{ \mathfrak 1 }' ) |^2 \, \frac{ d \omega \, d p \, d q }{ { p_{ \mathfrak 1 } }^0 { q_{ \mathfrak 1 } }^0 } \right)^{ \frac12 } \left( \iint e^{ - { q_{ \mathfrak 1 } }^0 } | ( { \tilde f }_{ \mathfrak 1 } - { \tilde f }_{ \mathfrak 2 } ) ( p ) |^2 \, d q \, d p \right)^{ \frac12 } \nonumber \\ 
& \leq C \left( \max_{ a , b } | { g_{ \mathfrak 1 } }^{ a b } - { g_{ \mathfrak 2 } }^{ a b } | \right) \| { \tilde f }_{ \mathfrak 1 } - { \tilde f }_{ \mathfrak 2 } \|_{ 0 , 0 } . \label{unique12} 
\end{align}
The quantity $ G_2 $ is estimated as in \eqref{unique9} by using 
\[
\left| \frac{ 1 }{ { p_{ \mathfrak 1 } }^0 { q_{ \mathfrak 1 } }^0 \sqrt{ s_{ \mathfrak 1 } } } - \frac{ 1 }{ { p_{ \mathfrak 2 } }^0 { q_{ \mathfrak 2 } }^0 \sqrt{ s_{ \mathfrak 2 } } } \right| \leq C \left( \max_{ a , b } | { g_{ \mathfrak 1 } }^{ a b } - { g_{ \mathfrak 2 } }^{ a b } | \right) , 
\]
so that we have 
\begin{align}
& \left| \int_{ \bbr^3 } G_2 ( { \tilde f }_{ \mathfrak 1 } - { \tilde f }_{ \mathfrak 2 } ) \, d p \right| \nonumber \\ 
& \leq C \left( \max_{ a , b } | { g_{ \mathfrak 1 } }^{ a b } - { g_{ \mathfrak 2 } }^{ a b } | \right) \iiint { \tilde f }_{ \mathfrak 1 } ( p_{ \mathfrak 1 }' ) { \tilde f }_{ \mathfrak 1 } ( q_{ \mathfrak 1 }' ) e^{ - \frac12 { q_{ \mathfrak 1 } }^0 } | ( { \tilde f }_{ \mathfrak 1 } - { \tilde f }_{ \mathfrak 2 } ) ( p ) | \, d \omega \, d q \, d p \nonumber \\ 
& \leq C \left( \max_{ a , b } | { g_{ \mathfrak 1 } }^{ a b } - { g_{ \mathfrak 2 } }^{ a b } | \right) \left( \iiint | { \tilde f }_{ \mathfrak 1 } ( p_{ \mathfrak 1 }' ) |^2 | { \tilde f }_{ \mathfrak 1 } ( q_{ \mathfrak 1 }' ) |^2 \, d \omega \, d p \, d q \right)^{ \frac12 } \left( \iint e^{ - { q_{ \mathfrak 1 } }^0 } | ( { \tilde f }_{ \mathfrak 1 } - { \tilde f }_{ \mathfrak 2 } ) ( p ) |^2 \, d q \, d p \right)^{ \frac12 } \nonumber \\ 
& \leq C \left( \max_{ a , b } | { g_{ \mathfrak 1 } }^{ a b } - { g_{ \mathfrak 2 } }^{ a b } | \right) \left( \iiint | { \tilde f }_{ \mathfrak 1 } ( p_{ \mathfrak 1 }' ) |^2 | { \tilde f }_{ \mathfrak 1 } ( q_{ \mathfrak 1 }' ) |^2 \frac{ { p_{ \mathfrak 1 } }^0 { q_{ \mathfrak 1 } }^0 }{ { p_{ \mathfrak 1 }' }^0 { q_{ \mathfrak 1 }' }^0 } \, d \omega \, d p_{ \mathfrak 1 }' \, d q_{ \mathfrak 1 }' \right)^{ \frac12 } \| { \tilde f }_{ \mathfrak 1 } - { \tilde f }_{ \mathfrak 2 } \|_{ 0 , 0 } \nonumber \\ 
& \leq C \left( \max_{ a , b } | { g_{ \mathfrak 1 } }^{ a b } - { g_{ \mathfrak 2 } }^{ a b } | \right) \left( \iiint | { \tilde f }_{ \mathfrak 1 } ( p ) |^2 | { \tilde f }_{ \mathfrak 1 } ( q ) |^2 \langle p \rangle \langle q \rangle \, d \omega \, d p \, d q \right)^{ \frac12 } \| { \tilde f }_{ \mathfrak 1 } - { \tilde f }_{ \mathfrak 2 } \|_{ 0 , 0 } \nonumber \\ 
& \leq C \left( \max_{ a , b } | { g_{ \mathfrak 1 } }^{ a b } - { g_{ \mathfrak 2 } }^{ a b } | \right) \| { \tilde f }_{ \mathfrak 1 } - { \tilde f }_{ \mathfrak 2 } \|_{ 0 , 0 } , \label{unique13} 
\end{align}
where we used 
\[
\frac{ { p_{ \mathfrak 1 } }^0 { q_{ \mathfrak 1 } }^0 }{ { p_{ \mathfrak 1 }' }^0 { q_{ \mathfrak 1 }' }^0 } \leq \frac{ ( { p_{ \mathfrak 1 }' }^0 + { q_{ \mathfrak 1 }' }^0 )^2 }{ { p_{ \mathfrak 1 }' }^0 { q_{ \mathfrak 1 }' }^0 } \leq C { p_{ \mathfrak 1 }' }^0 { q_{ \mathfrak 1 }' }^0 , 
\]
and the equivalence of $ { p_{ \mathfrak 1 } }^0 { q_{ \mathfrak 1 } }^0 $ and $ \langle p \rangle \langle q \rangle $ on $ [ t_0 , T_6 ] $. The estimate of $ G_3 $ is almost the same as that of $ L_3 $. We write 
\[
G_3 = \frac{ ( \det g_{ \mathfrak 2 }^{ - 1 } )^{ \frac12 } }{ 2 } \iint \frac{ 1 }{ { p_{ \mathfrak 2 } }^0 { q_{ \mathfrak 2 } }^0 \sqrt{ s_{ \mathfrak 2 } } } \left( ( { \tilde f }_{ \mathfrak 1 } + { \tilde f }_{ \mathfrak 2 } ) ( p_{ \mathfrak 1 }' ) ( { \tilde f }_{ \mathfrak 1 } - { \tilde f }_{ \mathfrak 2 } ) ( q_{ \mathfrak 1 }' ) + ( { \tilde f }_{ \mathfrak 1 } - { \tilde f }_{ \mathfrak 2 } ) ( p_{ \mathfrak 1 }' ) ( { \tilde f }_{ \mathfrak 1 } + { \tilde f }_{ \mathfrak 2 } ) ( q_{ \mathfrak 1 }' ) \right) e^{ - \frac12 { q_{ \mathfrak 1 } }^0 } \, d \omega \, d q , 
\]
and obtain 
\begin{align*}
& \left| \int_{ \bbr^3 } G_3 ( { \tilde f }_{ \mathfrak 1 } - { \tilde f }_{ \mathfrak 2 } ) \, d p \right| \nonumber \\ 
& \leq C \iiint \frac{ 1 }{ { p_{ \mathfrak 2 } }^0 { q_{ \mathfrak 2 } }^0 } | ( { \tilde f }_{ \mathfrak 1 } + { \tilde f }_{ \mathfrak 2 } ) ( p_{ \mathfrak 1 }' ) | | ( { \tilde f }_{ \mathfrak 1 } - { \tilde f }_{ \mathfrak 2 } ) ( q_{ \mathfrak 1 }' ) | | ( { \tilde f }_{ \mathfrak 1 } - { \tilde f }_{ \mathfrak 2 } ) ( p ) | e^{ - \frac12 { q_{ \mathfrak 1 } }^0 } \, d \omega \, d q \, d p \nonumber \\ 
& \quad + C \iiint \frac{ 1 }{ { p_{ \mathfrak 2 } }^0 { q_{ \mathfrak 2 } }^0 } | ( { \tilde f }_{ \mathfrak 1 } - { \tilde f }_{ \mathfrak 2 } ) ( p_{ \mathfrak 1 }' ) | | ( { \tilde f }_{ \mathfrak 1 } + { \tilde f }_{ \mathfrak 2 } ) ( q_{ \mathfrak 1 }' ) | | ( { \tilde f }_{ \mathfrak 1 } - { \tilde f }_{ \mathfrak 2 } ) ( p ) | e^{ - \frac12 { q_{ \mathfrak 1 } }^0 } \, d \omega \, d q \, d p \nonumber \\ 
& \leq C \left( \iiint \frac{ 1 }{ { p_{ \mathfrak 2 } }^0 { q_{ \mathfrak 2 } }^0 } | ( { \tilde f }_{ \mathfrak 1 } + { \tilde f }_{ \mathfrak 2 } ) ( p_{ \mathfrak 1 }' ) |^2 | ( { \tilde f }_{ \mathfrak 1 } - { \tilde f }_{ \mathfrak 2 } ) ( q_{ \mathfrak 1 }' ) |^2 \, d \omega \, d q \, d p \right)^{ \frac12 } \left( \iint | ( { \tilde f }_{ \mathfrak 1 } - { \tilde f }_{ \mathfrak 2 } ) ( p ) |^2 e^{ - { q_{ \mathfrak 1 } }^0 } \, d q \, d p \right)^{ \frac12 } \nonumber \\ 
& \quad + C \left( \iiint \frac{ 1 }{ { p_{ \mathfrak 2 } }^0 { q_{ \mathfrak 2 } }^0 } | ( { \tilde f }_{ \mathfrak 1 } - { \tilde f }_{ \mathfrak 2 } ) ( p_{ \mathfrak 1 }' ) |^2 | ( { \tilde f }_{ \mathfrak 1 } + { \tilde f }_{ \mathfrak 2 } ) ( q_{ \mathfrak 1 }' ) |^2 \, d \omega \, d q \, d p \right)^{ \frac12 } \left( \iint | ( { \tilde f }_{ \mathfrak 1 } - { \tilde f }_{ \mathfrak 2 } ) ( p ) |^2 e^{ - { q_{ \mathfrak 1 } }^0 } \, d q \, d p \right)^{ \frac12 } . 
\end{align*}
Here, we need to use the change of variables, but the kernel $ 1 / ( { p_{ \mathfrak 2 } }^0 { q_{ \mathfrak 2 } }^0 ) $ should be estimated as
\begin{align}
\frac{ 1 }{ { p_{ \mathfrak 2 } }^0 { q_{ \mathfrak 2 } }^0 } = \frac{ { p_{ \mathfrak 1 } }^0 { q_{ \mathfrak 1 } }^0 }{ { p_{ \mathfrak 2 } }^0 { q_{ \mathfrak 2 } }^0 } \frac{ 1 }{ { p_{ \mathfrak 1 } }^0 { q_{ \mathfrak 1 } }^0 } \leq \frac{ C }{ { p_{ \mathfrak 1 } }^0 { q_{ \mathfrak 1 } }^0 } , \label{p1q1p2q2}
\end{align}
by the equivalence of $ { p_{ \mathfrak 1 } }^0 { q_{ \mathfrak 1 } }^0 $, $ { p_{ \mathfrak 2 } }^0 { q_{ \mathfrak 2 } }^0 $ and $ \langle p \rangle \langle q \rangle $ on $ [ t_0 , T_6 ] $. Then, we obtain 
\begin{align}
& \left| \int_{ \bbr^3 } G_3 ( { \tilde f }_{ \mathfrak 1 } - { \tilde f }_{ \mathfrak 2 } ) \, d p \right| \leq C \| { \tilde f }_{ \mathfrak 1 } - { \tilde f }_{ \mathfrak 2 } \|^2_{ 0 , 0 } . \label{unique14} 
\end{align}
We now consider $ G_4 $: 
\[
G_4 = ( \det g_{ \mathfrak 2 }^{ - 1 } )^{ \frac12 } \iint \frac{ 1 }{ { p_{ \mathfrak 2 } }^0 { q_{ \mathfrak 2 } }^0 \sqrt{ s_{ \mathfrak 2 } } } \left( { \tilde f }_{ \mathfrak 2 } ( p_{ \mathfrak 1 }' ) { \tilde f }_{ \mathfrak 2 } ( q_{ \mathfrak 1 }' ) - { \tilde f }_{ \mathfrak 2 } ( p_{ \mathfrak 2 }' ) { \tilde f }_{ \mathfrak 2 } ( q_{ \mathfrak 2 }' ) \right) e^{ - \frac12 { q_{ \mathfrak 1 } }^0 } \, d \omega \, d q . 
\]
We notice that $ F : = { \tilde f }_{ \mathfrak 2 } ( p' ) { \tilde f }_{ \mathfrak 2 } ( q' ) $ is a function of $ g^{ a b } $ such that $ F ( { g_{ \mathfrak 1 } }^{ a b } ) = { \tilde f }_{ \mathfrak 2 } ( p_{ \mathfrak 1 }' ) { \tilde f }_{ \mathfrak 2 } ( q_{ \mathfrak 1 }' ) $ and $ F ( { g_{ \mathfrak 2 } }^{ a b } ) = { \tilde f }_{ \mathfrak 2 } ( p_{ \mathfrak 2 }' ) { \tilde f }_{ \mathfrak 2 } ( q_{ \mathfrak 2 }' ) $. Hence, by the argument of \eqref{g1-g2} again, we obtain 
\begin{align*}
| F ( { g_{ \mathfrak 1 } }^{ a b } ) - F ( { g_{ \mathfrak 2 } }^{ a b } ) | & = \left| \int_1^2 \frac{ \partial ( F ( { g_{ \mathfrak s } }^{ a b } ) ) }{ \partial { \mathfrak s } } \, d { \mathfrak s } \right| \\ 
& \leq C \int_1^2 \left| \frac{ \partial F }{ \partial g^{ c d } } ( { g_{ \mathfrak s } }^{ a b } ) \right| | { g_{ \mathfrak 2 } }^{ c d } - { g_{ \mathfrak 1 } }^{ c d } | \, d { \mathfrak s } , 
\end{align*}
where the quantity $ \partial F / \partial g^{ c d } $ is estimated by $ \partial_p { \tilde f }_{ \mathfrak 2 } ( p' ) { \tilde f }_{ \mathfrak 2 } ( q' ) $ and $ { \tilde f }_{ \mathfrak 2 } ( p' ) \partial_p { \tilde f }_{ \mathfrak 2 } ( q' ) $ with $ \partial p' / \partial g^{ c d } $ and $ \partial q' / \partial g^{ c d } $ evaluated at $ { g_{ \mathfrak s } }^{ a b } $. We apply Lemma \ref{lem dp'g} to obtain 
\begin{align*}
| F ( { g_{ \mathfrak 1 } }^{ a b } ) - F ( { g_{ \mathfrak 2 } }^{ a b } ) | & \leq C \left( \max_{ a , b } | { g_{ \mathfrak 2 } }^{ a b } - { g_{ \mathfrak 1 } }^{ a b } | \right) \langle p \rangle \langle q \rangle^4 \int_1^2 | \partial_p { \tilde f }_{ \mathfrak 2 } ( p_{ \mathfrak s }' ) | | { \tilde f }_{ \mathfrak 2 } ( q_{ \mathfrak s }' ) | + | { \tilde f }_{ \mathfrak 2 } ( p_{ \mathfrak s }' ) | | \partial_p { \tilde f }_{ \mathfrak 2 } ( q_{ \mathfrak s }' ) | \, d { \mathfrak s } , 
\end{align*}
which implies that 
\begin{align}
& \left| \int_{ \bbr^3 } G_4 ( { \tilde f }_{ \mathfrak 1 } - { \tilde f }_{ \mathfrak 2 } ) \, d p \right| \nonumber \\ 
& \leq C \left( \max_{ a , b } | { g_{ \mathfrak 1 } }^{ a b } - { g_{ \mathfrak 2 } }^{ a b } | \right) \int_1^2 \iiint \frac{ 1 }{ { p_{ \mathfrak 2 } }^0 { q_{ \mathfrak 2 } }^0 } \langle p \rangle \langle q \rangle^4 | \partial_p { \tilde f }_{ \mathfrak 2 } ( p_{ \mathfrak s }' ) | | { \tilde f }_{ \mathfrak 2 } ( q_{ \mathfrak s }' ) | e^{ - \frac12 { q_{ \mathfrak 1 } }^0 } | ( { \tilde f }_{ \mathfrak 1 } - { \tilde f }_{ \mathfrak 2 } ) ( p ) | \, d \omega \, d q \, d p \, d { \mathfrak s } \nonumber \\ 
& \quad + C \left( \max_{ a , b } | { g_{ \mathfrak 1 } }^{ a b } - { g_{ \mathfrak 2 } }^{ a b } | \right) \int_1^2 \iiint \frac{ 1 }{ { p_{ \mathfrak 2 } }^0 { q_{ \mathfrak 2 } }^0 } \langle p \rangle \langle q \rangle^4 | { \tilde f }_{ \mathfrak 2 } ( p_{ \mathfrak s }' ) | | \partial_p { \tilde f }_{ \mathfrak 2 } ( q_{ \mathfrak s }' ) | e^{ - \frac12 { q_{ \mathfrak 1 } }^0 } | ( { \tilde f }_{ \mathfrak 1 } - { \tilde f }_{ \mathfrak 2 } ) ( p ) | \, d \omega \, d q \, d p \, d { \mathfrak s } \nonumber \allowdisplaybreaks \\ 
& \leq C \left( \max_{ a , b } | { g_{ \mathfrak 1 } }^{ a b } - { g_{ \mathfrak 2 } }^{ a b } | \right) \nonumber \\ 
& \qquad \times \int_1^2 \left( \iiint \frac{ \langle p \rangle^2 }{ { p_{ \mathfrak 2 } }^0 { q_{ \mathfrak 2 } }^0 } | \partial_p { \tilde f }_{ \mathfrak 2 } ( p_{ \mathfrak s }' ) |^2 | { \tilde f }_{ \mathfrak 2 } ( q_{ \mathfrak s }' ) |^2 \, d \omega \, d q \, d p \right)^{ \frac12 } \left( \iint \langle q \rangle^8 e^{ - { q_{ \mathfrak 1 } }^0 } | ( { \tilde f }_{ \mathfrak 1 } - { \tilde f }_{ \mathfrak 2 } ) ( p ) |^2 \, d q \, d p \right)^{ \frac12 } \, d { \mathfrak s } \nonumber \\ 
& \quad + C \left( \max_{ a , b } | { g_{ \mathfrak 1 } }^{ a b } - { g_{ \mathfrak 2 } }^{ a b } | \right) \nonumber \\
& \qquad \times \int_1^2 \left( \iiint \frac{ \langle p \rangle^2 }{ { p_{ \mathfrak 2 } }^0 { q_{ \mathfrak 2 } }^0 } | { \tilde f }_{ \mathfrak 2 } ( p_{ \mathfrak s }' ) |^2 | \partial_p { \tilde f }_{ \mathfrak 2 } ( q_{ \mathfrak s }' ) |^2 \, d \omega \, d q \, d p \right)^{ \frac12 } \left( \iint \langle q \rangle^8 e^{ - { q_{ \mathfrak 1 } }^0 } | ( { \tilde f }_{ \mathfrak 1 } - { \tilde f }_{ \mathfrak 2 } ) ( p ) |^2 \, d q \, d p \right)^{ \frac12 } \, d { \mathfrak s } \nonumber \allowdisplaybreaks \\ 
& \leq C \left( \max_{ a , b } | { g_{ \mathfrak 1 } }^{ a b } - { g_{ \mathfrak 2 } }^{ a b } | \right) \| { \tilde f }_{ \mathfrak 1 } - { \tilde f }_{ \mathfrak 2 } \|_{ 0 , 0 } \int_1^2 \left( \iiint \frac{ \langle p_{ \mathfrak s }' \rangle^2 \langle q_{ \mathfrak s }' \rangle^2 }{ { p_{ \mathfrak s } }^0 { q_{ \mathfrak s } }^0 } | \partial_p { \tilde f }_{ \mathfrak 2 } ( p_{ \mathfrak s }' ) |^2 | { \tilde f }_{ \mathfrak 2 } ( q_{ \mathfrak s }' ) |^2 \, d \omega \, d q \, d p \right)^{ \frac12 } \, d { \mathfrak s } \nonumber \\ 
& \quad + C \left( \max_{ a , b } | { g_{ \mathfrak 1 } }^{ a b } - { g_{ \mathfrak 2 } }^{ a b } | \right) \| { \tilde f }_{ \mathfrak 1 } - { \tilde f }_{ \mathfrak 2 } \|_{ 0 , 0 } \int_1^2 \left( \iiint \frac{ \langle p_{ \mathfrak s }' \rangle^2 \langle q_{ \mathfrak s }' \rangle^2 }{ { p_{ \mathfrak s } }^0 { q_{ \mathfrak s } }^0 } | { \tilde f }_{ \mathfrak 2 } ( p_{ \mathfrak s }' ) |^2 | \partial_p { \tilde f }_{ \mathfrak 2 } ( q_{ \mathfrak s }' ) |^2 \, d \omega \, d q \, d p \right)^{ \frac12 } \, d { \mathfrak s } , \nonumber 
\end{align}
where we used Lemma \ref{lem pp'q'} and the argument of \eqref{p1q1p2q2}. Here, $ p_{ \mathfrak s }' $ and $ q_{ \mathfrak s }' $ are defined by \eqref{p'} and \eqref{q'} using the metric $ { g_{ \mathfrak s } }^{ a b } $. For the two integrals above, over $ \bbs^2_\omega \times \bbr^3_q \times \bbr^3_p $, we use the change of variables
\[
\frac{ d p_{ \mathfrak s }' \, d q_{ \mathfrak s }' }{ { p_{ \mathfrak s }' }^0 { q_{ \mathfrak s }' }^0 } = \frac{ d p \, d q }{ { p_{ \mathfrak s } }^0 { q_{ \mathfrak s } }^0 } , 
\]
and the equivalence of $ { p_{ \mathfrak s } }^0 { q_{ \mathfrak s } }^0 $ and $ { p_{ \mathfrak 2 } }^0 { q_{ \mathfrak 2 } }^0 $ to obtain 
\begin{align*}
& \left( \iiint \frac{ \langle p_{ \mathfrak s }' \rangle^2 \langle q_{ \mathfrak s }' \rangle^2 }{ { p_{ \mathfrak s } }^0 { q_{ \mathfrak s } }^0 } | \partial_p { \tilde f }_{ \mathfrak 2 } ( p_{ \mathfrak s }' ) |^2 | { \tilde f }_{ \mathfrak 2 } ( q_{ \mathfrak s }' ) |^2 \, d \omega \, d q \, d p \right)^{ \frac12 } \\ 
& \leq C \left( \iiint \frac{ \langle p \rangle^2 \langle q \rangle^2 }{ { p_{ \mathfrak s } }^0 { q_{ \mathfrak s } }^0 } | \partial_p { \tilde f }_{ \mathfrak 2 } ( p ) |^2 | { \tilde f }_{ \mathfrak 2 } ( q ) |^2 \, d \omega \, d q \, d p \right)^{ \frac12 } \\ 
& \leq C \left( \iint \langle p \rangle^2 \langle q \rangle^2 | \partial_p { \tilde f }_{ \mathfrak 2 } ( p ) |^2 | { \tilde f }_{ \mathfrak 2 } ( q ) |^2 \, d q \, d p \right)^{ \frac12 } \\ 
& \leq C \| { \tilde f }_{ \mathfrak 2 } \|_{ 1 , 1 } \| { \tilde f }_{ \mathfrak 2 } \|_{ 1 , 0 } , 
\end{align*}
which is bounded on $ [ t_0 , T_6 ] $. The second integral is estimated the same way. We conclude that 
\begin{align}
& \left| \int_{ \bbr^3 } G_4 ( { \tilde f }_{ \mathfrak 1 } - { \tilde f }_{ \mathfrak 2 } ) \, d p \right| \leq C \left( \max_{ a , b } | { g_{ \mathfrak 1 } }^{ a b } - { g_{ \mathfrak 2 } }^{ a b } | \right) \| { \tilde f }_{ \mathfrak 1 } - { \tilde f }_{ \mathfrak 2 } \|_{ 0 , 0 } . \label{unique15} 
\end{align}
For the last term $ G_5 $, we obtain the same result as above, but we skip the details, since it can be obtained by the same way as in \eqref{unique11} with the change of variables $ ( p_{ \mathfrak 2 }' , q_{ \mathfrak 2 }' ) \mapsto ( p , q ) $:
\begin{align}
\left| \int_{ \bbr^3 } G_5 ( { \tilde f }_{ \mathfrak 1 } - { \tilde f }_{ \mathfrak 2 } ) \, d p \right| \leq C \left( \max_{ a , b } | { g_{ \mathfrak 2 } }^{ a b } - { g_{ \mathfrak 1 } }^{ a b } | \right) \| { \tilde f }_{ \mathfrak 1 } - { \tilde f }_{ \mathfrak 2 } \|_{ 0 , 0 } . \label{unique16} 
\end{align}
To summarize, we combine \eqref{unique1}, \eqref{unique2}--\eqref{unique13}, \eqref{unique14}--\eqref{unique16} to obtain 
\begin{align}
\frac{ d }{ d t } \| { \tilde f }_{ \mathfrak 1 } - { \tilde f }_{ \mathfrak 2 } \|_{ 0 , 0 } \leq C \left( \max_{ a , b } | { g_{ \mathfrak 1 } }^{ a b } - { g_{ \mathfrak 2 } }^{ a b } | + \max_{ a , b } | { k_{ \mathfrak 1 } }^{ a b } - { k_{ \mathfrak 2 } }^{ a b } | + \| { \tilde f }_{ \mathfrak 1 } - { \tilde f }_{ \mathfrak 2 } \|_{ 0 , 0 } \right) . \label{uniquef} 
\end{align}
The differential inequalities for $ { g_{ \frak 1 } }^{ a b } - { g_{ \frak 2 } }^{ a b } $ and $ { k_{ \frak 1 } }^{ a b } - { k_{ \frak 2 } }^{ a b } $ are easily obtained. We obtain from \eqref{EBs^1} and \eqref{EBs^2}:
\begin{align*}
\frac{ d }{ d t } ( { g_{ \frak 1 } }^{ a b } - { g_{ \frak 2 } }^{ a b } ) & = - 2 ( { k_{ \mathfrak 1 } }^{ a b } - { k_{ \frak 2 } }^{ a b } ) , \\ 
\frac{ d }{ d t } ( { k_{ \frak 1 } }^{ a b } - { k_{ \frak 2 } }^{ a b } ) & = - 2 { k_{ \frak 1 } }^a_c { k_{ \frak 1 } }^{ b c } - { k_{ \frak 1 } } { k_{ \frak 1 } }^{ a b } - { R_{ \frak 1 } }^{ a b } + { S_{ \frak 1 } }^{ a b } + \frac12 ( { \rho_{ \frak 1 } } - { S_{ \frak 1 } } ) { g_{ \frak 1 } }^{ a b } + V ( \phi_{ \frak 1 } ) { g_{ \frak 1 } }^{ a b } \\ 
& \quad + 2 { k_{ \frak 2 } }^a_c { k_{ \frak 2 } }^{ b c } + { k_{ \frak 2 } } { k_{ \frak 2 } }^{ a b } + { R_{ \frak 2 } }^{ a b } - { S_{ \frak 2 } }^{ a b } - \frac12 ( { \rho_{ \frak 2 } } - { S_{ \frak 2 } } ) { g_{ \frak 2 } }^{ a b } - V ( \phi_{ \frak 2 } ) { g_{ \frak 2 } }^{ a b } . 
\end{align*}
Recall that $ { g_{ \frak 1 } }_{ a b } $ is a polynomial of $ { g_{ \frak 1 } }^{ a b } $ and $ ( \det { g_{ \frak 1 } }^{ - 1 } )^{ - 1 } $. Hence, the quantities $ { k_{ \frak 1 } }^a_c = { k_{ \frak 1 } }^{ a d } { g_{ \frak 1 } }_{ c d } $, $ k_{ \frak 1 } = { k_{ \frak 1 } }^{ a b } { g_{ \frak 1 } }_{ a b } $ and $ { R_{ \frak 1 } }^{ a b } $ are some polynomials of $ { g_{ \frak 1 } }^{ a b } $, $ { k_{ \frak 1 } }^{ a b } $ and $ ( \det { g_{ \frak 1 } }^{ - 1 } )^{ - 1 } $, and similar arguments are applied to the quantities $ { k_{ \frak 2 } }^a_c $, $ k_{ \frak 2 } $ and $ { R_{ \frak 2 } }^{ a b } $. Moreover, since the quantities $ { g_{ \frak 1 } }^{ a b } $, $ { k_{ \frak 1 } }^{ a b } $, $ ( \det { g_{ \frak 1 } }^{ - 1 } )^{ - 1 } $, $ { g_{ \frak 2 } }^{ a b } $, $ { k_{ \frak 2 } }^{ a b } $ and $ ( \det { g_{ \frak 2 } }^{ - 1 } )^{ - 1 } $ are all bounded on $ [ t_0 , T_6 ] $, we obtain 
\begin{align*}
& \left| - 2 { k_{ \frak 1 } }^a_c { k_{ \frak 1 } }^{ b c } - { k_{ \frak 1 } } { k_{ \frak 1 } }^{ a b } - { R_{ \frak 1 } }^{ a b } + 2 { k_{ \frak 2 } }^a_c { k_{ \frak 2 } }^{ b c } + { k_{ \frak 2 } } { k_{ \frak 2 } }^{ a b } + { R_{ \frak 2 } }^{ a b } \right| \leq C \left( \max_{ a , b } | { g_{ \frak 1 } }^{ a b } - { g_{ \frak 2 } }^{ a b } | + \max_{ a , b } | { k_{ \frak 1 } }^{ a b } - { k_{ \frak 2 } }^{ a b } | \right) , 
\end{align*}
for some positive $ C $. For the matter terms we consider 
\[
{ S_{ \frak 1 } }^{ a b } - { S_{ \frak 2 } }^{ a b } = ( \det { g_{ \frak 1 } }^{ - 1 } )^{ \frac12 } { g_{ \frak 1 } }^{ a c } { g_{ \frak 1 } }^{ b d } \int_{ \bbr^3 } { f_{ \frak 1 } } p_c p_d \frac{ d p }{ { p_{ \frak 1 } }^0 } - ( \det { g_{ \frak 2 } }^{ - 1 } )^{ \frac12 } { g_{ \frak 2 } }^{ a c } { g_{ \frak 2 } }^{ b d } \int_{ \bbr^3 } { f_{ \frak 2 } } p_c p_d \frac{ d p }{ { p_{ \frak 2 } }^0 } . 
\]
We note that $ ( \det { g_{ \frak 1 } }^{ - 1 } )^{ \frac12 } { g_{ \frak 1 } }^{ a c } { g_{ \frak 1 } }^{ b d } - ( \det { g_{ \frak 2 } }^{ - 1 } )^{ \frac12 } { g_{ \frak 2 } }^{ a c } { g_{ \frak 2 } }^{ b d } $ is bounded by a constant multiplied by $ \max_{ a , b } | { g_{ \frak 1 } }^{ a b } - { g_{ \frak 2 } }^{ a b } | $. We also note that the integrals above are estimated, in terms of $ { \tilde f }_{ \mathfrak 1 } $ and $ { \tilde f }_{ \mathfrak 2 } $, as in \eqref{matter est}:
\[
\left| \int_{ \bbr^3 } { f_{ \mathfrak 1 } } p_c p_d \frac{ d p }{ { p_{ \mathfrak 1 } }^0 } \right| = \left| \int_{ \bbr^3 } { { \tilde f }_{ \mathfrak 1 } } e^{ - \frac12 { p_{ \mathfrak 1 } }^0 } p_c p_d \frac{ d p }{ { p_{ \mathfrak 1 } }^0 } \right| \leq C \| { \tilde f }_{ \mathfrak 1 } \|_{ 0 , 0 } , 
\]
and similarly for the integral of $ f_{ \mathfrak 2 } $. Now, we consider 
\begin{align*}
& \int_{ \bbr^3 } { f_{ \mathfrak 1 } } p_c p_d \frac{ d p }{ { p_{ \mathfrak 1 } }^0 } - \int_{ \bbr^3 } { f_{ \mathfrak 2 } } p_c p_d \frac{ d p }{ { p_{ \mathfrak 2 } }^0 } \\
& = \int_{ \bbr^3 } { { \tilde f }_{ \mathfrak 1 } } e^{ - \frac12 { p_{ \mathfrak 1 } }^0 } p_c p_d \frac{ d p }{ { p_{ \mathfrak 1 } }^0 } - \int_{ \bbr^3 } { { \tilde f }_{ \mathfrak 2 } } e^{ - \frac12 { p_{ \mathfrak 2 } }^0 } p_c p_d \frac{ d p }{ { p_{ \mathfrak 2 } }^0 } \\
& = \int_{ \bbr^3 } ( { \tilde f }_{ \mathfrak 1 } - { \tilde f }_{ \mathfrak 2 } ) e^{ - \frac12 { p_{ \mathfrak 1 } }^0 } p_c p_d \frac{ d p }{ { p_{ \mathfrak 1 } }^0 } + \int_{ \bbr^3 } { \tilde f }_{ \mathfrak 2 } ( e^{ - \frac12 { p_{ \mathfrak 1 } }^0 } - e^{ - \frac12 { p_{ \mathfrak 2 } }^0 } ) p_c p_d \frac{ d p }{ { p_{ \mathfrak 1 } }^0 } + \int_{ \bbr^3 } { \tilde f }_{ \mathfrak 2 } e^{ - \frac12 { p_{ \mathfrak 2 } }^0 } p_c p_d \left( \frac{ 1 }{ { p_{ \mathfrak 1 } }^0 } - \frac{ 1 }{ { p_{ \mathfrak 2 } }^0 } \right) \, d p . 
\end{align*}
The first quantity on the right hand side is estimated the same way as above: 
\begin{align*}
\left| \int_{ \bbr^3 } ( { \tilde f }_{ \mathfrak 1 } - { \tilde f }_{ \mathfrak 2 } ) e^{ - \frac12 { p_{ \mathfrak 1 } }^0 } p_c p_d \frac{ d p }{ { p_{ \mathfrak 1 } }^0 } \right| \leq C \| { \tilde f }_{ \mathfrak 1 } - { \tilde f }_{ \mathfrak 2 } \|_{ 0 , 0 } . 
\end{align*}
For the second and the third quantities, we apply the same arguments as in $ L_4 $ and $ L_2 $, respectively, to obtain the upper bound $ \max_{ a , b } | { g_{ \mathfrak 1 } }^{ a b } - { g_{ \mathfrak 2 } }^{ a b } | $. Hence, we obtain 
\[
| { S_{ \mathfrak 1 } }^{ a b } - { S_{ \mathfrak 2 } }^{ a b } | \leq C \left( \max_{ a , b } | { g_{ \mathfrak 1 } }^{ a b } - { g_{ \mathfrak 2 } }^{ a b } | + \| { \tilde f }_{ \mathfrak 1 } - { \tilde f }_{ \mathfrak 2 } \|_{ 0 , 0 } \right) . 
\]
The other matter terms are similarly estimated. Consequently, we obtain 
\begin{align*}
\frac{ d }{ d t } | { g_{ \mathfrak 1 } }^{ a b } - { g_{ \mathfrak 2 } }^{ a b } | & \leq C \left( \max_{ a , b } | { k_{ \mathfrak 1 } }^{ a b } - { k_{ \mathfrak 2 } }^{ a b } | \right) , \\ 
\frac{ d }{ d t } | { k_1 }^{ a b } - { k_2 }^{ a b } | & \leq C \left( \max_{ a , b } | { g_{ \mathfrak 1 } }^{ a b } - { g_{ \mathfrak 2 } }^{ a b } | + \max_{ a , b } | { k_{ \mathfrak 1 } }^{ a b } - { k_{ \mathfrak 2 } }^{ a b } | + \| { \tilde f }_{ \mathfrak 1 } - { \tilde f }_{ \mathfrak 2 } \|_{ 0 , 0 } + | \phi_{ \mathfrak 1 } - \phi_{ \mathfrak 2 } | \right) . 
\end{align*}
For the scalar field we easily obtain
\begin{align*}
\frac{ d }{ d t } | \phi_{ \mathfrak 1 } - \phi_{ \mathfrak 2 } | & \leq C | \psi_{ \mathfrak 1 } - \psi_{ \mathfrak 2 } | , \\ 
\frac{ d }{ d t } | \psi_1 - \psi_2 | & \leq C \left( \max_{ a , b } | { g_{ \mathfrak 1 } }^{ a b } - { g_{ \mathfrak 2 } }^{ a b } | + \max_{ a , b } | { k_{ \mathfrak 1 } }^{ a b } - { k_{ \mathfrak 2 } }^{ a b } | + | \phi_{ \mathfrak 1 } - \phi_{ \mathfrak 2 } | + | \psi_{ \mathfrak 1 } - \psi_{ \mathfrak 2 } | \right) . 
\end{align*}
Finally, applying Gr{\" o}nwall's inequality to the above four differential inequalities together with \eqref{uniquef}, we obtain the uniqueness of solutions to the EBs system. This completes the proof of Proposition \ref{prop local}.

\section{Global existence and asymptotic behavior}\label{sec Global}

In this part we obtain the global existence of small solutions to the EBs system. The argument will be the standard: we assume that initial data is small and obtain certain decay properties, so that the time interval of Proposition \ref{prop local} will be extended to $ [ t_0 , \infty ) $. 

Recall that the potential function $ V : \bbr \to [ V_0 , \infty ) $ for the scalar field is a smooth function satisfying 
\[
V ( 0 ) = V_0 > 0 , \qquad V' ( 0 ) = 0 , \qquad V'' ( 0 ) > 0 .
\]
Hence, we may write 
\[
V ( \phi ) = V_0 + \frac12 V'' ( 0 ) \phi^2 + { \tilde V } ( \phi ) , 
\]
where $ { \tilde V } $ satisfies 
\begin{align}
| { \tilde V } ( \phi ) | \leq C | \phi |^3 , \qquad | { \tilde V } ' ( \phi ) | \leq C \phi^2 , \label{tilde V}
\end{align}
for small $ \phi $. Below, we will first consider the scalar field and the Hubble variable. We define
\begin{align*}
E & = \frac12 \left( \psi^2 + 3 \gamma \phi \psi + \left( \frac92 + \chi_0 \right) \gamma^2 \phi^2 \right) , \qquad \chi_0 = \frac{ V'' ( 0 ) }{ \gamma^2 } , \qquad \gamma = \sqrt{ \frac{ V_0 }{ 3 } } , \\ 
X & = 3 H^2 - \frac12 \psi^2 - V ( \phi ) , 
\end{align*} 
and
\[
E_0 = E ( t_0 ) , \qquad X_0 = X ( t_0 ) . 
\]
We will show that $ E $ decays, provided that $ E $ is initially small and the Hubble variable $ H $ is close to $ \gamma $. This implies that the scalar fields $ \phi $ and $ \psi $ are also small and decay. The estimates for $ X $ will show that $ H $ decays to $ \gamma $, if it is initially close to $ \gamma $. This will be studied in Section \ref{sec asymptotics1}.

\subsection{Asymptotics for the scalar field and the Hubble variable}\label{sec asymptotics1}

We first observe that the Hubble variable $ H $ is bounded below by $ \gamma $. Let us introduce the shear tensor $ \sigma_{ a b } $ defined by the trace free part of the second fundamental form, i.e., 
\[
k_{ a b } = \sigma_{ a b } + H g_{ a b } , \qquad H = \frac13 k , \qquad k = k_{ a b } g^{ a b } . 
\]
Then, the constraint equation \eqref{EBs6} can be written as 
\begin{align}
\frac12 R - \frac12 \sigma_{ a b } \sigma^{ a b } + 3 H^2 = \rho + \frac12 \psi^2 + V ( \phi ) . \label{constraint} 
\end{align}
We note that the Ricci scalar $ R $ is non-positive, and the potential has the positive lower bound $ V ( \phi ) \geq V_0 $. Hence, we have $ 3 H^2 \geq V_0 $, which is equivalent to 
\begin{align} 
H \geq \gamma , \label{H gamma}
\end{align}
if we assume $ H_0 \geq \gamma $. Moreover, we can see that $ H $ is decreasing. Applying \eqref{constraint} to the evolution equations \eqref{EBs1} and \eqref{EBs2}, we obtain the evolution equation for $ H $: 
\begin{align}
\frac{ d H }{ d t } = \frac16 R - \frac12 \sigma_{ a b } \sigma^{ a b } - \frac12 \rho - \frac12 \psi^2 - \frac16 S , \label{dH dt} 
\end{align}
which shows that 
\begin{align}
\frac{ d H }{ d t } \leq 0 . \label{dH} 
\end{align}
We can also observe that the quantity $ X $ is non-negative. Let us write \eqref{constraint} as 
\[
3 H^2 - \frac12 \psi^2 - V ( \phi ) = \rho - \frac12 R + \frac12 \sigma_{ a b } \sigma^{ a b } , 
\]
and note that the left hand side equals $ X $ and the right hand side is non-negative. Hence, we have 
\begin{align}
X = \rho - \frac12 R + \frac12 \sigma_{ a b } \sigma^{ a b } \geq 0 .  \label{X geq 0} 
\end{align} 
Finally, we notice that $ E $ is equivalent to $ \phi^2 + \psi^2 $. Applying the inequality $ | 3 \gamma \phi \psi | \leq \frac12 ( 9 \gamma^2 \phi^2 + \psi^2 ) $ to $ E $, we obtain 
\begin{align}
\frac12 \left( \frac12 \psi^2 + \chi_0 \gamma^2 \phi^2 \right) \leq E \leq \frac12 \left( \frac32 \psi^2 + ( 9 + \chi_0 ) \gamma^2 \phi^2 \right) . \label{E equiv} 
\end{align}
For more details, we refer to Chapter 27 of Ref.~\cite{Ringstrom}. The following is the result for the scalar field and the Hubble variable.

\begin{lemma}\label{lem H small} 
Suppose that $ H $, $ \phi $ and $ \psi $ satisfy the EBs system on an interval $ [ t_0 , T ] $. Then, there exist positive numbers $ \varepsilon $, $ b_1 $ and $ b_2 $ such that if initial data satisfy 
\begin{align}
H_0 - \gamma + E_0 < \varepsilon , \label{H small} 
\end{align}
then we have 
\begin{align}
& E ( t ) \leq E_0 e^{ - b_1 ( t - t_0 ) } , \label{E decay} \\ 
& X ( t ) \leq X_0 e^{ - 2 \gamma ( t - t_0 ) } , \label{X decay} \\ 
& 0 \leq H ( t ) - \gamma \leq C ( E_0 + X_0 ) e^{ - b_2 t } , \label{H decay} 
\end{align} 
where $ C $ does not depend on $ T $. 
\end{lemma}
\begin{proof}
By direct calculations, we have the following differential inequality for $ E $: 
\begin{align*}
\frac{ d E }{ d t } & = \psi \frac{ d \psi }{ d t } + \frac32 \gamma \psi^2 + \frac32 \gamma \phi \frac{ d \psi }{ d t } + \left( \frac92 + \chi_0 \right) \gamma^2 \phi \psi \\ 
& = \psi \left( - 3 H \psi - V' ( \phi ) \right) + \frac32 \gamma \psi^2 + \frac32 \gamma \phi \left( - 3 H \psi - V' ( \phi ) \right) + \left( \frac92 + \chi_0 \right) \gamma^2 \phi \psi \\ 
& = \left( - 3 H + \frac32 \gamma \right) \psi^2 - \psi V' ( \phi ) + \left( - \frac92 H + \frac92 \gamma \right) \gamma \phi \psi - \frac32 \gamma \phi V' ( \phi ) + \chi_0 \gamma^2 \phi \psi \\ 
& = \left( - 3 H + \frac32 \gamma \right) \psi^2 - \psi \left( V' ( \phi ) - \chi_0 \gamma^2 \phi \right) + \left( - \frac92 H + \frac92 \gamma \right) \gamma \phi \psi - \frac32 \gamma \phi \left( V' ( \phi ) - \chi_0 \gamma^2 \phi \right) - \frac32 \chi_0 \gamma^3 \phi^2 , 
\end{align*}
where we used \eqref{EBs4} for $ d \psi / d t $. Recall that $ H \geq \gamma $ and 
\[
V' ( \phi ) - \chi_0 \gamma^2 \phi = { \tilde V }' ( \phi ) . 
\]
Hence, we obtain by applying \eqref{tilde V}, \eqref{H gamma} and \eqref{dH}, 
\begin{align*}
\frac{ d E }{ d t } & \leq - \frac32 \gamma \psi^2 + | \psi { \tilde V }' ( \phi ) | + \frac94 ( H - \gamma ) ( \gamma^2 \phi^2 + \psi^2 ) + \frac32 | \gamma \phi { \tilde V }' ( \phi ) | - \frac32 \chi_0 \gamma^3 \phi^2 \\ 
& \leq - \frac32 \gamma ( \psi^2 + \chi_0 \gamma^2 \phi^2 ) + \frac94 ( H_0 - \gamma ) ( \gamma^2 \phi^2 + \psi^2 ) + C E^{ \frac32 } , 
\end{align*}
as long as $ \phi $ is small. We apply \eqref{E equiv} to conclude that there exist $ \varepsilon > 0 $ and $ b_1 > 0 $ such that if $ H_0 $ and $ E_0 $ satisfy \eqref{H small}, then 
\[
\frac{ d E }{ d t } \leq - b_1 E , 
\]
which implies the first result \eqref{E decay}. The differential inequality for $ X $ is obtained by applying \eqref{dH dt}, 
\begin{align*}
\frac{ d X }{ d t } & = 6 H \frac{ d H }{ d t } - \psi \frac{ d \psi }{ d t } - V' ( \phi ) \psi \\ 
& = 6 H \left( \frac16 R - \frac12 \sigma_{ a b } \sigma^{ a b } - \frac12 \rho - \frac12 \psi^2 - \frac16 S \right) - \psi \left( - 3 H \psi - V' ( \phi ) \right) - V' ( \phi ) \psi \\ 
& = - 2 H \left( - \frac12 R + \frac32 \sigma_{ a b } \sigma^{ a b } + \frac32 \rho + \frac12 S \right) . 
\end{align*}
We notice that all the terms in the parenthesis above are non-negative. Applying \eqref{H gamma}, we obtain 
\[
\frac{ d X }{ d t } \leq - 2 \gamma X , 
\]
which gives the second result \eqref{X decay}. For the third one, we rewrite \eqref{X decay} as 
\begin{align*}
3 H^2 - V_0 & \leq \frac12 \psi^2 + V ( \phi ) - V_0 + X_0 e^{ - 2 \gamma ( t - t_0 ) } \\ 
& \leq C ( E_0 e^{ - b_1 t } + X_0 e^{ - 2 \gamma t } ) , 
\end{align*}
and we choose $ b_2 = \min \{ b_1 , 2 \gamma \} $. Since $ 3 H^2 - V_0 \geq 6 \gamma ( H - \gamma ) $, we obtain the third result. 
\end{proof}

In Lemma \ref{lem H small}, we did not assume $ X_0 $ is small, but it is small under the assumption \eqref{H small}. By the definition of $ X $ and the decreasing property of $ H $, we obtain 
\begin{align*}
X & = 3 H^2 - \frac12 \psi^2 - V ( \phi ) \\ 
& \leq 3 H_0^2 - V_0 - \frac12 \psi^2 + | V ( \phi ) - V_0 | \\ 
& \leq 3 ( H_0 + \gamma ) ( H_0 - \gamma ) + \frac12 \psi^2 + C \phi^2 \\ 
& \leq C ( H_0 - \gamma ) + C E . 
\end{align*}
Hence, if we assume \eqref{H small}, then $ X_0 $ is also small. The following lemma follows from Lemma \ref{lem H small}.

\begin{lemma}\label{lem phi decay}
Suppose that $ H $, $ \phi $, $ \psi $, $ \rho $, $ R $ and $ \sigma_{ a b } $ satisfy the EBs system on an interval $ [ t_0 , T ] $. Then, there exist positive numbers $ \varepsilon $, $ b_1 $ and $ b_2 $ such that if initial data satisfy 
\begin{align}
H_0 - \gamma + E_0 < \varepsilon , \label{H small2} 
\end{align}
then we have 
\begin{align}
& | \phi ( t ) | + | \psi ( t ) | \leq C \sqrt{ \varepsilon } e^{ - \frac12 b_1 ( t - t_0 ) } , \label{phi decay} \\ 
& \rho ( t ) + | R ( t ) | + \sigma_{ a b } \sigma^{ a b } ( t ) \leq C \varepsilon e^{ - 2 \gamma t } , \label{shear decay} \\ 
& 0 \leq H ( t ) - \gamma \leq C \varepsilon e^{ - b_2 t } , \label{H decay2} 
\end{align}
where $ C $ does not depend on $ T $. 
\end{lemma}
\begin{proof}
The first and the third results, \eqref{phi decay} and \eqref{H decay2}, are obtained by \eqref{E decay} and \eqref{H decay} of Lemma \ref{lem H small}, respectively. The second result \eqref{shear decay} is obtained by \eqref{X decay} of Lemma \ref{lem H small} together with \eqref{X geq 0}. 
\end{proof}

\subsection{Proof of the main theorem}
We are now ready to prove the global existence. For the scalar fields $ \phi $ and $ \psi $, we will use Lemma \ref{lem phi decay}. For the metric $ g^{ a b } $ and the second fundamental form $ k^{ a b } $ with its trace free part $ \sigma^{ a b } $, we will need the following estimates: 
\begin{align}
| Z^2 g^{ a b } | , | Z^{ - 2 } g_{ a b } |  \leq C , \qquad | Z^2 \sigma^{ a b } | , | Z^{ - 2 } \sigma_{ a b } | \leq C \sqrt{ \varepsilon } e^{ - \gamma t } , \label{g bound}
\end{align}
which hold under the assumptions of Lemma \ref{lem phi decay}. For the proof of the estimates in \eqref{g bound}, we refer to Chapter 27 of Ref.~\cite{Ringstrom}. We now prove the main theorem of the paper.

\subsubsection{Global existence}
In Lemma \ref{lem phi decay}, we showed that there exists $ \varepsilon > 0 $ such that if $ H $ and $ E $ are initially small in the sense of \eqref{H small2}, then the scalar fields $ \phi $ and $ \psi $ decay to zero exponentially. This will be enough for the global existence of the scalar fields. Now, let us consider the metric $ g^{ a b } $ and the second fundamental form $ k^{ a b } $. We assume \eqref{H small2} for initial data, and let $ \varepsilon < 1 $ for simplicity. By direct calculations, we have 
\begin{align}
\frac{ d }{ d t } ( Z^2 g^{ a b } ) & = 2 Z \frac{ d Z }{ d t } g^{ a b } - 2 Z^2 k^{ a b } \nonumber \\ 
& = - 2 Z^2 \sigma^{ a b } - 2 ( H - \gamma ) Z^2 g^{ a b } . \label{dZg}
\end{align}
Applying \eqref{H decay2} and \eqref{g bound}, we obtain 
\begin{align}
| Z^2 g^{ a b } - Z^2 ( t_0 ) { g_0 }^{ a b } | \leq C \sqrt{ \varepsilon } , \label{g close}
\end{align}
where $ C $ does not depend on $ t $. Moreover, we have 
\begin{align*}
| Z^2 k^{ a b } - \gamma Z^2 g^{ a b } | & = | Z^2 \sigma^{ a b } + Z^2 H g^{ a b } - \gamma Z^2 g^{ a b } | \nonumber \\ 
& \leq | Z^2 \sigma^{ a b } | + ( H - \gamma ) | Z^2 g^{ a b } | \nonumber \\ 
& \leq C \sqrt{ \varepsilon } , 
\end{align*}
and combine this with \eqref{g close} to obtain 
\begin{align}
| Z^2 k^{ a b } - \gamma Z^2 ( t_0 ) { g_0 }^{ a b } | \leq C \sqrt{ \varepsilon } . \label{k close} 
\end{align}
We notice that \eqref{g close} and \eqref{k close} imply 
\begin{align}
& \max_{ a , b } | Z^2 g^{ a b } | \leq Z^2 ( t_0 ) \max_{ a , b } | { g_0 }^{ a b } | + C \sqrt{ \varepsilon } , \label{Zg bound}  \\ 
& \max_{ a , b } | Z^2 k^{ a b } | \leq \gamma Z^2 ( t_0 ) \max_{ a , b } | { g_0 }^{ a b } | + C \sqrt{ \varepsilon } , \label{Zk bound} 
\end{align}
where the constants $ C $ are independent of $ t $. For the distribution function $ f $, we need to check the conditions \eqref{local 4} and \eqref{local 5}. We first consider 
\begin{align*}
Z^2 g^{ a b } p_a p_b = ( Z^2 g^{ a b } - Z^2 ( t_0 ) { g_0 }^{ a b } ) p_a p_b + Z^2 ( t_0 ) { g_0 }^{ a b } p_a p_b , 
\end{align*}
where the first term on the right hand side can be estimated by using \eqref{g close} as 
\[
| ( Z^2 g^{ a b } - Z^2 ( t_0 ) { g_0 }^{ a b } ) p_a p_b | \leq C \sqrt{ \varepsilon } | p |^2 . 
\]
Since the initial data $ { g_0 }^{ a b } $ is positive definite, we conclude that there exists $ c_1 \geq 1 $, which is independent of $ t $, satisfying 
\begin{align}
\frac{ 1 }{ c_1 } | p |^2 \leq Z^2 g^{ a b } p_a p_b \leq c_1 | p |^2 , \label{g positive} 
\end{align}
for sufficiently small $ \varepsilon > 0 $. Similarly, we have 
\[
Z^2 k^{ a b } p_a p_b = Z^2 \sigma^{ a b } p_a p_b + Z^2 H g^{ a b } p_a p_b , 
\]
where $ | Z^2 \sigma^{ a b } p_a p_b | \leq C \sqrt{ \varepsilon } | p |^2 $, so that we can conclude that there exists, by abusing the notation, $ c_1 \geq 1 $ such that 
\begin{align}
\frac{ 1 }{ c_1 } | p |^2 \leq Z^2 k^{ a b } p_a p_b \leq c_1 | p |^2 . \label{k positive} 
\end{align}
Note that \eqref{Zg bound} implies the first condition in \eqref{local 5}. For the determinant of $ g^{ - 1 } $, we notice that 
\begin{align*}
\frac{ d \det g^{ - 1 } }{ d t } = - 6 H \det g^{ - 1 } , 
\end{align*}
to obtain 
\begin{align*}
\det g^{ - 1 } & = \det g_0^{ - 1 } \exp \left( - 6 \int_{ t_0 }^t H \, d \tau \right) \\ 
& = \det g_0^{ - 1 } \exp \left( - 6 \int_{ t_0 }^t ( H - \gamma ) \, d \tau \right) e^{ - 6 \gamma t } e^{ 6 \gamma t_0 } . 
\end{align*}
Since $ H - \gamma \geq 0 $ is small and integrable, we conclude that 
\begin{align}
\frac12 Z^6 ( t_0 ) \det g_0^{ - 1 } \leq Z^6 \det g^{ - 1 } \leq Z^6 ( t_0 ) \det g_0^{ - 1 } . \label{det positive} 
\end{align}
Now, we consider the following estimate:
\begin{align*}
& \frac{ d }{ d t } \| f ( t ) \|^2_{ g , m , N } \nonumber \\ 
& \leq Z^{ - 1 } \left( \sup_{ p , q } { \mathcal C }_1 \right) \| f ( t ) \|^2_{ g , m , N } + ( \det g )^{ - \frac14 } \left( \sup_{ p , q } { \mathcal C }_2 + \sup_{ p , q } { \mathcal C }_3 \right) \| f ( t ) \|^3_{ g , m , N } , 
\end{align*}
which is the estimate \eqref{iterationB}, where $ f^n $ and $ f^{ n + 1 } $ are replaced by $ f $. The quantities $ \sup_{ p , q } { \mathcal C }_1 $, $ \sup_{ p , q } { \mathcal C }_2 $ and $ \sup_{ p , q } { \mathcal C }_3 $ are bounded by a constant $ C $, since we have \eqref{g positive} and \eqref{k positive} together with \eqref{Zg bound}. Note that $ Z^{ - 1 } $ is integrable on $ [ t_0 , \infty ) $. The determinant is estimated by using \eqref{det positive}: 
\[
( \det g )^{ - \frac14 } \leq C Z^{ - \frac32 } \leq C e^{ - \frac32 \gamma t } , 
\]
which is also integrable on $ [ t_0 , \infty ) $. Hence, we obtain by Gr{\" o}nwall's lemma, 
\begin{align}
\| f ( t ) \|^2_{ g , m , N } \leq C \| f_0 \|^2_{ g , m , N } , \label{f bound} 
\end{align}
for sufficiently small initial data.

To summarize, let $ { g_0 }^{ a b } $, $ { k_0 }^{ a b } $, $ \phi_0 $, $ \psi_0 $ and $ f_0 $ be a set of initial data of the EBs system \eqref{EBs1}--\eqref{EBs7}. Suppose that initial data satisfy 
\[
H_0 - \gamma + E_0 < \varepsilon , \qquad \max_{ a , b } | { g_0 }^{ a b } | \leq C_0 , \qquad \| f_0 \|^2_{ g_0 , m + \frac12 , N } < \varepsilon , 
\]
and $ { g_0 }^{ a b } $ and $ { k_0 }^{ a b } $ are positive definite. Then, by the equivalence \eqref{E equiv}, we have 
\[
| \phi_0 | + | \psi_0 | < C \sqrt{ \varepsilon } . 
\]
For the initial value of the second fundamental form, we use \eqref{g bound} to obtain 
\begin{align*}
| { k_0 }^{ a b } | & = | { \sigma_0 }^{ a b } + H_0 { g_0 }^{ a b } - \gamma { g_0 }^{ a b } + \gamma { g_0 }^{ a b } | \leq C \sqrt{ \varepsilon } + C_0 ( \varepsilon + \gamma ) , 
\end{align*}
which implies that 
\[
\max_{ a , b } | { g_0 }^{ a b } | + \max_{ a , b } | { k_0 }^{ a b } | \leq C \sqrt{ \varepsilon } + C_0 ( 1 + \varepsilon + \gamma ) . 
\]
Hence, we can find $ C_1 \geq 1 $ such that initial data satisfy 
\begin{align*} 
\max_{ a , b } | { g_0 }^{ a b } | + \max_{ a , b } | { k_0 }^{ a b } | \leq C_1 \sqrt{ \varepsilon } + C_0 ( 1 + \varepsilon + \gamma ) , \qquad | \phi_0 | + | \psi_0 | \leq C_1 \sqrt{ \varepsilon } , \qquad \| f_0 \|^2_{ g_0 , m + \frac12 , N } \leq C_1 \varepsilon . 
\end{align*}
Then, by Proposition \ref{prop local}, we obtain a solution on an interval $ [ t_0 , T ] $, where we have 
\begin{align}
\sup_{ t_0 \leq t \leq T } \max_{ a , b } | { g }^{ a b } ( t ) | + \sup_{ t_0 \leq t \leq T } \max_{ a , b } | { k }^{ a b } ( t ) | \leq 2 C_1 \sqrt{ \varepsilon } + 2 C_0 ( 1 + \varepsilon + \gamma ) , \label{g local} \\ 
\sup_{ t_0 \leq t \leq T } | \phi ( t ) | + \sup_{ t_0 \leq t \leq T } | \psi ( t ) | \leq 2 C_1 \sqrt{ \varepsilon } , \label{phi local} \\ 
\sup_{ t_0 \leq t \leq T } \| f ( t ) \|^2_{ g , m , N } \leq 2 C_1 \varepsilon . \label{f local} 
\end{align}
On the other hand, the estiamtes \eqref{phi decay}, \eqref{Zg bound}, \eqref{Zk bound} and \eqref{f bound} show that there exists $ \varepsilon > 0 $ such that the solutions satisfy the following estimates on $ [ t_0 , T ] $: 
\begin{align*}
& | \phi ( t ) | \leq C_1 \sqrt{ \varepsilon } e^{ - \frac12 b_1 ( t - t_0 ) } , \\ 
& | \psi ( t ) | \leq C_1 \sqrt{ \varepsilon } e^{ - \frac12 b_1 ( t - t_0 ) } , \\ 
& \max_{ a , b } | g^{ a b } ( t ) | \leq C_0 + C_1 \sqrt{ \varepsilon } , \\ 
& \max_{ a , b } | k^{ a b } ( t ) | \leq C_0 \gamma + C_1 \sqrt{ \varepsilon } , \\ 
& \| f ( t ) \|^2_{ g , m , N } \leq C_1 \varepsilon , 
\end{align*}
by making $ C_1 $ larger if necessary. This implies that the estimates \eqref{g local}--\eqref{f local} hold on $ [ t_0 , \infty ) $, and we obtain the global existence for the EBs system.

\subsubsection{Asymptotic behavior}
Note that we have already obtained the following asymptotics: 
\begin{align*}
| \phi ( t ) | + | \psi ( t ) | & \leq C \sqrt{ \varepsilon } e^{ - \frac12 b_1 t } , \\ 
\rho ( t ) + | R ( t ) | + \sigma_{ a b } \sigma^{ a b } ( t ) & \leq C \varepsilon e^{ - 2 \gamma t } , \\ 
0 \leq H ( t ) - \gamma & \leq C \varepsilon e^{ - b_2 t } , \\ 
\det g^{ - 1 } & \leq C e^{ - 6 \gamma t } . 
\end{align*}
Integrating \eqref{dZg} over $ [ t_0 , t ] $, we have 
\[
Z^2 ( t ) g^{ a b } ( t ) = Z^2 ( t_0 ) { g_0 }^{ a b } - \int_{ t_0 }^t 2 Z^2 \sigma^{ a b } + 2 ( H - \gamma ) Z^2 g^{ a b } \, d \tau , 
\]
and define 
\[
{ g_\infty }^{ a b } = Z^2 ( t_0 ) { g_0 }^{ a b } - \int_{ t_0 }^\infty 2 Z^2 \sigma^{ a b } + 2 ( H - \gamma ) Z^2 g^{ a b } \, d \tau . 
\]
Using \eqref{H decay2} and \eqref{g bound}, we obtain 
\begin{align*}
| Z^2 ( t ) g^{ a b } ( t ) - { g_\infty }^{ a b } | & \leq 2 \int_t^\infty | Z^2 \sigma^{ a b } | + | ( H - \gamma ) Z^2 g^{ a b } | \, d \tau \\ 
& \leq C \sqrt{ \varepsilon } e^{ - \gamma t } + C \varepsilon e^{ - b_2 t } , 
\end{align*}
and we choose $ b_3 = \min \{ \gamma , b_2 \} $. This proves the asymptotic behavior of the metric components. Similarly, we obtain 
\[
| Z^{ - 2 } ( t ) g_{ a b } ( t ) - { g_\infty }_{ a b } | \leq C \sqrt{ \varepsilon } e^{ - b_3 t } . 
\]
This completes the proof of Theorem \ref{thm main}.

\section*{Acknowledgments}
H. Lee was supported by the Basic Science Research Program through the National Research Foundation of Korea (NRF) funded by the Ministry of Science, ICT and Future Planning (NRF- 2018R1A1A1A05078275). E. Nungesser has been supported by Grant MTM2017-85934-C3-3-P of Agencia Estatal de Investigaci\'{o}n (Spain).

\end{document}